\documentclass[11pt]{article}

\usepackage[margin=1in]{geometry}
\usepackage{amsmath,amssymb,amsthm,mathtools}
\usepackage{enumitem}
\usepackage{array,booktabs}
\usepackage{microtype}
\usepackage[table]{xcolor}
\usepackage{aliascnt}
\usepackage[
  colorlinks=true,
  linkcolor=blue,
  citecolor=blue,
  urlcolor=blue,
  linkbordercolor=blue,
  citebordercolor=blue,
  urlbordercolor=blue
]{hyperref}
\AtBeginDocument{\hypersetup{pdfborder={0 0 0.7}}}
\usepackage[nameinlink,capitalize,noabbrev]{cleveref}
\usepackage{lmodern}

\AddToHook{cmd/thebibliography/after}{\addcontentsline{toc}{section}{\refname}}

\newtheorem{theorem}{Theorem}[section]
\newaliascnt{proposition}{theorem}
\newtheorem{proposition}[proposition]{Proposition}
\aliascntresetthe{proposition}
\newaliascnt{lemma}{theorem}
\newtheorem{lemma}[lemma]{Lemma}
\aliascntresetthe{lemma}
\newaliascnt{corollary}{theorem}
\newtheorem{corollary}[corollary]{Corollary}
\aliascntresetthe{corollary}
\theoremstyle{definition}
\newaliascnt{definition}{theorem}

\aliascntresetthe{definition}
\theoremstyle{remark}
\newaliascnt{remark}{theorem}

\aliascntresetthe{remark}
\crefname{proposition}{proposition}{propositions}
\crefname{lemma}{lemma}{lemmas}
\crefname{corollary}{corollary}{corollaries}
\crefname{definition}{definition}{definitions}
\crefname{remark}{remark}{remarks}

\newcommand{\C}{\mathbb C}
\newcommand{\R}{\mathbb R}
\DeclareMathOperator{\E}{\mathbb E}
\DeclareMathOperator{\Prb}{\mathbb P}
\DeclareMathOperator{\Tr}{Tr}
\DeclareMathOperator{\rank}{rank}
\newcommand{\Id}{\mathrm I}
\newcommand{\ii}{\mathrm i}
\newcommand{\ee}{\mathrm e}
\newcommand{\dd}{\mathop{}\!\mathrm d}
\newcommand{\Ddir}{\mathrm D}
\newcommand{\F}{\mathrm F}
\newcommand{\vct}[1]{\boldsymbol{#1}}
\newcommand{\cD}{\mathcal D}
\newcommand{\cH}{\mathcal H}
\newcommand{\cI}{\mathcal I}
\newcommand{\cK}{\mathcal K}
\newcommand{\ip}[2]{\langle #1,#2\rangle}
\newcommand{\norm}[1]{\left\lVert #1\right\rVert}
\newcommand{\abs}[1]{\left\lvert #1\right\rvert}

\newcommand{\eps}{\varepsilon}
\newcommand{\frobenius}{\F}

\newcommand{\complex}{{\mathbb C}}

\DeclareMathOperator{\expct}{\E}

\newcommand{\id}{\Id}

\newcommand{\eqdef}{\coloneqq}
\newcommand{\tensor}{\otimes}

\newcommand{\suppress}[1]{}

\newcommand{\tH}{\widetilde{H}}

\title{Optimal Low-Rank Quantum State Tomography with\\
Bounded-Sample Joint Measurements}
  \author{Ashwin Nayak\thanks{Department
  of Combinatorics and Optimization, and Institute for Quantum Computing, University of Waterloo.
  Email: \texttt{ashwin.nayak@uwaterloo.ca}.}
  \and
  Xingyu Zhou\thanks{Department of Computer Science, University of British Columbia.
  Email: \texttt{zxingyu@cs.ubc.ca}.}}
\date{}

\begin{document}
\maketitle

\begin{abstract}
We determine the optimal sample complexity of low-rank quantum state tomography when each measurement may act jointly on at most $t$ samples.  For sufficiently small $\varepsilon$, estimating an unknown state on $\mathbb{C}^d$ of rank
at most $r$ to trace norm error $\varepsilon$ with constant success probability requires, and is achievable with,
$$\Theta\left(\frac{dr}{\varepsilon^2}\mathop{\mathrm{max}}\left\{1,\frac{r}{\sqrt{t}}\right\}\right)$$
samples.  The lower bound allows the protocol to choose each joint measurement adaptively using all previous classical outcomes; the matching upper bound is nonadaptive.  Thus joint measurements on at most $t$ samples improve the complexity of algorithms making single-sample measurements by at most a factor $\sqrt{t}$. Further, measuring order $r^2$ samples jointly is necessary and sufficient to attain the unrestricted collective rate.

For the lower bound, we vary the support of a state with fixed uniform spectrum and bound the Fisher information trace of every joint measurement on $t$ samples.  The adaptive Fisher chain rule and the van Trees inequality then give the trace norm lower bound.  For the upper bound, we construct and analyze a nonadaptive tomography protocol based on a Gaussian joint measurement.  An explicit second moment identity and a conditional Gaussian law outside the state's support give a rank-dependent error analysis, yielding the matching rate.
\end{abstract}

\clearpage
\setcounter{tocdepth}{2}
\tableofcontents
\clearpage

\section{Introduction}
\label{sec:introduction}

Quantum state tomography reconstructs a classical description of an unknown
state~\(\rho\), or more precisely of an approximation to~$\rho$, from measurements on independent samples of the state.  A state on
\(\C^d\) is represented by a positive semidefinite matrix of trace one. For a
\(q\)-qubit system, \(d=2^q\), so an unrestricted state has order \(d^2\)
real parameters. Intuitively, we would therefore expect a tomography algorithm to require of the order of~$d^2$ samples for a sufficiently accurate approximation. Perhaps surprisingly, the optimal sample complexity of the problem was established only a decade ago, by O'Donnell and Wright~\cite{ODonnellWright2016} and Haah, Harrow, Ji, Wu, and Yu~\cite{HaahHarrowJiWuYu2017}. These works proved that~$\Theta(d^2/ \varepsilon^2)$ samples of~$\rho$ are necessary and sufficient to learn the state to within~$\varepsilon$ in trace distance.  

Pure quantum states, central targets of state preparation, have rank
\(r=1\) even as the ambient dimension \(d\) grows exponentially with
the number of qubits.  They can be specified with only order \(d\) real parameters, and an early result due to Hayashi~\cite{Hayashi98-pure-state-tomography} shows how they may be learnt with correspondingly fewer samples, namely~$O(d)$ samples. Simpler and more efficient algorithms have been discovered since (see, e.g., Ref.~\cite{GutaKahnKuengTropp2020}). A natural question is whether mixed states with larger but bounded rank can similarly be learnt with fewer than~$\Theta(d^2)$ samples. In addition to being of theoretical interest, such states are also observed in practice when preparation noise leaves most of the spectral weight on a few eigenvectors. 
For example, Gross, Liu, Flammia, Becker, and Eisert~\cite{GrossLiuFlammiaBeckerEisert2010} discuss a reconstructed state from an experiment with eight ions, with~\(d=256\) and~\(99\%\)
of its spectral weight on just~\(11\) eigenvectors.
Motivated by these considerations, low rank tomography and its variants have been studied extensively. The optimal sample complexity for low rank tomography has again been shown to match the number of parameters required to specify the state. More precisely, $\Theta(dr/\varepsilon^2)$ samples of a state~$\rho$ with rank at most~$r$ are necessary and sufficient for reconstructing an~$\varepsilon$ approximation in trace distance~\cite{ODonnellWright2016,HaahHarrowJiWuYu2017,ScharnhorstSpileckiWright2025}. 

In this work, we return to the problem of learning states with bounded rank. While the sample-optimal algorithms reduce the number of samples needed to the intrinsic dimension of the state, they do so by employing a joint measurement on all the samples used. Joint measurements are, in a sense, unavoidable: algorithms that measure every sample separately in a nonadaptive fashion require~$\Omega(dr^2/\eps^2)$ samples~\cite{FlammiaGrossLiuEisert2012,HaahHarrowJiWuYu2017,LoweNayak2025}. This sample complexity is also known to be optimal (see Ref.~\cite{KuengRauhutTerstiege2017} with details provided by Ref.~\cite[Sec.~II.A]{HaahHarrowJiWuYu2017} and Ref.~\cite[Sec.~B.2]{LoweNayak2025} for one algorithm; and Ref.~\cite{GutaKahnKuengTropp2020} for a second algorithm). Thus restricting the measurements to one sample at a time can cost a factor of order~\(r\) in sample complexity.

The above two extremes of measurement strategies also demand different experimental resources.  A protocol with joint measurements may require many samples to be stored and processed coherently at once, whereas a protocol measuring individual samples requires no coherent storage or manipulation across samples. Joint measurements of a large number of samples may not be feasible from a practical point of view. Moreover, we may only have access to one or at best a few quantum devices that prepare the state, and the storage of states over time may be difficult. Even assuming we have access to a suitable system that can be prepared with sufficiently many samples, joint measurements may be computationally more demanding in time complexity. Finally, only a limited set of measurements may be available in an experimental set-up. This motivates the consideration of learning models that incorporate such limitations, and especially models in which a bounded number of samples are measured jointly. Measurements in which each sample is measured separately have been called by several different names in the literature: \emph{single-copy\/}, \emph{unentangled\/}, \emph{incoherent\/}, or \emph{independent\/}. Similarly, joint measurements have also been called \emph{entangled\/}, \emph{coherent\/}, or \emph{with quantum memory\/}. We avoid these other terms, as they conflict with concepts that predate them.

Joint measurements on a bounded number of samples interpolate between
the extremes: each measurement acts on at most~\(t\) fresh samples, after which only
classical information is retained by the algorithm. A further degree of freedom that is available with such measurements is adaptivity --- the ability to tune measurements based on the outcomes of the previous measurements. The parameter~\(t\) thus limits joint quantum processing while
allowing classical adaptivity across measurements.

Until recently, little was known about sample complexity of state tomography with adaptive bounded-sample joint measurements. Lowe and Nayak~\cite{LoweNayak2025} lifted the~$\Omega(d^3/\eps^2)$ lower bound for states of arbitrary unknown rank to adaptive single-sample measurements, when the measurements are restricted to a fixed set such as the set of efficiently implementable measurements (i.e., those with polynomial-size circuits). Chen, Huang, Li, Liu, and Sellke~\cite{ChenHuangLiLiuSellke2023} subsequently proved an~$\Omega(d^3/\eps^2)$ unconditional lower bound for any adaptive single-sample measurement strategy (with a finite number of outcomes). They also demonstrated the power of adaptivity when approximating the state with respect to infidelity. Both results address the case when no guarantee on the rank of the state is given.

Chen, Li, and Liu~\cite{ChenLiLiu2024} established an almost~\(\sqrt t\) factor improvement over single-sample tomography with adaptive~$t$-sample joint measurements for states of arbitrary rank (i.e., without a promise on the rank), for~\(t\le\min\{d^2,(\sqrt d/\eps)^{0.2}\}\). Their lower bound assumes \(t\le\eps^{-0.1}\), sufficiently small \(\eps\), and sufficiently large \(d\). The upper bound was subsequently improved by Pelecanos, Spilecki, and Wright~\cite{PelecanosSpileckiWright2025}, matching the lower bound due to Chen, Li, and Liu within this restricted range. For low-rank states, however, the optimal sample complexity for adaptive strategies remained open, even in the case of single-sample measurements. It was not clear whether the techniques from prior works could be strengthened to derive optimal rank-sensitive bounds.

We determine the entire rank-dependent interpolation between single-sample and the fully joint measurement strategies, up to universal constants. The lower bound allows for arbitrary classical adaptivity, while the upper bound is achieved by a matching nonadaptive algorithm.

\begin{theorem}[Adaptive tomography lower bound]
  \label{thm:adaptive-block-lower-bound}
  There are universal constants \(c,\eps_0>0\) such that the following
  holds. Let \(d\ge2\), \(1\le r\le d\), and \(t\ge1\) be integers, and let
  \(0<\eps\le\eps_0\). Suppose an adaptive protocol uses a total of \(n\) samples of an arbitrary unknown state \(\rho\) on \(\C^d\) of rank at most \(r\), makes joint
  measurements on at most \(t\) samples at a time, and outputs a state \(\widehat\rho\) such that
  \[
    \Prb_\rho\left[
      \norm{\widehat\rho-\rho}_1\le\eps
    \right]\ge\frac23.
  \]
  Then
  \begin{equation}
    n\ge
    c\,\frac{dr}{\eps^2}
    \max\left\{1,\frac{r}{\sqrt t}\right\}.
    \label{eq:main-adaptive-block-lower-bound}
  \end{equation}
\end{theorem}

Formally, the theorem applies to the following model of adaptive protocols. The protocol proceeds in rounds.  In each
round it chooses a joint positive operator-valued measure (POVM) on at most
\(t\) fresh samples as a function of private randomness and the classical outcomes of all earlier
measurements.  After the measurement, it retains only classical
information. The number of samples measured may vary between rounds, and we
allow adaptive stopping, provided that the bound of~$t$ on the number of samples measured jointly holds in every round.  We allow any countably generated measurable outcome space,
including finite, countable, and Euclidean outcome spaces.
The complete model is given in
\Cref{sec:block-adaptive-model}. 
A nonadaptive protocol is the special case in which the sequence of measurements is fixed in advance. In particular, they are independent of the measurement outcomes observed during an execution of the protocol.

We present a nonadaptive protocol with the following guarantee.

\begin{theorem}[Nonadaptive tomography upper bound]
  \label{thm:bounded-block-upper-bound}
  There is a universal constant \(C>0\) such that the following holds.  Let
  \(d\ge2\), \(1\le r\le d\), and \(t\ge1\) be integers, and let
  \(0<\eps\le1\).  There is a nonadaptive protocol using joint measurements
  on at most \(t\) samples at a time which, for every state \(\rho\) on
  \(\C^d\) of rank at most \(r\), outputs a state \(\widehat\rho\) satisfying
  \[
    \Prb_\rho\left[
      \norm{\widehat\rho-\rho}_1\le\eps
    \right]\ge\frac23
  \]
  and uses at most
  \begin{equation}
    C\,\frac{dr}{\eps^2}
    \max\left\{1,\frac r{\sqrt t}\right\}
    \label{eq:main-bounded-block-upper-bound}
  \end{equation}
  samples.
\end{theorem}

Together, the two theorems characterize the sample complexity, up to universal
constants, for all sufficiently small \(\eps\).  Equivalently, the optimal
rate is
\[
    \begin{cases}
      \Theta \! \left( dr^2/(\eps^2\sqrt t\, ) \right), & 1\le t\le r^2, \qquad \text{and} \\
      \Theta \! \left( dr/\eps^2 \right), &t\ge r^2.
    \end{cases}
\]
In particular, for \(t=1\), arbitrary classical adaptivity does not improve the optimal
nonadaptive single-sample rate.  Increasing the number of samples that may be
measured jointly to \(t\) improves this rate by a factor of order \(\sqrt t\)
until the optimal complexity of~$O(dr/\eps^2)$ is reached.  Measuring order \(r^2\) samples
jointly is therefore necessary and sufficient to attain the optimal sample complexity for unrestricted algorithms.

To our knowledge, this is the first matching rank-dependent characterization
for arbitrary adaptive joint measurements on at most \(t\) samples.  At
\(t=1\), the lower bound resolves the rank-dependent adaptive question posed
by Chen, Huang, Li, Liu, and Sellke
\cite{ChenHuangLiLiuSellke2023} and addresses the question of rank dependence
in adaptive tomography raised by Lowe and Nayak \cite{LoweNayak2025}.
For intermediate \(r\) and \(t\), the two bounds resolve the rank-dependent
interpolation left open by Chen, Li, and Liu~\cite{ChenLiLiu2024}.  At \(r=d\), our upper bound recovers the known optimal tradeoff due to Pelecanos, Spilecki, and Wright~\cite{PelecanosSpileckiWright2025}, while our lower bound removes the earlier restriction relating~\(t\) and~\(\eps\) in Ref.~\cite{ChenLiLiu2024}.

The proof of the lower bound in \cref{thm:adaptive-block-lower-bound} controls the Fisher information trace of every joint
POVM on \(t\) samples of a state with fixed uniform spectrum and varying
support.  A conditional score chain rule extends the bound to arbitrary
adaptive transcripts.  The van Trees inequality, also known as the
Bayesian Cram\'er--Rao inequality, then converts it into a trace norm lower
bound by localizing the support parameter in operator norm.

The algorithm we present averages independent unbiased matrix estimates obtained
from a Gaussian joint measurement. An explicit second moment identity
controls the error on the support of \(\rho\) and between the support and
its orthogonal complement in Frobenius norm.  Outside the support, we
exploit the measurement's Gaussian structure to bound the error in operator
norm.  Projecting the average onto the set of density matrices of rank at
most \(r\) then gives the output of the algorithm.

Both proofs circumvent the explicit use of representation theory; we provide an in-depth overview in \cref{sec:proof-overview}.

\subsection*{Related work}
\label{sec:related-work}

\Cref{tab:prior-bound-landscape} summarizes the best known previous and concurrent
sample complexity bounds by rank \(r\) and the number \(t\) of samples that
may be measured jointly. In the adaptive model we study, for states of rank at most \(r\) and sufficiently small \(\eps\), the results in
\Cref{thm:adaptive-block-lower-bound,thm:bounded-block-upper-bound} give the
optimal sample complexity up to universal constants throughout the
displayed regimes of \(r\) and \(t\).

\begin{table}[htb]
  \centering
  \caption{Best known previous and concurrent bounds for tomography with constant success probability
  and sufficiently small trace norm error \(\eps\).
  U and L denote upper and lower bounds;
  a dash means no separate bound is listed here.  The rightmost column ``Arbitrary measurements'' (without a restriction on~$t$) pertains to measurements which
  may act jointly on all available samples. Works that prove bounds that are within a polylogarithmic
factor of the optimal have also been included.}
  \label{tab:prior-bound-landscape}
  \medskip
  \small
  \begin{tabular}{@{}>{\centering\arraybackslash}p{0.19\textwidth}
      >{\raggedright\arraybackslash}p{0.23\textwidth}
      >{\raggedright\arraybackslash}p{0.28\textwidth}
      >{\raggedright\arraybackslash}p{0.22\textwidth}@{}}
    \toprule
      & \multicolumn{3}{c}{Maximum number \(t\) of samples measured jointly} \\
    \cmidrule(lr){2-4}
    Rank \(r\) & \(t=1\) & \(1<t\le r^2\) & Arbitrary measurements \\
    \midrule
    \(1\le r \ll d\)
      & U: \(O(dr^2/\eps^2)\) \cite{KuengRauhutTerstiege2017}\textsuperscript{\#\$}, \cite{GutaKahnKuengTropp2020}\(^{\#}\)
        \par L: \(\Omega(dr^2/\eps^2)\) \cite{HaahHarrowJiWuYu2017,LoweNayak2025}\(^{\ast}\)
      & U: \textemdash
        \par L: \textemdash
      & U: \(O(dr/\eps^2)\) \cite{ODonnellWright2016, HaahHarrowJiWuYu2017}
        \par L: \(\Omega(dr/\eps^2)\) \cite{HaahHarrowJiWuYu2017,ScharnhorstSpileckiWright2025} \\
    \midrule
    \(r=d\)
    \par (no rank promise)
      & U: \(O(d^3/\eps^2)\) \cite{KuengRauhutTerstiege2017}\textsuperscript{\#\$}, \cite{GutaKahnKuengTropp2020}\(^{\#}\)
        \par L: \(\Omega(d^3/\eps^2)\)
        \cite{HaahHarrowJiWuYu2017}\(^{\ast}\),
        \cite{ChenHuangLiLiuSellke2023}\(^{\dagger}\),
        \cite{KeskinLuoMajidRadzihovsky2026}\textsuperscript{\S}
      & U: \(O(d^3/(\eps^2\sqrt t))\)
        \par\cite{ChenLiLiu2024}\(^{\ddagger}\), \cite{PelecanosSpileckiWright2025,PelecanosSpileckiTangWright2025}\(^{\#}\)
        \par L: \(\Omega(d^3/(\eps^2\sqrt t))\) \cite{ChenLiLiu2024}\(^{\ddagger}\),
        \cite{KeskinLuoMajidRadzihovsky2026}\textsuperscript{\S}
      & U: \(O(d^2/\eps^2)\) \cite{ODonnellWright2016, HaahHarrowJiWuYu2017}
        \par L: \(\Omega(d^2/\eps^2)\) \cite{HaahHarrowJiWuYu2017} \\
    \bottomrule
  \end{tabular}
\medskip
  \par\raggedright
  \(^{\#}\) Nonadaptive algorithms. 
  \par  
  \textsuperscript{\$} See the introduction for the details of the algorithm in Ref.~\cite{KuengRauhutTerstiege2017}.
  \par 
  \(^{\ast}\) \cite{HaahHarrowJiWuYu2017,LoweNayak2025} apply to nonadaptive measurements.
  \par
  \(^{\dagger}\) \cite{ChenHuangLiLiuSellke2023} assumes measurements with finitely many outcomes, but addresses adaptive algorithms.
  \par
  \(^{\ddagger}\) \cite{ChenLiLiu2024} assumes \(t\le\min\{d^2,(\sqrt d/\eps)^{0.2}\}\) for its upper bound and
  \(t<\eps^{-0.1}\) for its lower bound; both allow adaptive algorithms.
  \par
  \textsuperscript{\S} \cite{KeskinLuoMajidRadzihovsky2026} was developed independently and concurrently with this work, and addresses adaptive algorithms.
\end{table}
We briefly elaborate on the related work below.

\paragraph{Low-rank tomography.}
Compressed sensing methods have been used to reconstruct low-rank states from estimates of a
small subset of Pauli expectation values
\cite{GrossLiuFlammiaBeckerEisert2010,FlammiaGrossLiuEisert2012}.
Other approaches include spectral thresholding
\cite{ButuceaGutaKypraios2015}, recovery from random rank-one measurements
\cite{KuengRauhutTerstiege2017}, and projected least squares
\cite{GutaKahnKuengTropp2020}.
For single-sample measurements, a tomography algorithm may be derived from the work of Kueng, Rauhut, and Terstiege~\cite{KuengRauhutTerstiege2017}; see~Refs.~\cite{HaahHarrowJiWuYu2017,LoweNayak2025}. The algorithm uses \(O(dr^2/\eps^2)\) samples and has constant success probability.
In the same setting, Gu{\c{t}}{\u{a}}, Kahn, Kueng, and Tropp~\cite{GutaKahnKuengTropp2020}
achieve the same upper bound via projected least squares.
Improving on prior work~\cite{HaahHarrowJiWuYu2017}, Lowe and Nayak~\cite{LoweNayak2025} proved the optimality of these algorithms. They showed an \(\Omega(dr^2/\eps^2)\) lower bound for
nonadaptive single-sample measurements with finitely many outcomes, for
states of exact rank~$r$ with \(1\le r\le d/3\) and \(0<\eps<1/8\). The restriction of~$r \le d/3$ can be removed using a standard technique.

For unrestricted joint measurements, O'Donnell and Wright~\cite{ODonnellWright2016} obtained
an \(O(dr/\eps^2)\) upper bound. In independent work, Haah, Harrow, Ji, Wu, and Yu~\cite{HaahHarrowJiWuYu2017} obtained an upper bound within a logarithmic factor of the same expression. Haah et al.\ also proved a lower bound matching this rate up to a logarithmic factor.
Scharnhorst, Spilecki, and Wright~\cite{ScharnhorstSpileckiWright2025} subsequently closed the logarithmic gap by proving the lower bound of \(\Omega(dr/\eps^2)\), for \(d>1\) and sufficiently small \(\eps\).

\paragraph{Adaptivity and bounded-sample joint measurements.}
When there is no promise on the rank, Lowe and Nayak~\cite{LoweNayak2025} proved that~$\Omega(d^3/\eps^2)$  samples of the state are required for tomography with adaptive single-sample
measurements, when the measurements are restricted to certain fixed sets. In particular, the bound applies to the set of efficiently implementable measurements (i.e., those with polynomial-size circuits). They asked how rank dependence could be incorporated into such a lower bound.
Chen, Huang, Li, Liu, and Sellke~\cite{ChenHuangLiLiuSellke2023} proved that
adaptive single-sample measurements with finitely many outcomes unconditionally require
\(\Omega(d^3/\eps^2)\) samples for tomography in trace norm, for sufficiently small \(\eps\)
and sufficiently large \(d\). They explicitly left the rank-dependent lower bound open.

For joint measurements on at most \(t\) samples, Chen, Li, and Liu~
\cite{ChenLiLiu2024} gave an
upper bound of~\(\widetilde{O}(d^3/(\sqrt t\,\eps^2))\) for
\(t\le\min\{d^2,(\sqrt d/\eps)^{0.2}\}\) with constant success probability.
They also proved a lower bound of \(\Omega(d^3/(\sqrt t\eps^2))\)
against adaptive protocols when \(\eps\) is sufficiently small,
\(d\) is sufficiently large, and \(t\le\eps^{-0.1}\).
Their open questions include the complexity of rank-dependent tomography, even for \(t=1\).

Pelecanos, Spilecki, and Wright~\cite{PelecanosSpileckiWright2025} subsequently gave the upper bound
\[
  O\left(\max\left\{\frac{d^3}{\sqrt t\eps^2},
    \frac{d^2}{\eps^2}\right\}\right)
\]
using independent applications of the debiased Keyl algorithm.
Pelecanos, Spilecki, Tang, and Wright~\cite{PelecanosSpileckiTangWright2025} obtained the same bound through their reduction from mixed state tomography to pure state tomography, while also guaranteeing time-efficient measurement and other desirable features.
Both analyses use unbiased estimators and second moment bounds.

In independent and concurrent work, Keskin, Luo, Majid, and Radzihovsky~\cite{KeskinLuoMajidRadzihovsky2026} have reported the optimal lower bound for
adaptive tomography without a rank promise, for every \(t\ge1\) and
sufficiently small \(\eps\).  Their bound matches the \(r=d\) case of the
lower bound we derive in \cref{thm:adaptive-block-lower-bound}.

\paragraph{Statistical methods.}
Gill and Levit~ \cite{GillLevit1995} give a general treatment of the van Trees inequality and its
statistical applications.
Butucea, Gu{\c{t}}{\u{a}}, and Kypraios~\cite{ButuceaGutaKypraios2015} used support rotations, Fisher
information, and van Trees to prove asymptotic Frobenius risk lower bounds
for a fixed Pauli measurement design.
Gill and Massar~\cite{GillMassar2000} established Fisher information tradeoffs for
multiparameter quantum estimation, and Zhou and
Chen~\cite{ZhouChen2026} proved a rank-dependent Fisher information bound for support rotations.
In classical estimation under communication constraints, Barnes, Han, and
\"Ozg\"ur~\cite{BarnesHanOzgur2020} combined a conditional Fisher information chain rule with
van Trees for sequential protocols.

\subsection*{Organization}

The remainder of the paper is organized as follows.  We give a proof overview
in \Cref{sec:proof-overview}, collect notation and the adaptive measurement model in
\Cref{sec:preliminaries}, prove the lower and upper bounds in
\Cref{sec:lower-bound-proof,sec:upper-bound-proof}, and conclude with further
questions in \Cref{sec:discussion}.

\section{Proof overview}
\label{sec:proof-overview}

\subsection{Lower bound}
\label{sec:lower-bound-overview}

The proof is organized around the van Trees inequality, also known as the
Bayesian Cram\'er--Rao inequality.  We first describe the argument for rank
\(r\le d/2\); a depolarizing channel reduction at the end handles larger
ranks.

\subsubsection{The statistical core: the van Trees inequality}

Consider a \(p\)-dimensional parameter \(\vct{\theta}\), drawn from a
probability distribution with smooth density \(\pi\), and let
\(\cI_{\mathrm{tr}}(\vct{\theta})\) be the Fisher information matrix of the
complete measurement transcript.  For now, one may think of
\(p=\Theta(dr)\): in the hard family below, these parameters describe
spectrum-preserving rotations of an \(r\)-dimensional support subspace of
\(\C^d\).
Under suitable regularity and boundary conditions, the van Trees inequality
states that every estimator \(T\) satisfies
\[
  \E\norm{T-\vct{\theta}}_2^2
  \ge
  \frac{p^2}{
    \E_{\vct{\theta}}
      \Tr\bigl(\cI_{\mathrm{tr}}(\vct{\theta})\bigr)
    +I(\pi)},
\]
where \(I(\pi)\) is the Fisher information of the density \(\pi\).  Thus, if
the denominator is small, the expected squared estimation error on the
left side must be large.  The version used here is stated in
\Cref{lem:vector-van-trees} and proven in \Cref{app:vector-van-trees}.

For the hard family below, the parameter is a matrix
\(X\in\C^{(d-r)\times r}\) satisfying
\(\norm X_{\mathrm{op}}<a\).  We construct \(\pi\) by smoothly truncating a
Gaussian distribution so that \(X\) is supported on this operator norm ball,
and show that \(I(\pi)=O(d^2r/a^2)\).  If an estimator \(T\) also satisfies
\(\norm T_{\mathrm{op}}\le a\), then standard matrix norm inequalities give
\(\norm{T-X}_1\ge\norm{T-X}_\F^2/(2a)\), converting the squared Frobenius loss
in van Trees into trace norm loss.
The construction of \(\pi\) and this conversion to expected trace norm loss
are given in
\Cref{lem:rectangular-local-prior,lem:localized-nuclear-risk}.

On the other hand, confidence amplification and a simple postprocessing step
turn a successful tomographic estimator into an estimator of \(X\) with
operator norm at most \(a\) and small expected trace norm loss.  Consequently,
a sufficiently large right side in the van Trees inequality rules out
such a tomographic estimator.  The main task is therefore to bound the
transcript term in the denominator from above.  The required postprocessing is carried out in
\Cref{sec:bounded-support-estimator-postprocessing}.

\subsubsection{Fisher information along an adaptive decision tree}

An adaptive protocol chooses each measurement according to the outcomes
observed so far.  We represent this process as a decision tree: each node
\(h\) contains the preceding outcomes together with any private randomness
used by the protocol, and the outgoing edges represent the possible outcomes
of the next measurement.  Once a node \(h\) is reached at round \(i\), the
protocol has selected a fixed POVM acting jointly on \(t_i(h)\le t\) fresh
samples.  Therefore, a Fisher information bound that holds for every fixed
POVM applies to the measurement selected at every node.  A complete
transcript corresponds to a root-to-leaf path through the tree, and the
unknown state determines the distribution over these paths.

It remains to combine the bounds along such a path.  Let \(H_i\) be the
random history reached before round \(i\), and let \(\cI_i(h;X)\) be the
Fisher information matrix of the fixed POVM selected at history \(h\) when
the parameter is \(X\).  We write \(\E_X\) for expectation over the transcript
generated at parameter \(X\).  Even with full adaptivity, Fisher information
satisfies the chain rule
\[
  \Tr\bigl(\cI_{\mathrm{tr}}(X)\bigr)
  =
  \E_X\left[
    \sum_i \Tr\bigl(\cI_i(H_i;X)\bigr)
  \right].
\]
We prove this identity in \Cref{lem:block-fisher-chain}.
The expectation averages over the random path generated at parameter \(X\),
while the sum contains the Fisher information of the fixed POVM used at each
node on that path.  The main technical step, addressed in the next subsection,
is to bound the Fisher information of the joint measurement selected at an
arbitrary node.  Once this nodewise bound is available, applying it inside the
sum and using the pathwise sample bound \(\sum_i t_i(H_i)\le n\) controls the
Fisher information of the complete transcript.

\subsubsection{The hard family and the information at one node}

The hard states that the protocol must distinguish all have rank \(r\) and are
maximally mixed on their supports.  We keep their eigenvalues fixed and vary
only the support.
Fix an \(r\)-dimensional reference subspace \(S\subseteq\C^d\), and consider
linear maps \(X:S\to S^\perp\) with small operator norm.  Such a map specifies
the nearby \(r\)-dimensional subspace
\[
  \{\vct{u}+X\vct{u}:\vct{u}\in S\}.
\]
When \(X=0\), this subspace is \(S\); small \(\norm{X}_{\mathrm{op}}\) therefore
describes a small rotation of the support away from \(S\).  Let \(\Pi_X\) be
the orthogonal projector onto this subspace and set \(\rho_X\coloneqq\Pi_X/r\).
The real and imaginary parts of \(X\) give \(p\coloneqq 2r(d-r)\) local parameters.  At \(X=0\), a direction
\(H\in\C^{(d-r)\times r}\) changes the state by
\[
  \Ddir_H\rho_0
  =\frac1r
    \begin{pmatrix}0&H^\dagger\\H&0\end{pmatrix}.
\]
For \(X\) and \(Y\) in this small operator norm neighborhood, closeness of the
corresponding states implies closeness of their support parameters:
\[
  \norm{X-Y}_1\le2r\norm{\rho_X-\rho_Y}_1.
\]
Thus, accurate tomography of this family would give an accurate estimator of
\(X\).  The construction and the metric comparison are proven in
\Cref{sec:support-rotation}.

The central estimate is that every joint POVM \(\mathsf M\) on \(t\) samples
satisfies
\begin{equation}
  \Tr\bigl(\cI_{\mathsf M}^{(t)}(X)\bigr)
  =O\left(
    dt\min\left\{1,\frac{\sqrt t}{r}\right\}
  \right),
  \label{eq:overview-block-fisher}
\end{equation}
for every \(X\).  Here the Fisher matrix uses the \(2r(d-r)\) real support
coordinates.  At a node where the protocol jointly measures
\(t_i(h)\le t\) samples, we apply the same estimate with \(t_i(h)\) in place of
\(t\).  This estimate is proven in
\Cref{prop:one-block-fisher-bound}.

We begin the proof of \cref{eq:overview-block-fisher} with the simpler
single-sample case.  Write the POVM density \(M_z\), with respect to a reference
measure \(\nu\), in blocks relative to \(S\oplus S^\perp\):
\[
  M_z=
  \begin{pmatrix}A_z&B_z^\dagger\\B_z&C_z\end{pmatrix}\succeq0.
\]
Here \(A_z\) acts on \(S\), \(C_z\) acts on \(S^\perp\), and \(B_z\) maps
\(S\) to \(S^\perp\).
Let \(q_X(z)\coloneqq\Tr(M_z\rho_X)\) denote the likelihood density of the
measurement outcome at parameter \(X\).  At the reference point \(X=0\),
\(\rho_0=\operatorname{diag}(\Id_r/r,0)\), so
\begin{equation}
  q_0(z)=\frac{\Tr(A_z)}r.
  \label{eq:overview-single-copy-likelihood}
\end{equation}
For a support rotation direction \(H\in\C^{(d-r)\times r}\), the state
derivative displayed above gives
\[
  \Ddir_Hq_0(z)
  =\Tr\bigl(M_z\Ddir_H\rho_0\bigr)
  =\frac2r\operatorname{Re}\left(\Tr(B_z^\dagger H)\right).
\]
Thus \(A_z\) determines the outcome probability, while \(B_z\) determines
its first-order response to a support rotation.  Summing the squared
derivatives over an orthonormal basis of \(\C^{(d-r)\times r}\), regarded as a
real inner product space, gives
\begin{equation}
  \norm{\nabla q_0(z)}_2^2
  =\frac4{r^2}\norm{B_z}_\F^2.
  \label{eq:overview-single-copy-gradient}
\end{equation}
Substituting
\cref{eq:overview-single-copy-likelihood,eq:overview-single-copy-gradient}
into the Fisher information trace formula yields
\[
  \Tr\bigl(\cI_{\mathsf M}^{(1)}(0)\bigr)
  =\frac4r\int
    \frac{\norm{B_z}_\F^2}{\Tr(A_z)}
    \,\dd\nu(z).
\]
If \(\Tr(A_z)=0\), positivity also forces \(B_z=0\), and the quotient is
taken to be zero.

It remains to bound the quotient.  Positivity of \(M_z\) implies
\[
  \norm{B_z}_\F^2\le\Tr(A_z)\Tr(C_z).
\]
Finally, POVM normalization gives
\(\int C_z\,\dd\nu(z)=\Id_{d-r}\), and hence
\[
  \Tr\bigl(\cI_{\mathsf M}^{(1)}(0)\bigr)
  \le\frac4r\int\Tr(C_z)\,\dd\nu(z)
  \le\frac{4d}{r}.
\]
The complete proof of the single-sample bound appears in
\Cref{sec:single-copy-fisher-bound}.  Both this calculation and the calculation for joint
measurements below are carried out at \(X=0\).  At the end of this subsection,
we explain how a change of basis transfers the resulting bounds to arbitrary
\(X\).

We now consider a joint POVM on \(t>1\) samples, again writing \(M_z\) and
\(\nu\) for its operator density and reference measure.
Consider an infinitesimal perturbation that tilts a unit vector
\(\vct{h}\in S\) toward a fixed direction in \(S^\perp\).  Differentiating
the product state in this direction gives a sum of \(t\) terms, according to which tensor factor is
rotated.  Squaring this sum in the
Fisher information formula produces terms involving either the same position
or two distinct positions.  The latter are controlled by the components of
the other tensor factors in the direction \(\vct{h}\).  For each fixed position
of the rotated tensor factor, the relevant sum of projections onto
\(\vct{h}\) on the other \(t-1\) factors is
\[
  N_{\vct{h}}^{(t-1)}
  \coloneqq\sum_{\ell=1}^{t-1}
    \Id_S^{\otimes(\ell-1)}\otimes
    |\vct{h}\rangle\langle\vct{h}|\otimes
    \Id_S^{\otimes(t-1-\ell)}.
\]
The goal is to repeat the single-sample bound on
\(\norm{B_z}_\F^2/\Tr(A_z)\).  With \(t\) samples, the analogous numerator
also contains the sum over possible excitation positions.  Let
\(\vct f_1,\ldots,\vct f_{d-r}\) be an orthonormal basis of \(S^\perp\).
For \(i\in[d-r]\), let \(C_{z,i}\) be
the POVM block on the subspace spanned by tensors with exactly one factor
in direction \(\vct f_i\) and all remaining factors in \(S\).  It plays the
role of \(C_z\) in the single-sample calculation.  If a common operator
\(T_{t-1,r}\) dominates \(N_{\vct{h}}^{(t-1)}\) for every unit
\(\vct{h}\), then placing one copy of \(\Id+T_{t-1,r}\) at each of the
\(t\) possible excitation positions gives an operator \(Q_{t,r}^{(i)}\)
on this block.  The Cauchy--Schwarz argument for a positive semidefinite
block matrix cancels the likelihood denominator and bounds the remaining
quotient by
\(\Tr(C_{z,i}Q_{t,r}^{(i)})\); see
\cref{eq:fisher-row-envelope-bound}.  For each fixed \(i\), this bound
controls rotations from any unit direction \(\vct{h}\in S\) toward
\(\vct f_i\).

We choose \(T_{t-1,r}\), and hence \(Q_{t,r}^{(i)}\), independently of the
outcome \(z\).  POVM normalization then gives
\[
  \int\Tr(C_{z,i}Q_{t,r}^{(i)})\,\dd\nu(z)
  =\Tr(Q_{t,r}^{(i)}).
\]
Thus the trace of this common upper bound controls the Fisher information.
It therefore suffices to find a direction-independent bound on
\(N_{\vct{h}}^{(t-1)}\) with small normalized trace, meaning the trace
divided by \(r^{t-1}\).
The immediate bound
\(N_{\vct{h}}^{(t-1)}\preceq(t-1)\Id\) has normalized trace \(t-1\), which is
too large.  For \(1\le k\le t-1\), define
\[
  \binom{N_{\vct{h}}^{(t-1)}}k
  \coloneqq\frac1{k!}\prod_{j=0}^{k-1}
    \left(N_{\vct{h}}^{(t-1)}
      -j\Id_{S^{\otimes(t-1)}}\right).
\]
Because the summands defining \(N_{\vct{h}}^{(t-1)}\) are commuting projectors,
\[
  \binom{N_{\vct{h}}^{(t-1)}}k
  =\sum_{\substack{J\subseteq[t-1]\\|J|=k}}
    \prod_{\ell\in J}\left(
      \Id_S^{\otimes(\ell-1)}\otimes
      |\vct{h}\rangle\langle\vct{h}|\otimes
      \Id_S^{\otimes(t-1-\ell)}
    \right).
\]
For each \(J\), the product on the right projects the selected \(k\) factors
onto \(\vct{h}^{\otimes k}\).  This vector belongs to the symmetric subspace
of those factors, so the product is bounded by the corresponding
orthogonal projector \(\Pi_{\mathrm{sym},J}\).  Therefore,
\[
  \binom{N_{\vct{h}}^{(t-1)}}k
  \preceq
  \Omega_{t-1,k}
  \coloneqq\sum_{\substack{J\subseteq[t-1]\\|J|=k}}
    \Pi_{\mathrm{sym},J}.
\]
The operator \(\Omega_{t-1,k}\) is independent of \(\vct{h}\).  Since
\(N_{\vct{h}}^{(t-1)}\) is a sum of \(t-1\) commuting projectors, its
eigenvalues belong to \(\{0,1,\ldots,t-1\}\).  Using the scalar inequality
\[
  q\le k+k\binom qk^{1/k},
  \qquad q\in\{0,1,\ldots,t-1\},
\]
and the bound above, we obtain
\[
  T_{t-1,r}\coloneqq
  k\Id_{S^{\otimes(t-1)}}+k\Omega_{t-1,k}^{1/k},
  \qquad
  N_{\vct{h}}^{(t-1)}
  \preceq k\Id_{S^{\otimes(t-1)}}
    +k\binom{N_{\vct{h}}^{(t-1)}}k^{1/k}
  \preceq T_{t-1,r}.
\]
Since the symmetric subspace of \(S^{\otimes k}\) has dimension
\(\binom{r+k-1}{k}\), standard binomial estimates with
\(k\coloneqq\lceil\sqrt{t-1}\rceil\) give
\[
  \frac{\Tr(T_{t-1,r})}{r^{t-1}}
  =O\left(k+\frac{t-1}{k}+\frac{t-1}{r}\right)
  =O\left(\sqrt t+\frac tr\right).
\]

Applying this estimate at each excitation position bounds the trace of
\(Q_{t,r}^{(i)}\).  The block positivity bound and POVM normalization then give
\[
  \Tr\bigl(\cI_{\mathsf M}^{(t)}(0)\bigr)
  \le\frac4{r^t}\sum_{i=1}^{d-r}
    \int\Tr(C_{z,i}Q_{t,r}^{(i)})\,\dd\nu(z)
  =\frac4{r^t}\sum_{i=1}^{d-r}\Tr(Q_{t,r}^{(i)})
  =O\left(\frac{dt}{r}\left(\sqrt t+\frac tr\right)\right).
\]
For \(t\le r^2\), we have \(t/r\le\sqrt t\), so the resulting bound is
\(O(dt\sqrt t/r)\).  For
\(t>r^2\), the standard classical--quantum Fisher inequality and additivity
of SLD quantum Fisher information give
\(\Tr(\cI_{\mathsf M}^{(t)}(0))\le8dt\).

We have now bounded the Fisher information trace at \(X=0\) for every joint
POVM.  For an arbitrary \(X\), we reduce to this case by a unitary change of
basis and show that the corresponding map on perturbations does not
increase their Frobenius norm.  We construct a unitary \(U_X\) satisfying
\(U_X^\dagger\rho_XU_X=\rho_0\).  Under this change of basis, every
perturbation \(H\) at \(X\) becomes a perturbation \(\mathcal L_X(H)\) at
\(0\), where
\[
  U_X^\dagger(\Ddir_H\rho_X)U_X
  =\Ddir_{\mathcal L_X(H)}\rho_0,
  \qquad
  \norm{\mathcal L_X(H)}_\F\le\norm H_\F.
\]
If \(\widetilde{\mathsf M}\) is obtained by conjugating \(\mathsf M\) with
\(U_X^{\otimes t}\), then
\[
  \cI_{\mathsf M}^{(t)}(X)
  =\mathcal L_X^\dagger
    \cI_{\widetilde{\mathsf M}}^{(t)}(0)\mathcal L_X.
\]
Since \(\mathcal L_X\) is a contraction,
\[
  \Tr\bigl(\cI_{\mathsf M}^{(t)}(X)\bigr)
  \le\Tr\bigl(\cI_{\widetilde{\mathsf M}}^{(t)}(0)\bigr).
\]
Applying the bounds proven at \(X=0\) to \(\widetilde{\mathsf M}\) proves
\cref{eq:overview-block-fisher}.  The complete proof appears in
\Cref{sec:multi-copy-fisher-bound}.

\subsubsection{Putting the bounds together}

Since \(t_i(H_i)\le t\) in each round and
\(\sum_i t_i(H_i)\le n\) on every transcript, the nodewise estimate and the
adaptive Fisher chain rule give
\[
  \Tr\bigl(\cI_{\mathrm{tr}}(X)\bigr)
  =O\left(
    dn\min\left\{1,\frac{\sqrt t}{r}\right\}
  \right).
\]
This transcript bound is stated formally in
\Cref{cor:adaptive-information-budget}.

The support rotation parameter has real dimension \(p=2r(d-r)\), and the
smoothly truncated Gaussian density described above satisfies
\(I(\pi)=O(d^2r/a^2)\).  If \(n\) is below a
sufficiently small universal
constant multiple of
\[
  \frac{dr}{a^2}
  \max\left\{1,\frac r{\sqrt t}\right\},
\]
then both terms in the van Trees denominator are \(O(d^2r/a^2)\).
The van Trees inequality gives expected squared
Frobenius loss \(\Omega(a^2r)\), and operator norm localization converts this
to expected trace norm loss \(\Omega(ar)\).

By contrast, the accuracy guarantee for an actual tomography algorithm
allows failure with constant probability: it requires trace norm error at
most \(\eps\) with probability at least \(2/3\).  For arbitrary matrix outputs, the error on
failure can be arbitrarily large, so this guarantee alone does not control
the expected loss.  Let \(\delta\in(0,1)\) be a target failure probability,
to be chosen as a small constant.  Confidence amplification
(\Cref{lem:metric-median}) reduces the failure probability to at most
\(\delta\), with trace norm error at most \(3\eps\), using
\(O(1+\log(1/\delta))\) independent runs on fresh samples.

We then postprocess the amplified output into an estimator \(\widehat X\)
whose operator norm is at most \(a\) on every run.  This preserves
\(O(r\eps)\) trace norm accuracy on the successful runs, as proven in
\Cref{lem:bounded-support-estimator-postprocessing}.  On the remaining
failure event, both \(X\) and \(\widehat X\) have operator norm at most
\(a\), so \(\norm{\widehat X-X}_1\le2ar\).  Splitting the expectation over
these two events gives expected trace norm loss
\(O(r\eps)+O(\delta ar)\).  Taking \(a\) to be a sufficiently large constant
multiple of \(\eps\), and then taking \(\delta\) to be a sufficiently small
constant, makes this upper bound incompatible with the lower bound
\(\Omega(ar)\) above.  For this fixed \(\delta\), the repetitions increase
the sample count only by a constant factor.  This gives the claimed lower
bound
\[
  n=\Omega\left(
    \frac{dr}{\eps^2}
    \max\left\{1,\frac r{\sqrt t}\right\}
  \right)
\]
when \(r\le d/2\); see \Cref{prop:rectangular-regime-lower-bound}.

When \(r>d/2\), the support rotation family has \(2r(d-r)\) real parameters,
which can be too few when \(r\) is close to \(d\).  We first fix an
\(r\)-dimensional subspace \(V\subseteq\C^d\) and embed into \(V\) the
hard family for ambient dimension \(r\) and rank
\(k\coloneqq\lfloor r/2\rfloor\).  Every embedded input state \(\sigma\)
then has support contained in the same fixed \(V\).
Let \(\Pi_V\) be the orthogonal projector onto \(V\).  Mixing each embedded
input equally with the maximally mixed state \(\Pi_V/r\) gives
\(\rho_\sigma\coloneqq(\sigma+\Pi_V/r)/2\), whose support is exactly \(V\)
and whose rank is exactly \(r\).  Any algorithm that works for all
rank-\(r\) states must therefore work for this lifted family.  The trace
norm distance between any two lifted states is exactly half the distance
between their corresponding inputs.

The channel acts independently on each sample, so an adaptive protocol
for the lifted family can be simulated on the original family using the
same number of samples in each round.  Its estimate can be converted back
with at most twice the error.  Thus the lower bound for dimension \(r\)
and rank \(k\) transfers to the rank-\(r\) problem in dimension \(d\).
Since \(k=\Theta(r)\) and \(r=\Theta(d)\), this gives the claimed rate.
The reduction is proven in \Cref{lem:depolarizing-rank-lift}.

\subsection{The tomography protocol}
\label{sec:upper-bound-overview}

The tomography protocol repeatedly applies a Gaussian joint measurement to fresh samples of the unknown state~$\rho$,
averages the Hermitian matrix estimates corresponding to the outcomes, and projects the average onto the
set of density matrices of rank at most \(r\).  We first define the measurement and
state the complete estimation procedure.  We then reduce its
analysis to three blocks of the deviation of the estimator from~$\rho$, explain the properties of the measurement
that help control the  blocks, and show why the construction has those
properties.

\subsubsection{The joint measurement and estimation procedure}

We first define \(\mathsf M_t\), a joint measurement on \(t\) samples.  For
\(1\le\ell\le\min\{d,t\}\), let \(\gamma_{d,\ell}\) be the
standard complex Gaussian law on \(\C^{d\times\ell}\), and define
\[
  \Gamma_{\ell,t}
  \coloneqq
  \int (GG^\dagger)^{\otimes t}\,\dd\gamma_{d,\ell}(G).
\]
Let \(\Pi_{\ell,t}\) be the orthogonal projector onto the support of
\(\Gamma_{\ell,t}\), put \(\Pi_{0,t}=0\), and set
\(\Delta_{\ell,t}\coloneqq\Pi_{\ell,t}-\Pi_{\ell-1,t}\).  For the outcome
\(J=\ell\), define
\[
  \mathsf M_t(\ell,\dd G)
  \coloneqq
  \Delta_{\ell,t}(\Gamma_{\ell,t}^+)^{1/2}
  (GG^\dagger)^{\otimes t}
  (\Gamma_{\ell,t}^+)^{1/2}\Delta_{\ell,t}
  \,\dd\gamma_{d,\ell}(G),
\]
where \(\Gamma_{\ell,t}^+\) denotes the Moore--Penrose pseudoinverse of
\(\Gamma_{\ell,t}\).  The Gaussian measure \(\gamma_{d,\ell}\) is a
reference measure.  The Born rule determines the outcome density with
respect to this measure.
The support projectors \(\Pi_{\ell,t}\) are nested, their ranges span the
entire tensor product space, and integrating \(\mathsf M_t(\ell,\dd G)\)
over \(G\) gives \(\Delta_{\ell,t}\); hence these densities form a POVM.
These claims are proven in
\Cref{lem:gaussian-moment-support,lem:gaussian-block-povm-normalization}.

The protocol uses \(\mathsf M_s\) to jointly measure \(s\) samples at a time,
where
\[
  s\coloneqq\min\{t,r^2\}.
\]
For a fixed total number of samples, increasing the number of samples measured
jointly improves the asymptotic error bound only until this number reaches
\(r^2\), which explains this choice.  For a sufficiently large universal
constant \(C_0\), set
\[
  B\coloneqq\left\lceil
     C_0\frac{dr^2}{\eps^2 s\sqrt s}
  \right\rceil,
  \qquad
  N\coloneqq Bs.
\]
Here \(B\) is the number of measurement rounds and \(N\) is the total
number of samples.
Put \(\ell_{\max}\coloneqq\min\{d,s\}\).  The measurement \(\mathsf M_s\) has
outcome \((J,G)\), where \(1\le J\le\ell_{\max}\) and
\(G\in\C^{d\times J}\).  From this outcome, we form the Hermitian matrix
\[
  Y\coloneqq\frac{GG^\dagger-J\Id_d}{s}.
\]
We apply the same measurement independently \(B\) times, using \(s\) fresh
samples each time.  If \(Y_1,\ldots,Y_B\) are the resulting matrices, set
\[
  \overline Y\coloneqq\frac1B\sum_{b=1}^B Y_b.
\]
The output is a density matrix nearest in Frobenius norm to \(\overline Y\),
subject to having rank at most \(r\):
\[
  \widehat\rho\in
  \underset{\sigma\in\cD_r(\C^d)}{\arg\min}\,
    \norm{\overline Y-\sigma}_\F.
\]
This is the complete protocol.  The same joint measurement is used in every
repetition, so the protocol is nonadaptive.

\subsubsection{Reducing the error to three blocks}

Suppose the unknown state \(\rho\) has rank at most \(r\). Choose an \(r\)-dimensional
subspace containing its support, let \(\Pi\) be the orthogonal projector
onto this subspace, and put
\(\Pi^\perp\coloneqq\Id_d-\Pi\) and
\(E\coloneqq\overline Y-\rho\).  The projection bound in
\Cref{lem:rank-constrained-frobenius-projection} gives
\[
  \norm{\widehat\rho-\rho}_1^2
  =O\left(r\left(
    \norm{\Pi E\Pi}_\F^2
    +\norm{\Pi^\perp E\Pi}_\F^2
    +r\norm{\Pi^\perp E\Pi^\perp}_{\mathrm{op}}^2
  \right)\right).
\]
Thus it suffices to control \(\Pi E\Pi\) and \(\Pi^\perp E\Pi\) in Frobenius
norm and \(\Pi^\perp E\Pi^\perp\) in operator norm.  The projector \(\Pi\) is
used only in the analysis and is not known to the protocol.

\subsubsection{Properties needed from one joint measurement}

The blocks involving \(\operatorname{Im}(\Pi)\) are controlled by the second moment~$\E_\rho[Y\otimes Y]$.
Suppose \(\rank(\rho)\le r\), and let
\(F\) be the swap operator between the two factors of \(\C^d\otimes\C^d\).
The measurement defined above gives an unbiased estimator with an explicit
second moment identity, and the
expected number of columns of \(G\) is small:
\begin{equation}
  \E_\rho[Y]=\rho,
  \qquad
  \E_\rho[Y\otimes Y]
  =\frac{s-1}{s}\rho^{\otimes2}
    +\frac1s(\rho\otimes\Id_d+\Id_d\otimes\rho)F
    +\frac{\E_\rho[J]}{s^2}F,
  \qquad
  \E_\rho[J]=O\bigl(\sqrt s\bigr).
  \label{eq:overview-gaussian-block-moment}
\end{equation}
The second moment identity in \cref{eq:overview-gaussian-block-moment} helps bound
the Frobenius norms of \(\Pi E\Pi\) and \(\Pi^\perp E\Pi\). Together with \(\E_\rho[J]=O(\sqrt s)\) and \(s\le r^2\), it implies
\begin{equation}
\label{eq-overview-deviation-bound}  
  \E_\rho\left[
    \norm{\Pi E\Pi}_\F^2
    +2\norm{\Pi^\perp E\Pi}_\F^2
  \right]
  =O \! \left(
    \frac{dr}{N\sqrt s}
  \right).
\end{equation}
The factor of \(s^{-1/2}\) is the gain from
measuring \(s\) samples jointly.  The moment identities are proven in
\Cref{sec:gaussian-block-moments}, the estimate on \(\E_\rho[J]\) in
\Cref{lem:gaussian-support-index-mean}, and the Frobenius bound in
\Cref{lem:gaussian-estimator-support-blocks}.

The \(\Pi^\perp E\Pi^\perp\) block requires a different property: the
projection bound asks for control in operator norm, which
the above Frobenius bound does not provide.  The additional property is
also part of \Cref{prop:gaussian-block-estimator}:
for every \(\ell\) with \(\Prb_\rho[J=\ell]>0\), conditional on
\(J=\ell\), the columns of \(\Pi^\perp G\) remain independent standard
complex Gaussian vectors in \(\operatorname{Im}(\Pi^\perp)\).  If \(r=d\),
then \(\Pi^\perp=0\) and there is no complementary block.
Otherwise, for independent outcomes
\((J_1,G_1),\ldots,(J_B,G_B)\), put \(K\coloneqq\sum_bJ_b\) and concatenate
the coordinate matrices of \(\Pi^\perp G_b\), in an orthonormal basis of
\(\operatorname{Im}(\Pi^\perp)\), into
\(Z\in\C^{(d-r)\times K}\).  Conditional on \(J_1,\ldots,J_B\), the matrix
\(Z\) is standard complex Gaussian.  In these coordinates,
\[
  \Pi^\perp E\Pi^\perp
  =\frac1N(ZZ^\dagger-K\Id_{d-r}),
\]
so controlling this block reduces to bounding the deviation
\(ZZ^\dagger-K\Id_{d-r}\) in operator norm.  The required bound is
\Cref{lem:gaussian-covariance-operator-norm}.  The resulting bounds in Frobenius
norm and operator norm for these blocks are proven in
\Cref{lem:gaussian-estimator-support-blocks}.

\subsubsection{Why the joint measurement has these properties}

We first verify that \(\mathsf M_s\) is a POVM.  Write
\(\mathcal V_{\ell,s}\) for the span of all \(L^{\otimes s}\) with
\(L\subseteq\C^d\) and \(\dim L\le\ell\), with
\(\mathcal V_{0,s}=\{0\}\).  Positivity, continuity, and the full support of
the reference Gaussian measure show that \(\Gamma_{\ell,s}\) has support
\(\mathcal V_{\ell,s}\).  These subspaces are nested and span the tensor
product space, so the differences \(\Delta_{\ell,s}\) are orthogonal projectors that
sum to the identity.  The pseudoinverse factors ensure that integrating
\(\mathsf M_s(\ell,\dd G)\) over \(G\) gives \(\Delta_{\ell,s}\),
establishing the validity of the measurement;
see \Cref{lem:gaussian-moment-support,lem:gaussian-block-povm-normalization}.

To bound \(\E_\rho[J]\), we use the distribution of \(J\) obtained from
this normalization.  The Born rule gives
\[
  \Prb_\rho[J=\ell]
  =\Tr(\Delta_{\ell,s}\rho^{\otimes s}),
  \qquad 1\le\ell\le\ell_{\max},
\]
and hence, for \(1\le q\le\ell_{\max}\),
\[
  \Prb_\rho[J\ge q]
  =\Tr\bigl((\Id_d^{\otimes s}-\Pi_{q-1,s})\rho^{\otimes s}\bigr).
\]
For each \(q\)-element subset \(T\subseteq[s]\), let \(A_T\) be the
orthogonal projector that antisymmetrizes the tensor factors in \(T\).
We bound the operator in this tail probability by a sum of these projectors:
\[
  \mathcal V_{q-1,s}
  =\bigcap_{\substack{T\subseteq[s]\\|T|=q}}\ker(A_T),
  \qquad
  \Id_d^{\otimes s}-\Pi_{q-1,s}
  \preceq\sum_{\substack{T\subseteq[s]\\|T|=q}}A_T.
\]
The two statements are proven in
\Cref{lem:gaussian-support-antisymmetrizers,lem:gaussian-support-index-mean}.

Let \(A_q\) be the antisymmetrizer on \((\C^d)^{\otimes q}\), and let
\(\lambda_1,\ldots,\lambda_d\) be the eigenvalues of \(\rho\).  Each of the
\(\binom sq\) projectors has the same expectation
\(\Tr(A_q\rho^{\otimes q})\) in the product state \(\rho^{\otimes s}\), so
\[
  \Prb_\rho[J\ge q]
  \le\binom sq\Tr(A_q\rho^{\otimes q})
  =\binom sq
    \sum_{1\le i_1<\cdots<i_q\le d}
      \lambda_{i_1}\cdots\lambda_{i_q}
  \le\frac{\binom sq}{q!}
    \left(\sum_{i=1}^d\lambda_i\right)^q
  =\frac{\binom sq}{q!}.
\]
The bound
\[
  \frac{\binom sq}{q!}
  \le\left(\frac{\ee^2s}{q^2}\right)^q
\]
shows that the tail decreases geometrically once \(q\) is a sufficiently
large constant multiple of \(\sqrt s\).  Summing the tail probabilities
gives \(\E_\rho[J]=O(\sqrt s)\).
The full argument is in \Cref{lem:gaussian-support-index-mean}.

We next compute the first and second moments of \(Y\).  The main tool is a
permutation expansion of Gaussian tensor moments.  Let \(\mathfrak S_s\)
be the permutations of \([s]\), let \(U_\pi\) permute the tensor factors
according to \(\pi\), and let \(c(\pi)\) count its cycles, including fixed
points.  The complex form of Isserlis's theorem \cite{Isserlis1918} gives
\[
  \Gamma_{\ell,s}
  =\sum_{\pi\in\mathfrak S_s}\ell^{c(\pi)}U_\pi.
\]
This identity and the resulting commutation properties are proven in
\Cref{lem:gaussian-moment-identity}.

For the contribution from \(J=\ell\), the Born rule leads to Gaussian
integrals containing \((GG^\dagger)^{\otimes s}\) from the POVM and one or
two factors \((GG^\dagger-\ell\Id_d)/s\) from the estimator.
For the first moment, label the measured tensor factors by \(1,\ldots,s\)
and place the estimator matrix on an auxiliary output tensor factor \(\C^d_a\).
For \(i\in[s]\), let \(F_{ai}\) be the swap operator
between this auxiliary factor and input factor \(i\).
Writing \(\Id_a\) for the identity on this auxiliary factor and setting
aside the scalar \(1/s\), the required Gaussian identity is
\[
  \int
    (GG^\dagger-\ell\Id_d)_a\otimes(GG^\dagger)^{\otimes s}
    \,\dd\gamma_{d,\ell}(G)
  =\sum_{i=1}^sF_{ai}(\Id_a\otimes\Gamma_{\ell,s}).
\]
To prove this equality, apply the permutation expansion to the Gaussian
integral of \((GG^\dagger)^{\otimes(s+1)}\), with factors labelled
\(a,1,\ldots,s\).  Subtracting \(\ell\Id_d\) in factor \(a\) cancels
exactly the permutations that fix \(a\).  Every remaining permutation is
obtained by inserting \(a\) into a cycle of a permutation of \([s]\),
immediately before one of the \(s\) input labels.  Summing over these labels
gives the right side.  For the second moment, apply the same argument
to \((GG^\dagger)^{\otimes(s+2)}\), with auxiliary factors \(a\) and \(b\).
Centering both factors cancels the permutations that fix either one.  The
remaining permutations insert \(a\) and \(b\) before distinct input labels,
consecutively before the same input label in either order, or together in a
two-cycle.

We now use these Gaussian identities in the Born rule calculation.  In
each layer, the pseudoinverse square roots cancel
\(\Gamma_{\ell,s}\), leaving \(\Delta_{\ell,s}\rho^{\otimes s}\).
This uses the support and commutation properties in
\Cref{lem:gaussian-moment-support,lem:gaussian-moment-identity}.
Summing over the layers gives
\[
  \sum_{\ell=1}^{\ell_{\max}}
    \Delta_{\ell,s}\rho^{\otimes s}=\rho^{\otimes s}.
\]
The first moment calculation therefore gives \(\E_\rho[Y]=\rho\).
With two output factors, the three types of permutations above contribute,
respectively,
\[
  \frac{s-1}{s}\rho^{\otimes2},
  \qquad
  \frac1s(\rho\otimes\Id_d+\Id_d\otimes\rho)F,
  \qquad
  \frac{\E_\rho[J]}{s^2}F.
\]
Their sum is the second moment in
\cref{eq:overview-gaussian-block-moment}.  The full calculation is given in
\Cref{lem:centered-gaussian-moments,lem:fixed-index-moments} and in the proof
of \Cref{prop:gaussian-block-estimator}.

Finally, we verify the conditional Gaussian law used to control
\(\Pi^\perp E\Pi^\perp\).  Fix \(J=\ell\) with positive probability and
recall that \(\Pi\rho=\rho\).  The support and commutation properties above
imply that the Born rule density of \(G\), relative to
\(\gamma_{d,\ell}\), depends only on \(\Pi G\).  Under this reference
measure, \(\Pi G\) and \(\Pi^\perp G\) are independent.  The outcome
density therefore changes only the distribution of \(\Pi G\).
Conditional on \(J=\ell\), the matrix \(\Pi^\perp G\) remains independent
of \(\Pi G\), with independent standard complex Gaussian columns in
\(\operatorname{Im}(\Pi^\perp)\).  This proves the property needed for the
complementary block; see the proof of \Cref{prop:gaussian-block-estimator}.

\subsubsection{Putting the bounds together}

Recall that the second moment identity in \cref{eq:overview-gaussian-block-moment} helps us bound
the Frobenius norms of \(\Pi E\Pi\) and \(\Pi^\perp E\Pi\) as in \cref{eq-overview-deviation-bound}. 
Applying the operator norm covariance bound to the conditional distribution
of the concatenated matrix \(Z\) gives
\[
  \E_\rho\norm{\Pi^\perp E\Pi^\perp}_{\mathrm{op}}^2
  =O \! \left(
    \frac d{N\sqrt s}+\frac{d^2}{N^2}
  \right).
\]
When \(r=d\), \(\E_\rho\norm{\Pi^\perp E\Pi^\perp}_{\mathrm{op}}^2=0\).
These estimates are proven in
\Cref{lem:gaussian-estimator-support-blocks}.

Substituting the estimates into the projection bound gives
\[
  \E_\rho\norm{\widehat\rho-\rho}_1^2
  =O \! \left(
    \frac{dr^2}{N\sqrt s}+\frac{d^2r^2}{N^2}
  \right).
\]
This bound also explains the choice \(s=\min\{t,r^2\}\).  For fixed \(N\),
increasing \(s\) up to \(r^2\) decreases the first term until it reaches
order \(dr/N\).  At that point, the bound gives sample complexity
\(O(dr/\eps^2)\), matching the optimal rate for unrestricted joint
measurements.  Thus using larger joint measurements cannot improve the
asymptotic sample complexity.
With the stated choice of \(B\) and sufficiently large \(C_0\), Markov's
inequality gives trace norm error at most \(\eps\) with probability at least
\(2/3\), using the claimed number of samples.
The same joint POVM is used in every round, so the protocol is nonadaptive.
The complete argument is given in \Cref{sec:upper-bound-aggregation}.

\section{Preliminaries}
\label{sec:preliminaries}

We collect notation for matrices, measurements, tensor products, and Gaussian
matrices, and then recall the classical and quantum Fisher information facts
used in the lower bound.  The adaptive Fisher information chain rule and the
van Trees inequality are introduced later, where they first enter the proof.

\subsection{States, matrix spaces, and measurements}
\label{sec:state-measurement-preliminaries}

Throughout, all Hilbert spaces are finite dimensional and complex.  For a
positive integer $n$, write $[n]\coloneqq\{1,\ldots,n\}$.  We write
$\ii\coloneqq\sqrt{-1}$ for the imaginary unit and
$\ee$ for Euler's number.  For
$\vct{u},\vct{v}\in\cH$, we write $\ip{\vct{u}}{\vct{v}}$ for their inner
product and use the convention
$\ip{\vct{u}}{\vct{v}}=\vct{u}^\dagger\vct{v}$.  We use bold lowercase
symbols for vectors and ordinary lowercase symbols for their scalar
coordinates.
For a linear map $A:\cH_1\to\cH_2$, we write $A^\dagger$ for its adjoint.
Its kernel, image, and rank are
\[
  \ker(A)\coloneqq\{\vct{u}\in\cH_1:A\vct{u}=0\},
  \qquad
  \operatorname{Im}(A)
  \coloneqq\{A\vct{u}:\vct{u}\in\cH_1\},
  \qquad
  \rank(A)\coloneqq\dim\operatorname{Im}(A).
\]
When \(\cH_1=\cH_2\), we write \(\Tr(A)\) for its trace.
For a subspace $\mathcal U\subseteq\cH$, let $\mathcal U^\perp$ denote its
orthogonal complement and let $\Pi_\mathcal U$ denote the orthogonal projector onto
$\mathcal U$.  Thus
\[
  \operatorname{Im}(\Pi_\mathcal U)=\mathcal U,
  \qquad
  \Pi_\mathcal U^\dagger=\Pi_\mathcal U,
  \qquad
  \Pi_\mathcal U^2=\Pi_\mathcal U.
\]
We generally use the letter $\Pi$ for projectors; subscripts specify the
subspace or parameter when needed.
The notation $A\succeq0$ means that $A$ is positive semidefinite.  For
Hermitian operators $A$ and $B$, we write $A\preceq B$ if
$B-A\succeq0$.  A density matrix, or a quantum state, on a Hilbert space
$\cH$ is an operator
$\rho\succeq0$ with $\Tr(\rho)=1$.  We write $\cD(\cH)$ for the set of
density matrices on $\cH$, and $\Id_\cH$ for the identity operator on
$\cH$.  In particular, $\Id_d$ denotes the identity on $\C^d$.
For a positive integer $r\le\dim\cH$, write
\[
  \cD_r(\cH)
  \coloneqq\{\rho\in\cD(\cH):\rank(\rho)\le r\}.
\]
For a positive semidefinite operator $A$, its support is the subspace
\(\operatorname{supp}(A)\coloneqq\operatorname{Im}(A)=\ker(A)^\perp\).
For a projector $\Pi$ on $\cH$, write
\(\Pi^\perp\coloneqq\Id_\cH-\Pi\) for the projector onto
\(\operatorname{Im}(\Pi)^\perp\).

For a matrix $H$, its Frobenius norm is
\[
  \norm H_\F\coloneqq\sqrt{\Tr(H^\dagger H)}.
\]
Its trace norm, nuclear norm, or Schatten-$1$ norm, is
\[
  \norm H_1\coloneqq\Tr\left(\sqrt{H^\dagger H}\right).
\]
We also write
\(\abs H\coloneqq\sqrt{H^\dagger H}\) for its absolute value.
For a vector $\vct{u}$, write
$\norm{\vct{u}}\coloneqq\sqrt{\ip{\vct{u}}{\vct{u}}}$.  The operator norm
of a linear map $A$ is
\[
  \norm A_{\mathrm{op}}
  \coloneqq\sup_{\norm{\vct{u}}=1}\norm{A\vct{u}}.
\]
For a complex number $z$, $\operatorname{Re}(z)$ denotes its real part.
For positive integers \(m,r\), we regard the complex matrix space
$\C^{m\times r}$ as a real Euclidean space
with inner product
\begin{equation}
  \ip{H}{K}_{\R}
  \coloneqq\operatorname{Re}\left(\Tr(H^\dagger K)\right).
  \label{eq:real-matrix-inner-product}
\end{equation}
For $1\le i\le m$ and $1\le j\le r$, let
$E_{ij}\in\C^{m\times r}$ denote the matrix whose $(i,j)$ entry is one and
whose other entries are zero.  Then
\begin{equation}
  \{E_{ij},\,\ii E_{ij}:1\le i\le m,\ 1\le j\le r\}
  \label{eq:real-matrix-basis}
\end{equation}
is an orthonormal basis of $\C^{m\times r}$ regarded as a real inner product
space, with the inner product defined in
\cref{eq:real-matrix-inner-product}.  The basis matrices $E_{ij}$ and
$\ii E_{ij}$ correspond to the real and imaginary coordinates of the
$(i,j)$ entry, respectively.

For a differentiable real-valued function \(f\) on this matrix space, write
\(X_{ij}=x_{ij}+\ii y_{ij}\), with \(x_{ij},y_{ij}\in\R\).  Its gradient
consists of the partial derivatives with respect to these \(2mr\) real
coordinates.  We also write this gradient as a matrix whose \((i,j)\) entry
is \(\partial f/\partial x_{ij}+\ii\,\partial f/\partial y_{ij}\).
The Euclidean norm \(\norm{\nabla f(X)}_2\) equals the Frobenius norm of
this matrix.

Let $(\mathsf Z,\mathcal Z)$ be a measurable outcome space, where
$\mathsf Z$ is the set of possible outcomes and $\mathcal Z$ is its
sigma-algebra of measurable events.
A positive operator-valued measure
(POVM) on $(\mathsf Z,\mathcal Z)$ is a countably additive map $\mathsf M$
from $\mathcal Z$ to positive semidefinite operators on $\cH$, normalized by
\[
  \mathsf M(\mathsf Z)=\Id_\cH.
\]
When a state $\rho$ is measured by $\mathsf M$, the probability of observing
an outcome in a measurable event $\mathsf E\in\mathcal Z$ is
\begin{equation}
  \Prb_\rho^{\mathsf M}(\mathsf E)
  \coloneqq\Tr\bigl(\rho\,\mathsf M(\mathsf E)\bigr).
  \label{eq:born-rule-povm}
\end{equation}

\subsection{Tensor permutations and Gaussian matrices}
\label{sec:tensor-gaussian-preliminaries}

For an operator $L$ on $\cH$ and $\ell\in[n]$, we write
\[
  L^{[\ell]}
  \coloneqq
  \Id_\cH^{\otimes(\ell-1)}\otimes L\otimes
  \Id_\cH^{\otimes(n-\ell)}
\]
for the operator on $\cH^{\otimes n}$ that acts as $L$ on the $\ell$-th
tensor factor and as the identity on every other factor.  The ambient tensor
power will always be clear from context.

Let $\mathfrak S_n$ be the symmetric group on $[n]$.  For
$\pi\in\mathfrak S_n$, let $U_\pi$ denote the operator on
$\cH^{\otimes n}$ defined by
\[
  U_\pi(\vct{v}_1\otimes\cdots\otimes\vct{v}_n)
  \coloneqq
  \vct{v}_{\pi^{-1}(1)}\otimes\cdots\otimes
  \vct{v}_{\pi^{-1}(n)}.
\]
This convention gives $U_\pi U_\tau=U_{\pi\tau}$.  The swap operator on two tensor
factors is denoted by $F$; thus
$F(\vct{u}\otimes\vct{v})=\vct{v}\otimes\vct{u}$ and
\begin{equation}
  \Tr\bigl(F(A\otimes B)\bigr)=\Tr(AB)
  \label{eq:swap-trace-identity}
\end{equation}
for operators $A,B$ on the same space.

For an operator $K$ on $n$ tensor factors and $T\subseteq[n]$, $\Tr_T(K)$
denotes the partial trace over the factors in $T$, defined by linearity and
\[
  \Tr_T(A_1\otimes\cdots\otimes A_n)
  \coloneqq
  \left(\prod_{i\in T}\Tr(A_i)\right)
  \bigotimes_{i\in[n]\setminus T}A_i,
\]
where $A_i$ acts on the $i$-th factor and the remaining factors retain their
original order.

For a positive semidefinite operator $A$, we write $A^+$ for its
Moore--Penrose pseudoinverse.  Thus $(A^+)^{1/2}$ acts as $A^{-1/2}$ on the support
of $A$ and as zero on its kernel.

A standard complex Gaussian scalar is $z=x+\ii y$, where $x$ and $y$ are
independent real Gaussian random variables with mean zero and variance
$1/2$.  A standard complex Gaussian matrix has independent entries with this
distribution.  We write $\gamma_{d,\ell}$ for the probability measure of a
standard complex Gaussian matrix in $\C^{d\times\ell}$.

\subsection{Classical Fisher information}
\label{sec:classical-fisher-information}

We first recall the classical definition.  Let $\Theta\subseteq\R^p$ be an
open parameter set, write
$\vct{\theta}\coloneqq(\theta_1,\ldots,\theta_p)$, and let
$\{\Prb_{\vct{\theta}}:\vct{\theta}\in\Theta\}$ be a family of probability
measures on $(\mathsf Z,\mathcal Z)$ with likelihood densities
$q_{\vct{\theta}}$ relative to a parameter-independent measure $\nu$:
\begin{equation}
  \Prb_{\vct{\theta}}(\mathsf E)
  =\int_{\mathsf E}q_{\vct{\theta}}(z)\,\dd\nu(z)
  \label{eq:dominated-model}
\end{equation}
for every $\mathsf E\in\mathcal Z$.  Assume that
$\vct{\theta}\mapsto q_{\vct{\theta}}(z)$ is differentiable for
$\nu$-almost every $z$.  The measure-theoretic conventions and regularity
conditions used for continuous outcomes are collected in
\Cref{app:statistical-formalism}.  When \(\nu\) is counting measure on a
discrete outcome space, the integrals below become sums and
\(q_{\vct{\theta}}\) is the probability mass function.  The likelihood gradient is
\[
  \nabla q_{\vct{\theta}}(z)
  \coloneqq\bigl(\partial_1q_{\vct{\theta}}(z),\ldots,
          \partial_pq_{\vct{\theta}}(z)\bigr)^{\mathsf T}.
\]
For a direction $\vct{v}\in\R^p$, we use the notation
\[
  \Ddir_{\vct{v}}q_{\vct{\theta}}(z)
  \coloneqq
  \left.\frac{\dd}{\dd s}
    q_{\vct{\theta}+s\vct{v}}(z)\right|_{s=0}
  =\vct{v}^{\mathsf T}\nabla q_{\vct{\theta}}(z).
\]
The Fisher information matrix is defined as
\begin{equation}
  \cI(\vct{\theta})
  \coloneqq\int_{\mathsf Z}
    \frac{
      \nabla q_{\vct{\theta}}(z)
      \nabla q_{\vct{\theta}}(z)^{\mathsf T}
    }{q_{\vct{\theta}}(z)}\,\dd\nu(z).
  \label{eq:fisher-matrix}
\end{equation}
The integrand in \cref{eq:fisher-matrix} is defined to be zero when
$q_{\vct{\theta}}(z)=0$.

On the set where $q_{\vct{\theta}}(z)>0$, the score vector is
\[
  \vct{s}_{\vct{\theta}}(z)
  \coloneqq\nabla\log q_{\vct{\theta}}(z)
  =\frac{\nabla q_{\vct{\theta}}(z)}{q_{\vct{\theta}}(z)}.
\]
Consequently,
\[
  \cI(\vct{\theta})
  =\int_{\mathsf Z}
    \vct{s}_{\vct{\theta}}(z)
    \vct{s}_{\vct{\theta}}(z)^{\mathsf T}
    q_{\vct{\theta}}(z)\,\dd\nu(z),
\]
with the score set to zero on the zero-likelihood set.  Equivalently, the
Fisher information in a direction $\vct{v}$ is
\[
  \vct{v}^{\mathsf T}\cI(\vct{\theta})\vct{v}
  =\int_{\mathsf Z}
    \frac{
      (\Ddir_{\vct{v}}q_{\vct{\theta}}(z))^2
    }{q_{\vct{\theta}}(z)}\,\dd\nu(z).
\]

The quantity used in our lower bound is the Fisher information trace
$\Tr\bigl(\cI(\vct{\theta})\bigr)$.  It is given by
\begin{equation}
  \Tr\bigl(\cI(\vct{\theta})\bigr)
  =\int_{\mathsf Z}
    \frac{
      \norm{\nabla q_{\vct{\theta}}(z)}_2^2
    }{q_{\vct{\theta}}(z)}\,\dd\nu(z)
  =\sum_{j=1}^p\int_{\mathsf Z}
    \frac{
      (\partial_jq_{\vct{\theta}}(z))^2
    }{q_{\vct{\theta}}(z)}\,\dd\nu(z).
  \label{eq:fisher-trace}
\end{equation}
Here $\norm{\cdot}_2$ is the Euclidean norm on $\R^p$.

\subsection{Quantum Fisher information}
\label{sec:quantum-fisher-information}

We first apply the classical definition to the outcome distribution of a
quantum measurement.  To differentiate likelihoods at a fixed outcome, we
write all outcome laws relative to one measure that does not depend on the
state parameter.  In finite dimensions, every POVM provides such a measure
automatically.  Let $d\coloneqq\dim\cH$, let $\mathsf M$ be a POVM on
$(\mathsf Z,\mathcal Z)$, and define
\begin{equation}
  \nu(\mathsf E)\coloneqq\frac1d\Tr\bigl(\mathsf M(\mathsf E)\bigr).
  \label{eq:povm-dominating-measure}
\end{equation}
Finite dimensionality guarantees a positive semidefinite operator density
$M_z$ such that
\begin{equation}
  \mathsf M(\mathsf E)=\int_{\mathsf E} M_z\,\dd\nu(z),
  \qquad
  \int_{\mathsf Z}M_z\,\dd\nu(z)=\Id_\cH.
  \label{eq:povm-density}
\end{equation}
This dominated POVM formulation is standard in quantum statistics; see, for
example, \cite{BarndorffNielsenGill2000}.  Its short proof is included in
\Cref{app:single-povm-domination}.

Let $\vct{\theta}\mapsto\rho_{\vct{\theta}}$ be a differentiable family of
states, and let
$\Prb_{\vct{\theta}}\coloneqq
\Prb_{\rho_{\vct{\theta}}}^{\mathsf M}$ be its outcome law.  Combining
\cref{eq:born-rule-povm,eq:povm-density} gives
\[
  \Prb_{\vct{\theta}}(\mathsf E)
  =\int_{\mathsf E}\Tr(M_z\rho_{\vct{\theta}})\,\dd\nu(z).
\]
Therefore the likelihood density is
\begin{equation}
  q_{\vct{\theta}}(z)
  \coloneqq\Tr(M_z\rho_{\vct{\theta}}).
  \label{eq:povm-likelihood}
\end{equation}
For $\vct{v}\in\R^p$, write
\[
  \Ddir_{\vct{v}}\rho_{\vct{\theta}}
  \coloneqq
  \left.\frac{\dd}{\dd s}
    \rho_{\vct{\theta}+s\vct{v}}\right|_{s=0}.
\]
Since the right side of \cref{eq:povm-likelihood} is linear in the
state, its directional derivative is
\begin{equation}
  \Ddir_{\vct{v}}q_{\vct{\theta}}(z)
  =\Tr(M_z\Ddir_{\vct{v}}\rho_{\vct{\theta}}).
  \label{eq:povm-likelihood-derivative}
\end{equation}
We denote the resulting classical Fisher information matrix by
$\cI_{\mathsf M}(\vct{\theta})$.  Although the formulas use the particular
dominating measure $\nu$ and density $M_z$, this Fisher matrix depends only on
the outcome distributions generated by $\mathsf M$ and
$\rho_{\vct{\theta}}$.
The subscript $\mathsf M$ denotes the measurement that induces this classical
Fisher information matrix.

The quantum Fisher information matrix is instead associated directly with
the family $\vct{\theta}\mapsto\rho_{\vct{\theta}}$ and does not depend on a
choice of measurement.  For $j\in[p]$, let $\vct{e}_j$ denote the $j$-th standard basis
vector of $\R^p$, and write
$\partial_j\rho_{\vct{\theta}}
\coloneqq\Ddir_{\vct{e}_j}\rho_{\vct{\theta}}$.  A
\emph{symmetric logarithmic derivative} for the $j$-th parameter is a
Hermitian operator $L_j$ satisfying
\begin{equation}
  \partial_j\rho_{\vct{\theta}}
  =\frac12\bigl(
    L_j\rho_{\vct{\theta}}+\rho_{\vct{\theta}}L_j
  \bigr).
  \label{eq:symmetric-logarithmic-derivative}
\end{equation}
In finite dimensions, such an operator exists for a differentiable family of
states.  When $\rho_{\vct{\theta}}$ is not full rank, $L_j$ need not be
unique, but the matrix defined below is independent of the choice.  The
\emph{quantum Fisher information matrix} is
\begin{equation}
  \bigl(\cI_{\mathrm Q}(\vct{\theta})\bigr)_{jk}
  \coloneqq\frac12\Tr\bigl(
    \rho_{\vct{\theta}}(L_jL_k+L_kL_j)
  \bigr).
  \label{eq:quantum-fisher-matrix}
\end{equation}
This is a real symmetric positive semidefinite matrix.
The subscript $\mathrm Q$ distinguishes this measurement-independent quantum
quantity from the measurement-induced classical Fisher matrix
$\cI_{\mathsf M}$.

The quantum Fisher information matrix bounds from above the classical
Fisher information obtainable from any parameter-independent POVM:
\begin{equation}
  \cI_{\mathsf M}(\vct{\theta})
  \preceq\cI_{\mathrm Q}(\vct{\theta}).
  \label{eq:classical-quantum-fisher-bound}
\end{equation}
See \cite{BraunsteinCaves1994,GillMassar2000} for the corresponding
classical--quantum Fisher information comparison.

We will also use additivity under tensor products.  Fix a positive integer
$t$, and denote the quantum Fisher information matrix of the product family
$\rho_{\vct{\theta}}^{\otimes t}$ by
$\cI_{\mathrm Q}^{(t)}(\vct{\theta})$, where the superscript $(t)$ denotes
the number of samples.  A symmetric logarithmic derivative for the $j$-th
parameter of this family is
\[
  L_j^{(t)}
  \coloneqq\sum_{\ell=1}^t L_j^{[\ell]}.
\]
Indeed, differentiating the product state and applying
\cref{eq:symmetric-logarithmic-derivative} in each tensor factor gives
\[
  \partial_j\rho_{\vct{\theta}}^{\otimes t}
  =\frac12\left(
    L_j^{(t)}\rho_{\vct{\theta}}^{\otimes t}
    +\rho_{\vct{\theta}}^{\otimes t}L_j^{(t)}
  \right).
\]
Taking the trace in \cref{eq:symmetric-logarithmic-derivative} and using
$\Tr(\rho_{\vct{\theta}})=1$ gives
\[
  \Tr(\rho_{\vct{\theta}}L_j)
  =\Tr(\partial_j\rho_{\vct{\theta}})
  =\partial_j\Tr(\rho_{\vct{\theta}})
  =0.
\]
Consequently, for distinct $\ell,\ell'\in[t]$,
\[
  \Tr\bigl(
    \rho_{\vct{\theta}}^{\otimes t}L_j^{[\ell]}L_k^{[\ell']}
  \bigr)
  =\Tr(\rho_{\vct{\theta}}L_j)
    \Tr(\rho_{\vct{\theta}}L_k)
  =0.
\]
When we apply \cref{eq:quantum-fisher-matrix} to the product family using
$L_j^{(t)}$ and $L_k^{(t)}$, the cross terms between distinct tensor factors
vanish, while the $t$ terms with $\ell=\ell'$ each give the corresponding
entry of the single-sample quantum Fisher information matrix.  Consequently,
\begin{equation}
  \cI_{\mathrm Q}^{(t)}(\vct{\theta})
  =t\cI_{\mathrm Q}(\vct{\theta}).
  \label{eq:quantum-fisher-additivity}
\end{equation}
For a joint POVM $\mathsf M$ on $t$ samples, denote the induced classical
Fisher information matrix by
$\cI_{\mathsf M}^{(t)}(\vct{\theta})$; the same superscript again denotes
the number of samples measured jointly.  Then
\begin{equation}
  \Tr\bigl(\cI_{\mathsf M}^{(t)}(\vct{\theta})\bigr)
  \le t\Tr\bigl(\cI_{\mathrm Q}(\vct{\theta})\bigr).
  \label{eq:joint-measurement-quantum-fisher-bound}
\end{equation}

\subsection{Adaptive measurement protocols}
\label{sec:block-adaptive-model}

Fix a positive integer \(t\).  An adaptive protocol using joint measurements
on at most \(t\) samples receives independent samples of an unknown state.
In round \(i\), it chooses
an integer \(t_i\in\{1,\ldots,t\}\) and a joint POVM on \(t_i\) samples as
measurable functions of a private random seed and all earlier classical
outcomes.  It applies this POVM to \(t_i\) fresh samples and retains only the
classical outcome.
The samples measured in that round are destroyed, and no quantum memory
persists between rounds.  The measurable outcome space
\((\mathsf Z_i,\mathcal Z_i)\) in each round is assumed to be countably
generated: there is a countable collection of measurable events that
generates \(\mathcal Z_i\).  This includes finite and countable outcome
spaces, as well as Euclidean spaces with their Borel sigma-algebras.  A
protocol may halt early; its sample complexity is the
maximum, over every seed and every transcript, of the total number of samples
used.  Equivalently, after halting we may pad the protocol with zero-sample,
one-outcome rounds.
The final estimate is a measurable function of the private seed and the
observed outcomes.

The error in
\Cref{thm:adaptive-block-lower-bound,thm:bounded-block-upper-bound} is measured in the
Schatten \(1\) norm.  If trace distance is defined with the conventional
factor \(1/2\), only the universal constants change.

The restriction to classical memory between rounds is part of the theorem:
with persistent quantum memory, successive rounds could combine into a
larger coherent measurement and the number of fresh samples measured in each
round would no longer describe the experiment analyzed here.

\section{From Fisher information bounds to the adaptive lower bound}
\label{sec:lower-bound-proof}

This section proves \Cref{thm:adaptive-block-lower-bound}.  We first construct
a local parameterization of a hard family of states.  We next bound the Fisher
information obtainable by jointly measuring \(t\) samples at a time, beginning
with the simpler single-sample case.  We then
accumulate these bounds from each round along an adaptive transcript, convert the
resulting Fisher information bound into a lower bound on expected trace norm
loss using the van Trees inequality, and finish with a depolarizing channel
reduction for the large-rank regime.

\subsection{The hard family and its local coordinates}
\label{sec:support-rotation}

The hard states all have rank \(r\) and are maximally mixed on their supports.
We keep the eigenvalues fixed and vary only the support.  We parameterize
these support rotations by a matrix \(X\), with \(X=0\) corresponding to the
reference support.  The probability density introduced later is supported on
matrices \(X\in\C^{(d-r)\times r}\) with small operator norm, so only nearby
supports enter the lower bound argument.

Let $d\ge2$ and $1\le r<d$ be integers, and set $m\coloneqq d-r$.  Fix an orthogonal
decomposition
\[
  \C^d=S\oplus S^\perp,
  \qquad \dim S=r,
  \qquad \dim S^\perp=m.
\]
Choose orthonormal bases of $S$ and $S^\perp$, and use them to write vectors
in these subspaces as elements of $\C^r$ and $\C^m$, respectively.  Let
$\Pi_0$ be the orthogonal projector onto $S$.  With respect to
the resulting orthonormal basis of $\C^d$, our reference state has the block
representation
\begin{equation}
  \rho_0\coloneqq\frac1r\Pi_0
  =\frac1r\begin{pmatrix}
      \Id_r&0\\
      0&0_m
    \end{pmatrix}.
  \label{eq:reference-projector-state}
\end{equation}
We use this family directly when \(r\le d/2\); the case \(r>d/2\) is reduced
to this regime in \Cref{sec:rank-lifting}.

For $X\in\C^{m\times r}$, define a linear map
$V_X:\C^r\to\C^d$ by
\begin{equation}
  V_X\coloneqq
  \begin{pmatrix}\Id_r\\X\end{pmatrix}
  (\Id_r+X^\dagger X)^{-1/2}.
  \label{eq:graph-isometry}
\end{equation}
The normalization on the right makes $V_X$ an isometry.
Consequently,
\begin{equation}
  \Pi_X\coloneqq V_XV_X^\dagger,
  \qquad
  \rho_X\coloneqq\frac1r\Pi_X
  \label{eq:graph-family}
\end{equation}
are respectively an orthogonal projector and a density matrix of rank $r$.
The image of $\Pi_X$ is the $r$-dimensional subspace
\[
  \operatorname{Im}(\Pi_X)
  =\{(\vct{u},X\vct{u}):\vct{u}\in\C^r\}.
\]
At \(X=0\), this subspace is \(S\).  Matrices \(X\) with small operator norm
parameterize the \(r\)-dimensional subspaces near \(S\), obtained by tilting
vectors in \(S\) toward \(S^\perp\).  The columns of \(V_X\) form an
orthonormal basis of the corresponding support.

We next compute the tangent to this family at $X=0$.  Substituting
\cref{eq:graph-isometry} into the definition
$\Pi_X=V_XV_X^\dagger$ in \cref{eq:graph-family} gives
\[
  \Pi_X=
  \begin{pmatrix}\Id_r\\X\end{pmatrix}
  (\Id_r+X^\dagger X)^{-1}
  \begin{pmatrix}\Id_r&X^\dagger\end{pmatrix}.
\]
For a direction $H\in\C^{m\times r}$, substitute $X=sH$.  The derivative at
$s=0$ of $(\Id_r+s^2H^\dagger H)^{-1}$ is zero, and hence
\[
  \Ddir_H\Pi_0
  \coloneqq\left.\frac{\dd}{\dd s}\Pi_{sH}\right|_{s=0}
  =
  \begin{pmatrix}0\\H\end{pmatrix}
  \begin{pmatrix}\Id_r&0\end{pmatrix}
  +
  \begin{pmatrix}\Id_r\\0\end{pmatrix}
  \begin{pmatrix}0&H^\dagger\end{pmatrix}
  =
  \begin{pmatrix}0&H^\dagger\\H&0\end{pmatrix}.
\]
Since $\rho_X=\Pi_X/r$, the corresponding state derivative is
\begin{equation}
  \Ddir_H\rho_0
  \coloneqq\left.\frac{\dd}{\dd s}\rho_{sH}\right|_{s=0}
  =\frac1r\Ddir_H\Pi_0
  =\frac1r\begin{pmatrix}0&H^\dagger\\H&0\end{pmatrix}.
  \label{eq:state-tangent-at-zero}
\end{equation}
The real and imaginary parts of the $mr$ entries of $X$ therefore give
$2mr$ real tangent coordinates at the reference state.

The block representation also gives the metric comparison used when we later
recover the parameter from a tomographic estimate.  The lower left block of
\(\Pi_X\) is \(X(\Id_r+X^\dagger X)^{-1}\).  For \(X\) near zero, this block
changes at essentially the same rate as \(X\).  The following lemma makes
this precise: if \(\Pi_X\) and \(\Pi_Y\) are close in trace norm, then so are
\(X\) and \(Y\).

\begin{lemma}[Trace norm inverse Lipschitz bound]
  \label{lem:graph-metric}
  If $\norm X_{\mathrm{op}},\norm Y_{\mathrm{op}}\le1/4$, then
  \begin{equation}
    \norm{X-Y}_1
    \le2\norm{\Pi_X-\Pi_Y}_1
    =2r\norm{\rho_X-\rho_Y}_1.
    \label{eq:graph-inverse-lipschitz}
  \end{equation}
\end{lemma}

\begin{proof}
  Put $R_X\coloneqq(\Id_r+X^\dagger X)^{-1}$.  The lower left block of
  $\Pi_X$ is $XR_X$, and
  \[
    XR_X-YR_Y=(X-Y)R_X+Y(R_X-R_Y).
  \]
  Since $\norm X_{\mathrm{op}}\le1/4$,
  \[
    \Id_r\preceq\Id_r+X^\dagger X\preceq\frac{17}{16}\Id_r,
    \qquad
    \frac{16}{17}\Id_r\preceq R_X\preceq\Id_r.
  \]
  Therefore,
  \[
    \norm{(X-Y)R_X}_1\ge\frac{16}{17}\norm{X-Y}_1.
  \]
  Moreover,
  \[
    R_X-R_Y=R_X(R_Y^{-1}-R_X^{-1})R_Y
      =R_X(Y^\dagger Y-X^\dagger X)R_Y,
  \]
  and $\norm{R_X}_{\mathrm{op}},\norm{R_Y}_{\mathrm{op}}\le1$.  Using the
  standard trace norm inequality
  $\norm{ABC}_1\le\norm A_{\mathrm{op}}\norm B_1\norm C_{\mathrm{op}}$,
  we obtain
  \[
    \norm{Y(R_X-R_Y)}_1
    \le\norm Y_{\mathrm{op}}
      \norm{Y^\dagger Y-X^\dagger X}_1.
  \]
  Since
  \[
    Y^\dagger Y-X^\dagger X
    =Y^\dagger(Y-X)+(Y-X)^\dagger X,
  \]
  applying the matrix norm inequality again gives
  \begin{align*}
    \norm{Y(R_X-R_Y)}_1
    &\le\norm Y_{\mathrm{op}}
      \bigl(\norm X_{\mathrm{op}}+\norm Y_{\mathrm{op}}\bigr)
      \norm{X-Y}_1
    \le\frac18\norm{X-Y}_1.
  \end{align*}
  The reverse triangle inequality now gives
  \begin{align*}
    \norm{XR_X-YR_Y}_1
    &\ge\norm{(X-Y)R_X}_1-\norm{Y(R_X-R_Y)}_1\\
    &\ge\left(\frac{16}{17}-\frac18\right)\norm{X-Y}_1
    \ge\frac12\norm{X-Y}_1.
  \end{align*}
  Taking a block compression cannot increase the trace norm, which
  proves the first inequality in \cref{eq:graph-inverse-lipschitz}; the
  equality follows from $\rho_X=\Pi_X/r$.
\end{proof}

\subsection{Fisher information trace for joint measurements}
\label{sec:fisher-information-bounds}

Recall that \(m=d-r\).  The real and imaginary parts of
\(X\in\C^{m\times r}\) give the \(2mr\) real coordinates introduced in
\Cref{sec:support-rotation}.  For a joint POVM \(\mathsf M\) on \(t\)
samples, let \(\cI_{\mathsf M}^{(t)}(X)\) denote the classical
Fisher information matrix of its outcome distribution in these coordinates.
The following proposition bounds its trace for every measurement and every
parameter \(X\).  This is the main technical input to the adaptive lower
bound.

\begin{proposition}[Fisher information trace for a joint measurement on \(t\) samples]
  \label{prop:one-block-fisher-bound}
  There is a universal constant \(C>0\) such that, for every positive
  integer \(t\), every \(X\in\C^{m\times r}\), and every joint POVM
  \(\mathsf M\) on \(t\) samples,
  \begin{equation}
    \Tr\bigl(\cI_{\mathsf M}^{(t)}(X)\bigr)
    \le Cdt\min\left\{1,\frac{\sqrt t}{r}\right\}.
    \label{eq:one-block-fisher-bound}
  \end{equation}
\end{proposition}

We first prove the bound at the reference point \(X=0\), beginning with the
single-sample case and then treating joint measurements on \(t>1\) samples.  We
finally transfer the resulting bound to every parameter \(X\).

\subsubsection{Single-sample measurements}
\label{sec:single-copy-fisher-bound}

Let $\mathsf M$ be an arbitrary POVM on $\C^d$ with measurable outcome
space $(\mathsf Z,\mathcal Z)$.  Let $\nu$ and $M_z$ be its parameter-independent
dominating measure and operator density from
\cref{eq:povm-dominating-measure,eq:povm-density}.  Relative to
$S\oplus S^\perp$, write
\begin{equation}
  M_z=
  \begin{pmatrix}
    A_z&B_z^\dagger\\
    B_z&C_z
  \end{pmatrix}
  \succeq0,
  \label{eq:povm-blocks}
\end{equation}
where $A_z:S\to S$, $B_z:S\to S^\perp$, and
$C_z:S^\perp\to S^\perp$. Taking the three blocks of
$\int M_z\,\dd\nu(z)=\Id_d$ gives
\begin{equation}
  \int_{\mathsf Z}A_z\,\dd\nu(z)=\Id_r,
  \qquad
  \int_{\mathsf Z}B_z\,\dd\nu(z)=0,
  \qquad
  \int_{\mathsf Z}C_z\,\dd\nu(z)=\Id_m.
  \label{eq:povm-block-normalization}
\end{equation}

For each $X$, denote the likelihood density of the outcome of $\mathsf M$ on
$\rho_X$ by
\[
  q_X(z)\coloneqq\Tr(M_z\rho_X).
\]
In other words, the probability of an outcome in a measurable event
$B\in\mathcal Z$ is $\int_Bq_X(z)\,\dd\nu(z)$.  At the reference point
$X=0$, \cref{eq:reference-projector-state,eq:povm-blocks} give
\begin{equation}
  q_0(z)=\frac{\Tr(A_z)}{r}.
  \label{eq:reference-likelihood}
\end{equation}
For $H\in\C^{m\times r}$,
\cref{eq:povm-likelihood-derivative,eq:state-tangent-at-zero} give
\begin{equation}
  \Ddir_Hq_0(z)
  =\Tr\bigl(M_z\Ddir_H\rho_0\bigr)
  =\frac1r\left(\Tr(B_z^\dagger H)+\Tr(B_zH^\dagger)\right)
  =\frac2r\operatorname{Re}\left(\Tr(B_z^\dagger H)\right).
  \label{eq:one-copy-likelihood-derivative}
\end{equation}
We now calculate the numerator $\norm{\nabla q_0(z)}_2^2$ in the
Fisher information trace formula \cref{eq:fisher-trace}.  Here
$\nabla q_0(z)\in\R^{2mr}$ denotes the gradient of $X\mapsto q_X(z)$ at
$X=0$ in the coordinates associated with the basis in
\cref{eq:real-matrix-basis}.  For $1\le i\le m$ and $1\le j\le r$, we have
\(\Tr(B_z^\dagger E_{ij})=\overline{(B_z)_{ij}}\).
Substituting \cref{eq:one-copy-likelihood-derivative} in each coordinate gives
\begin{align}
  \norm{\nabla q_0(z)}_2^2
  &=\sum_{i=1}^m\sum_{j=1}^r
  \left((\Ddir_{E_{ij}}q_0(z))^2
    +(\Ddir_{\ii E_{ij}}q_0(z))^2\right) \notag\\
  &=\frac4{r^2}\sum_{i=1}^m\sum_{j=1}^r
    \left(
      \operatorname{Re}\left(\overline{(B_z)_{ij}}\right)^2
      +\operatorname{Re}\left(\ii\overline{(B_z)_{ij}}\right)^2
    \right) \notag\\
  &=\frac4{r^2}\sum_{i=1}^m\sum_{j=1}^r
      \abs{(B_z)_{ij}}^2
   =\frac4{r^2}\norm{B_z}_\F^2.
  \label{eq:summed-one-copy-derivatives}
\end{align}

Let $\cI_{\mathsf M}^{(1)}(0)$ denote the measurement Fisher information
matrix of the family $X\mapsto q_X$ at $X=0$, in the real coordinates
$\{E_{ij},\ii E_{ij}\}_{i,j}$.  The superscript $(1)$ indicates that
$\mathsf M$ is applied to one sample of the state.  Substituting
\cref{eq:reference-likelihood,eq:summed-one-copy-derivatives} into the
Fisher information trace formula \cref{eq:fisher-trace} yields the exact
identity
\begin{equation}
  \Tr\bigl(\cI_{\mathsf M}^{(1)}(0)\bigr)
  =\frac4r\int_{\mathsf Z}
    \frac{\norm{B_z}_\F^2}{\Tr(A_z)}\,\dd\nu(z).
  \label{eq:one-copy-fisher-block-expression}
\end{equation}
If $\Tr(A_z)=0$, then $A_z=0$, and positivity of $M_z$
forces $B_z=0$.  Hence the likelihood and its derivatives vanish, and the
quotient above is defined to be zero.

We next bound this quotient using the positivity of $M_z$.  Applying the
Cauchy--Schwarz inequality for the positive semidefinite form induced by
$M_z$ to the $j$-th coordinate vector of $S$ and the $i$-th coordinate
vector of $S^\perp$ gives
\[
  \abs{(B_z)_{ij}}^2\le(A_z)_{jj}(C_z)_{ii},
  \qquad i\in[m],\ j\in[r].
\]
Summing over these coordinates yields
\[
  \norm{B_z}_\F^2
  =\sum_{i=1}^m\sum_{j=1}^r\abs{(B_z)_{ij}}^2
  \le\sum_{i=1}^m\sum_{j=1}^r(A_z)_{jj}(C_z)_{ii}
  =\Tr(A_z)\Tr(C_z).
\]
When $\Tr(A_z)>0$, it gives
\[
  \frac{\norm{B_z}_\F^2}{\Tr(A_z)}\le\Tr(C_z).
\]
Substituting this inequality into
\cref{eq:one-copy-fisher-block-expression} and using
\cref{eq:povm-block-normalization}, we conclude that
\begin{align}
  \Tr\bigl(\cI_{\mathsf M}^{(1)}(0)\bigr)
  &\le\frac4r\int_{\mathsf Z}\Tr(C_z)\,\dd\nu(z) \notag\\
  &=\frac4r\Tr\left(\int_{\mathsf Z}C_z\,\dd\nu(z)\right) \notag\\
  &=\frac4r\Tr(\Id_m)
   \le\frac{4d}{r}.
  \label{eq:one-copy-fisher-bound}
\end{align}

\subsubsection{Joint measurements on multiple samples}
\label{sec:multi-copy-fisher-bound}

For \(t>1\), differentiating \(\rho_X^{\otimes t}\) produces a sum of
\(t\) terms, one for each tensor factor.  The Fisher information formula
squares this sum, so we must control the cross terms between different
positions.  We reduce these cross terms to an operator that counts how many
of the remaining \(t-1\) factors lie in the direction of a unit vector
\(\vct{h}\in S\), and then construct an upper bound independent of
\(\vct{h}\) with small trace.  This gives
the \(O(dt\sqrt t/r)\) bound when \(t\le r^2\).  When \(t>r^2\), the additive
quantum Fisher information bound gives the \(O(dt)\) estimate.

Fix an integer $t\ge2$.  Let $\vct{e}_1,\ldots,\vct{e}_r$ and
$\vct{f}_1,\ldots,\vct{f}_m$ denote the orthonormal basis
vectors of $S$ and $S^\perp$, respectively.  We use the same notation
$E_{ij}$ for the operator $|\vct{f}_i\rangle\langle\vct{e}_j|$ on $\C^d$,
which maps $S$ to $S^\perp$ and vanishes on $S^\perp$.
Recall from \cref{sec:support-rotation} that $\Pi_0$ is the projector onto
$S$.  The zero-excitation subspace and the reference state on $t$ samples are
\[
  \cH_0\coloneqq S^{\otimes t},
  \qquad
  \rho_0^{\otimes t}=\frac1{r^t}\Pi_0^{\otimes t}.
\]
For each $1\le i\le m$, define the corresponding one-excitation subspace
\[
  \cH_{1,i}
  \coloneqq
  \bigoplus_{\ell=1}^t
  S^{\otimes(\ell-1)}\otimes
  \operatorname{span}\{\vct{f}_i\}\otimes
  S^{\otimes(t-\ell)},
\]
and let $\Pi_{1,i}$ be the orthogonal projector onto $\cH_{1,i}$.  The summand
indexed by $\ell$ consists of the tensors whose $\ell$-th factor lies in
$\operatorname{span}\{\vct{f}_i\}$ and whose other factors lie in $S$.

For $1\le i\le m$ and $1\le j\le r$, define the collective raising operator
\[
  R_{i,j}
  \coloneqq
  \sum_{\ell=1}^t
  \Pi_0^{\otimes(\ell-1)}\otimes E_{ij}\otimes
  \Pi_0^{\otimes(t-\ell)}.
\]
This operator maps $\cH_0$ into $\cH_{1,i}$ and vanishes on
$\cH_0^\perp$.  Its adjoint $R_{i,j}^\dagger$ vanishes on
$\cH_{1,i}^\perp$.
For a direction $H\in\C^{m\times r}$, the product rule gives
\[
  \left.\frac{\dd}{\dd\lambda}
    \rho_{\lambda H}^{\otimes t}\right|_{\lambda=0}
  =\sum_{\ell=1}^t
  \rho_0^{\otimes(\ell-1)}\otimes
  \Ddir_H\rho_0\otimes
  \rho_0^{\otimes(t-\ell)}.
\]
Using $\rho_0=\Pi_0/r$ and \cref{eq:state-tangent-at-zero}, we obtain
\begin{align}
  \left.\frac{\dd}{\dd\lambda}
    \rho_{\lambda E_{ij}}^{\otimes t}\right|_{\lambda=0}
  &=\frac1{r^t}\bigl(R_{i,j}+R_{i,j}^\dagger\bigr), \notag\\
  \left.\frac{\dd}{\dd\lambda}
    \rho_{\lambda\ii E_{ij}}^{\otimes t}\right|_{\lambda=0}
  &=\frac{\ii}{r^t}\bigl(R_{i,j}-R_{i,j}^\dagger\bigr).
  \label{eq:t-copy-coordinate-tangents}
\end{align}
Thus the tangent directions above connect the zero-excitation subspace
only to the one-excitation subspaces.

Let $\mathsf M$ be an arbitrary POVM on $(\C^d)^{\otimes t}$ with outcome
space $(\mathsf Z,\mathcal Z)$.  Let $\nu$ and $M_z$ be its
parameter-independent dominating measure and operator density from
\cref{eq:povm-dominating-measure,eq:povm-density}.  For
$1\le i\le m$, define the blocks
\begin{equation}
  A_z\coloneqq\Pi_0^{\otimes t}M_z\Pi_0^{\otimes t},
  \qquad
  B_{z,i}\coloneqq\Pi_{1,i}M_z\Pi_0^{\otimes t},
  \qquad
  C_{z,i}\coloneqq\Pi_{1,i}M_z\Pi_{1,i}.
  \label{eq:t-copy-povm-blocks}
\end{equation}
Thus $A_z$ acts on $\cH_0$, $B_{z,i}$ maps $\cH_0$ to $\cH_{1,i}$, and
$C_{z,i}$ acts on $\cH_{1,i}$.  Multiplying the identity
$\int_{\mathsf Z}M_z\,\dd\nu(z)=\Id_d^{\otimes t}$ on the left and right
by the corresponding projectors gives
\begin{equation}
  \int_{\mathsf Z}A_z\,\dd\nu(z)=\Pi_0^{\otimes t},
  \qquad
  \int_{\mathsf Z}B_{z,i}\,\dd\nu(z)=0,
  \qquad
  \int_{\mathsf Z}C_{z,i}\,\dd\nu(z)=\Pi_{1,i}.
  \label{eq:t-copy-povm-block-normalization}
\end{equation}

For each $X$, define the likelihood density
\[
  q_X^{(t)}(z)
  \coloneqq\Tr\bigl(M_z\rho_X^{\otimes t}\bigr).
\]
At the reference point, $\rho_0^{\otimes t}=\Pi_0^{\otimes t}/r^t$, so
\begin{equation}
  q_0^{(t)}(z)=\frac{\Tr(A_z)}{r^t}.
  \label{eq:t-copy-reference-likelihood}
\end{equation}
Substituting
\cref{eq:povm-likelihood-derivative,eq:t-copy-coordinate-tangents} in each
coordinate gives
\begin{align}
  \norm{\nabla q_0^{(t)}(z)}_2^2
  &=\sum_{i=1}^m\sum_{j=1}^r
    \left(
      (\Ddir_{E_{ij}}q_0^{(t)}(z))^2
      +(\Ddir_{\ii E_{ij}}q_0^{(t)}(z))^2
    \right) \notag\\
  &=\frac4{r^{2t}}\sum_{i=1}^m\sum_{j=1}^r
    \left(
      \operatorname{Re}\left(
        \Tr(B_{z,i}^\dagger R_{i,j})
      \right)^2
      +\operatorname{Re}\left(
        \ii\Tr(B_{z,i}^\dagger R_{i,j})
      \right)^2
    \right) \notag\\
  &=\frac4{r^{2t}}\sum_{i=1}^m\sum_{j=1}^r
    \abs{\Tr(B_{z,i}^\dagger R_{i,j})}^2.
  \label{eq:summed-t-copy-derivatives}
\end{align}

Let $\cI_{\mathsf M}^{(t)}(0)$ denote the measurement Fisher information
matrix induced by $\mathsf M$ at $X=0$ in these coordinates.  Substituting
\cref{eq:t-copy-reference-likelihood,eq:summed-t-copy-derivatives} into
\cref{eq:fisher-trace} yields the exact identity
\begin{equation}
  \Tr\bigl(\cI_{\mathsf M}^{(t)}(0)\bigr)
  =\frac4{r^t}\int_{\mathsf Z}
    \frac{
      \displaystyle\sum_{i=1}^m\sum_{j=1}^r
      \abs{\Tr(B_{z,i}^\dagger R_{i,j})}^2
    }{\Tr(A_z)}
    \,\dd\nu(z).
  \label{eq:t-copy-fisher-block-expression}
\end{equation}
If $\Tr(A_z)=0$, then $A_z=0$, and positivity of the restriction of $M_z$ to
$\cH_0\oplus\cH_{1,i}$ forces $B_{z,i}=0$ for every $i$.  Hence the
likelihood and its derivatives vanish, and the quotient above is defined to
be zero.

For a unit vector
$\vct{h}=\sum_{j=1}^rh_j\vct{e}_j\in S$, define
\[
  R_{i,\vct{h}}\coloneqq\sum_{j=1}^r\overline{h_j}R_{i,j}.
\]
For fixed $z$ and $i$, the sum over the \(r\) coordinate directions can be
written as the largest value over a unit direction \(\vct{h}\in S\):
\begin{equation}
  \sum_{j=1}^r
    \abs{\Tr(B_{z,i}^\dagger R_{i,j})}^2
  =\sup_{\substack{\vct{h}\in S\\\norm{\vct{h}}=1}}
    \abs{\sum_{j=1}^r\overline{h_j}
      \Tr(B_{z,i}^\dagger R_{i,j})}^2
  =\sup_{\substack{\vct{h}\in S\\\norm{\vct{h}}=1}}
    \abs{\Tr(B_{z,i}^\dagger R_{i,\vct{h}})}^2.
  \label{eq:fisher-row-duality}
\end{equation}
The first equality follows from Cauchy--Schwarz and its equality condition.
The second uses the definition of $R_{i,\vct{h}}$.

To bound this supremum, we factor $M_z$ through its positive semidefinite
square root and apply the Frobenius Cauchy--Schwarz inequality.  For
$1\le i\le m$ and $\nu$-almost every $z$, define
\[
  Y_{z,0}\coloneqq\sqrt{M_z}\,\Pi_0^{\otimes t},
  \qquad
  Y_{z,1,i}\coloneqq\sqrt{M_z}\,\Pi_{1,i}.
\]
The block definitions in \cref{eq:t-copy-povm-blocks} give
\[
  A_z=Y_{z,0}^\dagger Y_{z,0},
  \qquad
  B_{z,i}=Y_{z,1,i}^\dagger Y_{z,0},
  \qquad
  C_{z,i}=Y_{z,1,i}^\dagger Y_{z,1,i}.
\]
For every unit vector $\vct{h}\in S$, Cauchy--Schwarz and cyclicity of the
trace give
\begin{align}
  \abs{\Tr(B_{z,i}^\dagger R_{i,\vct{h}})}^2
  &=\abs{\Tr(Y_{z,0}^\dagger Y_{z,1,i}R_{i,\vct{h}})}^2 \notag\\
  &\le\norm{Y_{z,0}}_\F^2
       \norm{Y_{z,1,i}R_{i,\vct{h}}}_\F^2 \notag\\
  &=\Tr(A_z)\Tr\bigl(
      C_{z,i}R_{i,\vct{h}}R_{i,\vct{h}}^\dagger
    \bigr).
  \label{eq:fisher-row-positive-block-bound}
\end{align}
Since $C_{z,i}\succeq0$, any positive semidefinite operator $Q$ on
$\cH_{1,i}$ satisfying
$R_{i,\vct{h}}R_{i,\vct{h}}^\dagger\preceq Q$ for every unit
$\vct{h}$ bounds the right side by
$\Tr(A_z)\Tr(C_{z,i}Q)$.  It therefore suffices to construct such a
direction-independent upper bound.
To construct this upper bound, we decompose the one-excitation space
according to the position of the excitation.  This reduces the problem to an
operator that counts how many of the remaining tensor factors lie in the
direction $\vct{h}$.

\paragraph{Reduction to occupation number operators.}

For $1\le\ell\le t$, use a hat to indicate the omitted tensor factor and let
\[
  \cK_{\widehat\ell}
  \coloneqq\bigotimes_{q\in[t]\setminus\{\ell\}}S_q,
\]
where $S_q$ denotes the copy of $S$ in the $q$-th tensor factor and the
factors retain their original order.  The map that inserts $\vct{f}_i$ in
position $\ell$ is an isometry from $\cK_{\widehat\ell}$ onto the $\ell$-th
summand of $\cH_{1,i}$, so
\[
  \cH_{1,i}
  \cong\bigoplus_{\ell=1}^t\cK_{\widehat\ell}.
\]
In the rest of the proof, we use these insertion isometries to write
operators on $\cH_{1,i}$ as block operators on this direct sum.

For each unit vector $\vct{h}\in S$, let
$\Pi_{\vct{h}}\coloneqq|\vct{h}\rangle\langle\vct{h}|$ be the orthogonal projector
onto $\operatorname{span}\{\vct{h}\}$.  For $\ell\in[t]$ and
$q\in[t]\setminus\{\ell\}$, define the operator on
$\cK_{\widehat\ell}$
\[
  \Pi_{\vct{h},q}^{(\widehat\ell)}
  \coloneqq
  \bigotimes_{p\in[t]\setminus\{\ell\}}
  \begin{cases}
    \Pi_{\vct{h}},&p=q,\\
    \Id_{S_p},&p\ne q.
  \end{cases}
\]
For each $\ell\in[t]$, define the occupation number operator on
$\cK_{\widehat\ell}$, the tensor product of the remaining $t-1$ factors, by
\begin{equation}
  N_{\vct{h},\widehat\ell}^{(t-1)}
  \coloneqq\sum_{q\in[t]\setminus\{\ell\}}
    \Pi_{\vct{h},q}^{(\widehat\ell)}.
  \label{eq:omitted-factor-occupation-number}
\end{equation}
Extend $\vct{h}$ to an orthonormal basis
$\vct{u}_1=\vct{h},\vct{u}_2,\ldots,\vct{u}_r$ of $S$.  The vectors
\[
  \vct{v}_{\vct{j}}
  \coloneqq\bigotimes_{q\in[t]\setminus\{\ell\}}\vct{u}_{j_q},
  \qquad j_q\in[r],
\]
form an orthonormal basis of $\cK_{\widehat\ell}$.  Each projector
$\Pi_{\vct{h},q}^{(\widehat\ell)}$ leaves $\vct{v}_{\vct{j}}$ unchanged
when $j_q=1$ and sends it to zero otherwise.  Consequently,
\[
  N_{\vct{h},\widehat\ell}^{(t-1)}\vct{v}_{\vct{j}}
  =\bigl|\{q\in[t]\setminus\{\ell\}:j_q=1\}\bigr|
    \,\vct{v}_{\vct{j}}.
\]
Thus $N_{\vct{h},\widehat\ell}^{(t-1)}$ is diagonal in this basis.  Its
eigenvalue on $\vct{v}_{\vct{j}}$ counts the tensor factors
$\vct{u}_{j_q}$ equal to $\vct{h}$.

\begin{lemma}[Reduction to occupation number operators]
  \label{lem:raising-map-occupation-reduction}
  For every $1\le i\le m$ and every unit vector $\vct{h}\in S$,
  \begin{equation}
    R_{i,\vct{h}}R_{i,\vct{h}}^\dagger
    \preceq
    \bigoplus_{\ell=1}^t
      \left(
        \Id_{\cK_{\widehat\ell}}
        +N_{\vct{h},\widehat\ell}^{(t-1)}
      \right).
    \label{eq:raising-map-occupation-reduction}
  \end{equation}
\end{lemma}

\begin{proof}
  Fix $i$ and $\vct{h}$ as in the statement.  Using the preceding direct sum
  representation, write
  $\vct{\xi}=(\vct{\xi}_1,\ldots,\vct{\xi}_t)\in\cH_{1,i}$, where
  $\vct{\xi}_\ell\in\cK_{\widehat\ell}$.  Thus $\vct{\xi}_\ell$ represents
  the component obtained by inserting $\vct{f}_i$ in position $\ell$.
  Let $J_{\ell,\vct{h}}:\cK_{\widehat\ell}\to\cH_0$ be the isometry that
  instead inserts $\vct{h}$ in position $\ell$.
  The $q$-th term of $R_{i,\vct{h}}^\dagger$ applies
  $|\vct{h}\rangle\langle\vct{f}_i|$ to tensor factor $q$.  On the vector
  obtained from $\vct{\xi}_\ell$ by inserting $\vct{f}_i$ in position
  $\ell$, every tensor factor
  other than $\ell$ lies in $S$.  Thus, if $q\ne\ell$, this term vanishes
  because $\vct{f}_i\perp S$.  The term with $q=\ell$ instead replaces
  $\vct{f}_i$ by $\vct{h}$.  Therefore,
  \[
    R_{i,\vct{h}}^\dagger\vct{\xi}
    =\sum_{\ell=1}^tJ_{\ell,\vct{h}}\vct{\xi}_\ell.
  \]

  On $\cH_0$, the operator
  $J_{\ell,\vct{h}}J_{\ell,\vct{h}}^\dagger=\Pi_{\vct{h}}^{[\ell]}$
  is the orthogonal projector onto the range of $J_{\ell,\vct{h}}$.
  For distinct $\ell,q\in[t]$, the projectors
  $\Pi_{\vct{h}}^{[\ell]}$ and $\Pi_{\vct{h}}^{[q]}$ commute, so
  \[
    \ip{J_{\ell,\vct{h}}\vct{\xi}_\ell}
      {J_{q,\vct{h}}\vct{\xi}_q}
    =\ip{\Pi_{\vct{h}}^{[q]}J_{\ell,\vct{h}}\vct{\xi}_\ell}
      {\Pi_{\vct{h}}^{[\ell]}J_{q,\vct{h}}\vct{\xi}_q}
    =\ip{J_{\ell,\vct{h}}
      \Pi_{\vct{h},q}^{(\widehat\ell)}\vct{\xi}_\ell}
      {J_{q,\vct{h}}
      \Pi_{\vct{h},\ell}^{(\widehat q)}\vct{\xi}_q}.
  \]
  Since the insertion maps are isometries, Cauchy--Schwarz and
  $2xy\le x^2+y^2$ give
  \[
    2\operatorname{Re}\ip{J_{\ell,\vct{h}}\vct{\xi}_\ell}
      {J_{q,\vct{h}}\vct{\xi}_q}
    \le
      2\norm{\Pi_{\vct{h},q}^{(\widehat\ell)}\vct{\xi}_\ell}
       \norm{\Pi_{\vct{h},\ell}^{(\widehat q)}\vct{\xi}_q}
    \le
      \ip{\vct{\xi}_\ell}
        {\Pi_{\vct{h},q}^{(\widehat\ell)}\vct{\xi}_\ell}
      +\ip{\vct{\xi}_q}
        {\Pi_{\vct{h},\ell}^{(\widehat q)}\vct{\xi}_q}.
  \]

  We now bound the quadratic form of
  $R_{i,\vct{h}}R_{i,\vct{h}}^\dagger$ on $\vct{\xi}$:
  \begin{alignat*}{2}
    \ip{\vct{\xi}}{
      R_{i,\vct{h}}R_{i,\vct{h}}^\dagger\vct{\xi}}
    &=\norm{R_{i,\vct{h}}^\dagger\vct{\xi}}^2
      &\quad&=\sum_{\ell=1}^t\norm{\vct{\xi}_\ell}^2
      +2\sum_{\ell<q}
        \operatorname{Re}\ip{J_{\ell,\vct{h}}\vct{\xi}_\ell}
          {J_{q,\vct{h}}\vct{\xi}_q}\\
    &&&\le\sum_{\ell=1}^t\left(
      \norm{\vct{\xi}_\ell}^2
      +\sum_{q\ne\ell}
        \ip{\vct{\xi}_\ell}{
          \Pi_{\vct{h},q}^{(\widehat\ell)}\vct{\xi}_\ell}
      \right)\\
    &&&=\sum_{\ell=1}^t
      \ip{\vct{\xi}_\ell}{
        \bigl(\Id_{\cK_{\widehat\ell}}
          +N_{\vct{h},\widehat\ell}^{(t-1)}\bigr)
          \vct{\xi}_\ell}.
  \end{alignat*}
  Here the inequality follows by summing the cross term bound over all
  pairs $\ell<q$ and regrouping the terms by excitation position.  For each
  fixed $\ell$, every $q\ne\ell$ appears once.
  Since this holds for every $\vct{\xi}\in\cH_{1,i}$, it proves
  \cref{eq:raising-map-occupation-reduction}.
\end{proof}

\paragraph{A uniform bound on occupation number operators.}

Each operator in \cref{eq:omitted-factor-occupation-number} has the same form
on a tensor product of $t-1$ copies of $S$.  We state the required
direction-independent bound for a general number $n$ of factors.  For every
positive integer $n$ and every unit vector $\vct{h}\in S$, define
\[
  N_{\vct{h}}^{(n)}
  \coloneqq\sum_{\ell=1}^n\Pi_{\vct{h}}^{[\ell]}
  =\sum_{\ell=1}^n\Id_S^{\otimes(\ell-1)}\otimes
    \Pi_{\vct{h}}\otimes\Id_S^{\otimes(n-\ell)}.
\]
We now construct a positive semidefinite operator \(T_{n,r}\), independent of
\(\vct{h}\), such that \(N_{\vct{h}}^{(n)}\preceq T_{n,r}\) for every unit
\(\vct{h}\in S\), while its normalized trace remains small.

\begin{lemma}[Uniform bound on occupation number operators]
  \label{lem:occupation-number-envelope}
  There is a universal constant \(C>0\) such that, for all positive integers $n,r$
  and every \(r\)-dimensional complex Hilbert space \(S\), there exists a
  positive semidefinite operator
  $T_{n,r}$ on $S^{\otimes n}$ such that
  \begin{equation}
    N_{\vct{h}}^{(n)}\preceq T_{n,r}
    \qquad
    \text{for every unit vector $\vct{h}\in S$},
    \label{eq:occupation-number-domination}
  \end{equation}
  and
  \begin{equation}
    \frac{\Tr(T_{n,r})}{r^n}
    \le C\left(\sqrt n+\frac nr\right).
    \label{eq:occupation-envelope-trace}
  \end{equation}
\end{lemma}

\begin{proof}
  We obtain an upper bound independent of $\vct{h}$ by bounding a suitable
  polynomial in $N_{\vct{h}}^{(n)}$ above by a sum of projectors onto symmetric
  subspaces.  Fix an integer $1\le k\le n$.  For a subset $J\subseteq[n]$
  with $|J|=k$, let
  $\mathfrak S_J$ be the set of permutations of $J$.  For each
  $\pi\in\mathfrak S_J$, let $U_\pi$ permute the corresponding tensor factors
  while fixing the others.  The orthogonal projector onto the subspace
  invariant under these permutations is
  \[
    \Pi_{\mathrm{sym},J}
    \coloneqq\frac1{k!}\sum_{\pi\in\mathfrak S_J}U_\pi.
  \]
  Set
  \[
    \Omega_{n,k}
    \coloneqq
    \sum_{\substack{J\subseteq[n]\\|J|=k}}
      \Pi_{\mathrm{sym},J}.
  \]
  The operator $\Omega_{n,k}$ is positive semidefinite and does not depend on
  $\vct{h}$.

  Since the projectors $\Pi_{\vct{h}}^{[\ell]}$ commute,
  \[
    \binom{N_{\vct{h}}^{(n)}}k
    \coloneqq\frac1{k!}
      \prod_{s=0}^{k-1}
      \bigl(N_{\vct{h}}^{(n)}-s\Id_S^{\otimes n}\bigr)
    =\sum_{\substack{J\subseteq[n]\\|J|=k}}
      \prod_{\ell\in J}\Pi_{\vct{h}}^{[\ell]}.
  \]
  In particular, this binomial coefficient operator is positive semidefinite.
  For a fixed $J$, the product on the right projects onto the tensors whose
  factors in $J$ all lie in $\operatorname{span}\{\vct{h}\}$.  Such tensors
  are invariant under every permutation of those factors, so
  \[
    \prod_{\ell\in J}\Pi_{\vct{h}}^{[\ell]}
    \preceq\Pi_{\mathrm{sym},J}.
  \]
  Summing this inequality over $J$ gives
  \[
    \binom{N_{\vct{h}}^{(n)}}k\preceq\Omega_{n,k}.
  \]

  We next recover a bound on $N_{\vct{h}}^{(n)}$ from the bound on
  $\binom{N_{\vct{h}}^{(n)}}{k}$.
  Fix $\vct{h}$ and extend it to an orthonormal basis of $S$.  Each vector
  $\vct{b}_1\otimes\cdots\otimes\vct{b}_n$ in the associated tensor product
  basis is an eigenvector of $N_{\vct{h}}^{(n)}$, with eigenvalue equal to the
  number of indices $\ell\in[n]$ for which $\vct{b}_\ell=\vct{h}$.  Thus every
  eigenvalue belongs to $\{0,1,\ldots,n\}$.  For every such eigenvalue $x$,
  \[
    x\le k+k\binom{x}{k}^{1/k}.
  \]
  When $x<k$, the first term on the right already bounds $x$; when $x\ge k$,
  the inequality follows from $\binom{x}{k}\ge(x/k)^k$.  Applying it to
  every eigenvalue of $N_{\vct{h}}^{(n)}$ gives
  \[
    N_{\vct{h}}^{(n)}
    \preceq
    k\Id_S^{\otimes n}
      +k\left(\binom{N_{\vct{h}}^{(n)}}{k}\right)^{1/k}.
  \]
  For every $\pi\in\mathfrak S_n$, conjugation permutes the summands:
  \[
    U_\pi N_{\vct{h}}^{(n)}U_\pi^\dagger
    =\sum_{\ell=1}^n\Pi_{\vct{h}}^{[\pi(\ell)]}
    =N_{\vct{h}}^{(n)}.
  \]
  Thus $N_{\vct{h}}^{(n)}$ commutes with every tensor permutation $U_\pi$
  and hence with $\Omega_{n,k}$.  Taking positive
  semidefinite $k$-th roots in a common eigenbasis, the binomial coefficient
  operator inequality above gives
  \begin{equation}
    N_{\vct{h}}^{(n)}
    \preceq k\Id_S^{\otimes n}+k\Omega_{n,k}^{1/k}.
    \label{eq:occupation-envelope-fixed-k}
  \end{equation}
  Here $\Omega_{n,k}^{1/k}$ denotes the positive semidefinite $k$-th root of
  $\Omega_{n,k}$.

  It remains to choose $k$ so that the operator on the right has small trace.
  Since the symmetric subspace of $S^{\otimes k}$ has dimension
  $\binom{r+k-1}{k}$, each $\Pi_{\mathrm{sym},J}$ has trace
  $\binom{r+k-1}{k}r^{n-k}$.  Therefore
  \[
    \frac{\Tr(\Omega_{n,k})}{r^n}
    =\binom nk\frac{\binom{r+k-1}{k}}{r^k}.
  \]
  Concavity of $x\mapsto x^{1/k}$, Jensen's inequality, and the standard
  binomial estimate give, for a universal constant \(C'>0\),
  \[
    \frac{\Tr(\Omega_{n,k}^{1/k})}{r^n}
    \le\left(
      \frac{\Tr(\Omega_{n,k})}{r^n}
    \right)^{1/k}
    \le\frac{\ee n}{k}\frac{\ee(r+k-1)}{kr}
    \le C'\left(\frac{n}{k^2}+\frac{n}{kr}\right).
  \]
  Here the second inequality uses
  $\binom uk\le(\ee u/k)^k$ for integers $u\ge k\ge1$.
  Choose $k\coloneqq\lceil\sqrt n\rceil$ and define
  \[
    T_{n,r}
    \coloneqq k\Id_S^{\otimes n}+k\Omega_{n,k}^{1/k}.
  \]
  The inequality in \cref{eq:occupation-number-domination} follows from
  \cref{eq:occupation-envelope-fixed-k}, while
  \[
    \frac{\Tr(T_{n,r})}{r^n}
    \le k+C'\left(\frac nk+\frac nr\right)
    \le C\left(\sqrt n+\frac nr\right).
  \]
  This proves \cref{eq:occupation-envelope-trace}.
\end{proof}

\paragraph{A uniform bound on the one-excitation subspace.}

We apply the occupation number bound to each summand
$\cK_{\widehat\ell}$ in the direct sum representation of $\cH_{1,i}$.
This produces a direction-independent upper
bound on $R_{i,\vct{h}}R_{i,\vct{h}}^\dagger$ for use in
\cref{eq:fisher-row-positive-block-bound}.

\begin{lemma}[Uniform bound on the one-excitation subspace]
  \label{lem:one-excitation-envelope}
  There is a universal constant \(C>0\) such that, for every
  $1\le i\le m$, there is a positive semidefinite operator
  $Q_{t,r}^{(i)}$ on $\cH_{1,i}$ such that
  \begin{equation}
    R_{i,\vct{h}}R_{i,\vct{h}}^\dagger\preceq Q_{t,r}^{(i)}
    \qquad
    \text{for every unit vector $\vct{h}\in S$},
    \label{eq:one-excitation-domination}
  \end{equation}
  and
  \begin{equation}
    \Tr(Q_{t,r}^{(i)})
    \le Ct r^{t-1}
      \left(\sqrt t+\frac tr\right).
    \label{eq:one-excitation-envelope-trace}
  \end{equation}
\end{lemma}

\begin{proof}
  Fix $1\le i\le m$.  For each $\ell\in[t]$, let
  $T_{t-1,r}^{(\widehat\ell)}$ be the operator on
  $\cK_{\widehat\ell}$ obtained by placing the tensor factors of
  $T_{t-1,r}$ from \Cref{lem:occupation-number-envelope} in the positions
  $[t]\setminus\{\ell\}$, in their original order.  Define
  \[
    Q_{t,r}^{(i)}
    \coloneqq\bigoplus_{\ell=1}^t
      \left(
        \Id_{\cK_{\widehat\ell}}
        +T_{t-1,r}^{(\widehat\ell)}
      \right).
  \]
  This operator does not depend on $\vct{h}$.  Under the same placement of
  tensor factors, $N_{\vct{h}}^{(t-1)}$ becomes
  $N_{\vct{h},\widehat\ell}^{(t-1)}$.  Therefore,
  \Cref{lem:occupation-number-envelope} gives
  \[
    N_{\vct{h},\widehat\ell}^{(t-1)}
    \preceq T_{t-1,r}^{(\widehat\ell)}
    \qquad
    \text{for every $\ell\in[t]$ and every unit $\vct{h}\in S$}.
  \]
  Combining these inequalities with
  \Cref{lem:raising-map-occupation-reduction} proves
  \cref{eq:one-excitation-domination}.  Finally, by
  \Cref{lem:occupation-number-envelope},
  \[
    \Tr(Q_{t,r}^{(i)})
    =t\left(r^{t-1}+\Tr(T_{t-1,r})\right)
    \le Ct r^{t-1}\left(\sqrt t+\frac tr\right),
  \]
  which proves \cref{eq:one-excitation-envelope-trace}.
\end{proof}

\paragraph{Completion of the Fisher bound for joint measurements.}

For each $1\le i\le m$, we now use the direction-independent operator
$Q_{t,r}^{(i)}$ to bound the numerator in
\cref{eq:t-copy-fisher-block-expression}.  Since
$C_{z,i}\succeq0$, combining
\cref{eq:fisher-row-positive-block-bound,eq:one-excitation-domination}
and taking the supremum over $\vct{h}$ as in \cref{eq:fisher-row-duality}
gives
\begin{equation}
  \sum_{j=1}^r
    \abs{\Tr(B_{z,i}^\dagger R_{i,j})}^2
  \le\Tr(A_z)\Tr\bigl(C_{z,i}Q_{t,r}^{(i)}\bigr).
  \label{eq:fisher-row-envelope-bound}
\end{equation}
Substituting \cref{eq:fisher-row-envelope-bound} into
\cref{eq:t-copy-fisher-block-expression} and using the POVM normalization
in \cref{eq:t-copy-povm-block-normalization} and the trace bound
\cref{eq:one-excitation-envelope-trace}, we obtain
\begin{align}
  \Tr\bigl(\cI_{\mathsf M}^{(t)}(0)\bigr)
  &\le\frac4{r^t}\sum_{i=1}^m
    \int_{\mathsf Z}
      \Tr\bigl(C_{z,i}Q_{t,r}^{(i)}\bigr)\,\dd\nu(z) \notag\\
  &=\frac4{r^t}\sum_{i=1}^m
    \Tr\left(
      \left(\int_{\mathsf Z}C_{z,i}\,\dd\nu(z)\right)
      Q_{t,r}^{(i)}
    \right) \notag\\
  &=\frac4{r^t}\sum_{i=1}^m\Tr(Q_{t,r}^{(i)}) \notag\\
  &\le \frac4{r^t}\sum_{i=1}^m
    Ct r^{t-1}\left(\sqrt t+\frac tr\right) \notag\\
  &\le4C\frac{dt}{r}
    \left(\sqrt t+\frac tr\right).
  \label{eq:t-copy-fisher-envelope-bound}
\end{align}

For $t>r^2$, we instead use the general quantum Fisher information bound from
the preliminaries, which gives the required linear bound in this regime.
Recall that a symmetric logarithmic derivative for a tangent direction $H$
is a Hermitian operator $L_H$ satisfying
$\Ddir_H\rho_0=(\rho_0L_H+L_H\rho_0)/2$.  At $X=0$, take the symmetric
logarithmic derivatives for the real and imaginary tangent directions
$E_{ij}$ and $\ii E_{ij}$ to be
\[
  L_{i,j}^{\mathrm R}
  \coloneqq2(E_{ij}+E_{ij}^\dagger),
  \qquad
  L_{i,j}^{\mathrm I}
  \coloneqq2\ii(E_{ij}-E_{ij}^\dagger).
\]
Indeed, substituting $\rho_0=\Pi_0/r$ gives
\[
  \frac12\bigl(\rho_0L_{i,j}^{\mathrm R}
    +L_{i,j}^{\mathrm R}\rho_0\bigr)
  =\Ddir_{E_{ij}}\rho_0,
  \qquad
  \frac12\bigl(\rho_0L_{i,j}^{\mathrm I}
    +L_{i,j}^{\mathrm I}\rho_0\bigr)
  =\Ddir_{\ii E_{ij}}\rho_0,
\]
as required by \cref{eq:symmetric-logarithmic-derivative}.  Also,
\[
  (L_{i,j}^{\mathrm R})^2
  =(L_{i,j}^{\mathrm I})^2
  =4(E_{ij}^\dagger E_{ij}+E_{ij}E_{ij}^\dagger).
\]
Since $\rho_0=\Pi_0/r$ is supported on $S$, only
$E_{ij}^\dagger E_{ij}=|\vct{e}_j\rangle\langle\vct{e}_j|$ contributes to
the trace.  Thus each of the $2mr$ coordinates contributes
\[
  \Tr\bigl(\rho_0(L_{i,j}^{\mathrm R})^2\bigr)
  =\Tr\bigl(\rho_0(L_{i,j}^{\mathrm I})^2\bigr)
  =\frac4r
\]
to the trace of the single-sample quantum Fisher information matrix.  Hence
\[
  \Tr\bigl(\cI_{\mathrm Q}(0)\bigr)=8m\le8d.
\]
The classical--quantum Fisher inequality and tensor product additivity in
\cref{eq:joint-measurement-quantum-fisher-bound} therefore give
\begin{equation}
  \Tr\bigl(\cI_{\mathsf M}^{(t)}(0)\bigr)\le8dt.
  \label{eq:t-copy-quantum-fisher-bound}
\end{equation}

Combining \cref{eq:t-copy-fisher-envelope-bound} for $t\le r^2$ with
\cref{eq:t-copy-quantum-fisher-bound} for $t>r^2$, there is a universal
constant \(C'>0\) such that
\begin{equation}
  \Tr\bigl(\cI_{\mathsf M}^{(t)}(0)\bigr)
  \le C'dt\min\left\{1,\frac{\sqrt t}{r}\right\},
  \label{eq:t-copy-fisher-bound}
\end{equation}
for every joint POVM $\mathsf M$ on $t$ samples.

\paragraph{Transfer to an arbitrary parameter.}

The preceding bounds apply only at \(X=0\).  To prove the bound in
\Cref{prop:one-block-fisher-bound} at an arbitrary $X$ for every \(t\ge1\),
we construct a unitary $U_X$ such that
\[
  U_X^\dagger\rho_XU_X=\rho_0.
\]
We show that conjugation sends a perturbation $H$ at $X$ to a perturbation
$\mathcal L_X(H)$ at $0$, where
\[
  \norm{\mathcal L_X(H)}_\F\le\norm H_\F.
\]
We then use these facts to bound the Fisher information trace at $X$ by that
of the conjugated measurement at $0$.

Recall the isometry $V_X$ from \cref{eq:graph-isometry}, whose range is the
support of $\rho_X$.  To construct $U_X$, define
\[
  W_X
  \coloneqq
  \begin{pmatrix}
    -X^\dagger\\
    \Id_m
  \end{pmatrix}
  (\Id_m+XX^\dagger)^{-1/2}.
\]
Direct multiplication gives
\[
  W_X^\dagger W_X=\Id_m,
  \qquad
  V_X^\dagger W_X=0.
\]
Thus
\[
  U_X\coloneqq\begin{pmatrix}V_X&W_X\end{pmatrix}
\]
is unitary.  Its first $r$ columns span the support of $\rho_X$, and therefore
$U_X^\dagger\rho_XU_X=\rho_0$.

To compare the likelihood derivatives at $X$ and $0$, for each
$H\in\C^{m\times r}$, we seek a direction $\mathcal L_X(H)$ satisfying
\[
  U_X^\dagger(\Ddir_H\rho_X)U_X
  =\Ddir_{\mathcal L_X(H)}\rho_0.
\]
Since $\rho_X=\Pi_X/r$, \Cref{eq:state-tangent-at-zero} shows that the desired
identity is equivalent to
\[
  U_X^\dagger(\Ddir_H\Pi_X)U_X
  =\begin{pmatrix}
    0&\mathcal L_X(H)^\dagger\\
    \mathcal L_X(H)&0
  \end{pmatrix}.
\]
This dictates the construction of $\mathcal L_X(H)$: we show that the
diagonal blocks vanish and take the lower left block to be $\mathcal L_X(H)$.
Since $\Pi_{X+sH}$ is a projector for every $s$, differentiating
$\Pi_{X+sH}^2=\Pi_{X+sH}$ at $s=0$ gives
\[
  (\Ddir_H\Pi_X)\Pi_X+\Pi_X(\Ddir_H\Pi_X)=\Ddir_H\Pi_X.
\]
Multiplying this identity on the left and right by $\Pi_X$, and
separately by $\Id_d-\Pi_X$, gives
\[
  2\Pi_X(\Ddir_H\Pi_X)\Pi_X=\Pi_X(\Ddir_H\Pi_X)\Pi_X,
  \qquad
  0=(\Id_d-\Pi_X)(\Ddir_H\Pi_X)(\Id_d-\Pi_X).
\]
The columns of $V_X$ and $W_X$ span the ranges of $\Pi_X$ and $\Id_d-\Pi_X$,
respectively, so the diagonal blocks vanish:
\[
  U_X^\dagger(\Ddir_H\Pi_X)U_X
  =\begin{pmatrix}
    0&V_X^\dagger(\Ddir_H\Pi_X)W_X\\
    W_X^\dagger(\Ddir_H\Pi_X)V_X&0
  \end{pmatrix}.
\]
It remains to compute the lower left block.  The definition of $W_X$ gives
\[
  W_X^\dagger\begin{pmatrix}\Id_r\\X\end{pmatrix}
  =(\Id_m+XX^\dagger)^{-1/2}
    \begin{pmatrix}-X&\Id_m\end{pmatrix}
    \begin{pmatrix}\Id_r\\X\end{pmatrix}
  =0.
\]
Using this identity together with $V_X^\dagger V_X=\Id_r$ and
$W_X^\dagger V_X=0$, differentiating $\Pi_X=V_XV_X^\dagger$ and the explicit
formula for $V_X$ gives
\begin{align*}
  W_X^\dagger(\Ddir_H\Pi_X)V_X
  &=W_X^\dagger\left((\Ddir_HV_X)V_X^\dagger
    +V_X(\Ddir_HV_X)^\dagger\right)V_X \notag\\
  &=W_X^\dagger(\Ddir_HV_X) \notag\\
  &=W_X^\dagger\left[
    \begin{pmatrix}0\\H\end{pmatrix}
    (\Id_r+X^\dagger X)^{-1/2}
    +\begin{pmatrix}\Id_r\\X\end{pmatrix}
    \Ddir_H\!\left((\Id_r+X^\dagger X)^{-1/2}\right)
  \right] \notag\\
  &=W_X^\dagger\begin{pmatrix}0\\H\end{pmatrix}
    (\Id_r+X^\dagger X)^{-1/2} \notag\\
  &=(\Id_m+XX^\dagger)^{-1/2}
    H(\Id_r+X^\dagger X)^{-1/2}.
\end{align*}
Define the real-linear map
\[
  \mathcal L_X(H)
  \coloneqq
  (\Id_m+XX^\dagger)^{-1/2}
  H(\Id_r+X^\dagger X)^{-1/2}.
\]
Both matrices multiplying $H$ have operator norm at most one, so
\[
  \norm{\mathcal L_X(H)}_\F\le\norm H_\F.
\]
Combining these identities gives
\[
  U_X^\dagger(\Ddir_H\rho_X)U_X
  =\frac1r
  \begin{pmatrix}
    0&\mathcal L_X(H)^\dagger\\
    \mathcal L_X(H)&0
  \end{pmatrix}
  =\Ddir_{\mathcal L_X(H)}\rho_0.
\]

Let $\widetilde{\mathsf M}$ be the conjugated POVM with density
\[
  \widetilde M_z
  \coloneqq
  (U_X^{\otimes t})^\dagger M_zU_X^{\otimes t},
\]
and let
\[
  \widetilde q_K^{(t)}(z)
  \coloneqq
  \Tr\left(\widetilde M_z\rho_K^{\otimes t}\right)
\]
be its likelihood for \(K\in\C^{m\times r}\), with \(U_X\) held fixed as \(K\)
varies.  The preceding
state identities give, for every $H$,
\[
  q_X^{(t)}(z)=\widetilde q_0^{(t)}(z),
  \qquad
  \Ddir_Hq_X^{(t)}(z)
  =\Ddir_{\mathcal L_X(H)}\widetilde q_0^{(t)}(z).
\]
The derivative identity says that a direction $H$ at $X$ changes the
likelihood in the same way as the direction $\mathcal L_X(H)$ at $0$.
By the definitions of the gradient and the adjoint, this gives
\(\nabla q_X^{(t)}(z)=\mathcal L_X^\dagger
  \nabla\widetilde q_0^{(t)}(z)\)
in real coordinates, where $\mathcal L_X^\dagger$ is the adjoint with respect
to the real Frobenius inner product.  Since
\(q_X^{(t)}(z)=\widetilde q_0^{(t)}(z)\), the definition of Fisher information
gives
\[
  \cI_{\mathsf M}^{(t)}(X)
  =\int
    \frac{\mathcal L_X^\dagger
      \nabla\widetilde q_0^{(t)}(z)
      \nabla\widetilde q_0^{(t)}(z)^{\mathsf T}\mathcal L_X}
      {\widetilde q_0^{(t)}(z)}\,\dd\nu(z)
  =\mathcal L_X^\dagger
    \cI_{\widetilde{\mathsf M}}^{(t)}(0)\mathcal L_X.
\]
Since $\mathcal L_X$ is a contraction,
$\mathcal L_X\mathcal L_X^\dagger\preceq\Id_{2mr}$, and hence
\[
  \Tr\bigl(\cI_{\mathsf M}^{(t)}(X)\bigr)
  =\Tr\bigl(\cI_{\widetilde{\mathsf M}}^{(t)}(0)
      \mathcal L_X\mathcal L_X^\dagger\bigr)
  \le\Tr\bigl(\cI_{\widetilde{\mathsf M}}^{(t)}(0)\bigr).
\]
\begin{proof}[Proof of \Cref{prop:one-block-fisher-bound}]
  At \(X=0\), the case \(t=1\) is
  \cref{eq:one-copy-fisher-bound}, and the case \(t\ge2\) is
  \cref{eq:t-copy-fisher-bound}.  For an arbitrary \(X\), the trace inequality
  above bounds \(\Tr(\cI_{\mathsf M}^{(t)}(X))\) by the corresponding quantity
  for the conjugated measurement at \(0\).  Applying the \(X=0\) bounds proves
  \cref{eq:one-block-fisher-bound}.
\end{proof}

\subsection{Adaptive accumulation across measurement rounds}
\label{sec:adaptive-information}

\Cref{prop:one-block-fisher-bound} controls the joint measurement selected in
one round.  To bound the full transcript of an adaptive protocol, represent
the protocol as a decision tree: a complete transcript is a root-to-leaf
path, and at the node corresponding to history \(h\) before round \(i\), the
protocol has selected a POVM acting jointly on \(t_i(h)\le t\) fresh samples.
We use the standard conditional score chain rule for Fisher information
along this path and then sum the resulting round-by-round bounds using the
pathwise sample bound.

Fix a protocol with pathwise sample bound \(N\).  Every nontrivial round uses at
least one sample, so there are at most \(N\) such rounds on any execution.  We
give every transcript a fixed length by appending dummy rounds after the
protocol halts until it has \(N\) rounds.  A dummy round uses no samples
and contributes no Fisher information, so this does not change the protocol.
Let \(U\) be its private seed, let \(Z_i\) be the outcome in round \(i\), and
write \(Z_{<i}\coloneqq(Z_1,\ldots,Z_{i-1})\).  At
history \(h=(U,Z_{<i})\), the protocol selects an integer
\(t_i(h)\in\{0,1,\ldots,t\}\) and a joint POVM on \(t_i(h)\) samples.  Let
\(p_{i,X}(\cdot\mid h)\) denote its conditional outcome density relative to
a parameter-independent reference measure \(\nu_i(\cdot\mid h)\) when the
input state is \(\rho_X\).

For finite outcomes, the transcript likelihood is obtained by multiplying
the conditional probabilities.  The same factorization holds for the
continuous outcome spaces in our model:
\Cref{app:statistical-formalism} constructs jointly measurable
conditional densities with respect to a parameter-independent transcript
measure and verifies the likelihood regularity used below.  Thus
\begin{equation}
  q_{N,X}(u,z_1,\ldots,z_N)
  \coloneqq\prod_{i=1}^N p_{i,X}(z_i\mid u,z_{<i}).
  \label{eq:full-transcript-likelihood}
\end{equation}
Under this transcript law, \(U,Z_1,\ldots,Z_N\) are the random seed and
outcomes.  We write \(\E_X\) for expectation when the unknown state is
\(\rho_X\).

Let \(\cI_i(h;X)\) be the Fisher information matrix contributed by the
round \(i\) measurement selected at history \(h\), evaluated at \(X\), and let
\(\cI_{\mathrm{tr}}^{(N)}(X)\) be the Fisher information matrix of the full
transcript.  Both are expressed in the same \(2mr\) real support coordinates,
so their matrices can be added.
Here the subscript \(\mathrm{tr}\) stands for transcript, while the
superscript \((N)\) denotes the number of padded rounds.

\begin{lemma}[Adaptive Fisher chain rule]
  \label{lem:block-fisher-chain}
  For every \(X\),
  \begin{equation}
    \cI_{\mathrm{tr}}^{(N)}(X)
    =
    \sum_{i=1}^N
      \E_X\!\left[\cI_i(U,Z_{<i};X)\right].
    \label{eq:block-fisher-chain}
  \end{equation}
\end{lemma}

\begin{proof}
  Let \(H_i\coloneqq(U,Z_{<i})\) be the history before round \(i\).  With
  \(\nabla_X\) denoting the gradient in the \(2mr\) real coordinates of
  \(X\), define the round \(i\) score by
  \[
    S_i(X)
    \coloneqq
    \nabla_X\log p_{i,X}(Z_i\mid H_i).
  \]
  As in \Cref{sec:classical-fisher-information}, the score is set to zero
  when the conditional likelihood vanishes.
  The likelihood factorization and the regularity established in
  \Cref{lem:block-likelihood-regularity} give
  \[
    \nabla_X\log q_{N,X}
    =
    \sum_{i=1}^N S_i(X),
  \]
  almost surely under the transcript law.
  By definition, the Fisher matrix of the transcript is therefore
  \begin{align*}
    \cI_{\mathrm{tr}}^{(N)}(X)
    &=\E_X\!\left[
      \left(\sum_{i=1}^NS_i(X)\right)
      \left(\sum_{j=1}^NS_j(X)\right)^{\mathsf T}
    \right]\\
    &=\sum_{i=1}^N\E_X[S_i(X)S_i(X)^{\mathsf T}]
      +\sum_{i<j}\E_X\!\left[
        S_i(X)S_j(X)^{\mathsf T}+S_j(X)S_i(X)^{\mathsf T}
      \right].
  \end{align*}
  We show that the second sum vanishes.  Each roundwise score has conditional
  mean zero given the preceding history.  Indeed,
  \begin{align*}
    \E_X[S_i(X)\mid H_i]
    &=
    \int_{\mathsf Z_i}p_{i,X}(z\mid H_i)
      \nabla_X\log p_{i,X}(z\mid H_i)\,\nu_i(\dd z\mid H_i)\\
    &=
    \int_{\mathsf Z_i}\nabla_Xp_{i,X}(z\mid H_i)\,\nu_i(\dd z\mid H_i)\\
    &=
    \nabla_X\int_{\mathsf Z_i}p_{i,X}(z\mid H_i)\,\nu_i(\dd z\mid H_i)
    =
    \nabla_X1
    =0.
  \end{align*}
  The reference measure is parameter independent at each fixed history, and
  the interchange of differentiation and integration is justified by
  \Cref{lem:block-likelihood-regularity}.  Therefore, if \(i<j\), then
  \(H_j\) contains the outcome of round \(i\), and hence \(S_i(X)\) is
  determined by \(H_j\).  It follows that
  \[
    \E_X[S_i(X)S_j(X)^{\mathsf T}]
    =
    \E_X\!\left[
      S_i(X)\E_X[S_j(X)^{\mathsf T}\mid H_j]
    \right]
    =0.
  \]
  The transpose cross term vanishes in the same way.  For the diagonal terms,
  the definition of conditional Fisher information and the law of total
  expectation give
  \[
    \E_X[S_i(X)S_i(X)^{\mathsf T}]
    =
    \E_X\!\left[
      \E_X[S_i(X)S_i(X)^{\mathsf T}\mid H_i]
    \right]
    =
    \E_X\!\left[\cI_i(H_i;X)\right].
  \]
  Substituting these identities into the expansion above proves
  \cref{eq:block-fisher-chain}.  The seed contributes no term because its law
  is parameter independent.
\end{proof}

\begin{corollary}[Adaptive Fisher information bound]
  \label{cor:adaptive-information-budget}
  Consider the family \(\{\rho_X\}\) of rank-\(r\) states defined in
  \Cref{sec:support-rotation}.  Let an adaptive protocol jointly measure at
  most \(t\) fresh samples in each round and use at most \(N\) samples along
  every root-to-leaf path.  Then the Fisher information matrix of its
  transcript satisfies, for every \(X\),
  \begin{equation}
    \Tr\bigl(\cI_{\mathrm{tr}}^{(N)}(X)\bigr)
    \le CdN\min\left\{1,\frac{\sqrt t}{r}\right\}.
    \label{eq:adaptive-information-budget}
  \end{equation}
\end{corollary}

\begin{proof}
  If \(t_i(h)=0\), the conditional Fisher information is zero.  Otherwise,
  \Cref{prop:one-block-fisher-bound} gives
  \[
    \Tr\bigl(\cI_i(h;X)\bigr)
    \le Cd t_i(h)
      \min\left\{1,\frac{\sqrt{t_i(h)}}{r}\right\}
    \le Cd t_i(h)\min\left\{1,\frac{\sqrt t}{r}\right\}.
  \]
  Taking traces in \cref{eq:block-fisher-chain} and using the pathwise sample
  bound yields
  \[
    \Tr\bigl(\cI_{\mathrm{tr}}^{(N)}(X)\bigr)
    \le
    Cd\min\left\{1,\frac{\sqrt t}{r}\right\}
    \E_X\!\left[\sum_{i=1}^Nt_i(U,Z_{<i})\right]
    \le CdN\min\left\{1,\frac{\sqrt t}{r}\right\}.
  \]
\end{proof}

\subsection{From a Fisher information bound to expected trace norm loss}
\label{sec:bayesian-reduction}

The adaptive Fisher information bound controls how much the transcript can
reveal about the support parameter \(X\).  We now convert it into a lower bound
on expected trace norm loss.  We apply the van Trees inequality using a smooth
probability density \(\pi\) supported on parameters satisfying
\(\norm X_{\mathrm{op}}<a\).  Its denominator contains the Fisher information
contributed by the measurements and \(I(\pi)\).  Restricting the estimator \(T\)
to satisfy \(\norm T_{\mathrm{op}}\le a\) then converts the resulting
Frobenius error bound into a trace norm bound.  The postprocessing in
\Cref{sec:bounded-support-estimator-postprocessing} will produce exactly such
an estimator from the output of a tomography protocol.

\subsubsection{The van Trees inequality}
\label{sec:vector-van-trees}

The van Trees inequality formalizes a simple principle: accurate estimation
requires sufficient information.  The denominator in
\cref{eq:vector-van-trees} below contains the Fisher information contributed
by the measurement transcript and \(I(\pi)\), the Fisher information of the
probability density \(\pi\).  If their sum is small, then the expected squared
estimation error on the left side must be large.

Let \(\Theta\subseteq\R^p\) be open, and draw \(\vct{\theta}\) from a smooth,
compactly supported probability density \(\pi\) on \(\Theta\).  Conditional on
\(\vct{\theta}\), let the observation have likelihood \(q_{\vct{\theta}}\) in a
regular dominated model.  Let \(\cI(\vct{\theta})\) denote the Fisher
information matrix of this likelihood family.  The precise regularity conditions are stated in
\Cref{app:vector-van-trees}; \Cref{lem:block-likelihood-regularity} verifies
them for the adaptive transcript model used here.

We write \(C_c^1(\Theta;\R)\) for the real-valued functions on \(\Theta\)
that are continuously differentiable and vanish outside a compact subset of
\(\Theta\).  The superscript \(1\) means that the first partial derivatives
exist and are continuous, the subscript \(c\) means compact support, and
\((\Theta;\R)\) gives the domain and codomain.
To obtain a nonnegative probability density that vanishes near the boundary
and whose Fisher information is easy to control, we begin with a function
\(\psi\in C_c^1(\Theta;\R)\), normalize it in \(L_2\), and define
\[
  \pi(\vct{\theta})\coloneqq\psi(\vct{\theta})^2,
  \qquad
  \norm\psi_2^2
  \coloneqq
  \int_\Theta\abs{\psi(\vct{\theta})}^2\,\dd\vct{\theta}
  =1.
\]
Then \(\pi\) integrates to one and vanishes near the boundary of \(\Theta\).
We define its Fisher information by
\begin{equation}
  I(\pi)
  \coloneqq
  4\int_\Theta\norm{\nabla\psi(\vct{\theta})}_2^2
  \,\dd\vct{\theta}.
  \label{eq:prior-fisher-energy}
\end{equation}
This agrees with the usual Fisher information of the density \(\pi\).
Indeed, wherever \(\pi(\vct{\theta})>0\),
\[
  \norm{\nabla\log\pi(\vct{\theta})}_2^2
    \pi(\vct{\theta})
  =\frac{\norm{\nabla\pi(\vct{\theta})}_2^2}
    {\pi(\vct{\theta})}
  =\frac{\norm{2\psi(\vct{\theta})
    \nabla\psi(\vct{\theta})}_2^2}{\psi(\vct{\theta})^2}
  =4\norm{\nabla\psi(\vct{\theta})}_2^2.
\]
When \(\pi(\vct{\theta})=0\), the logarithm is undefined, whereas the
expression in \cref{eq:prior-fisher-energy} remains well defined.

\begin{lemma}[Van Trees inequality]
  \label{lem:vector-van-trees}
  Assume the preceding likelihood regularity, and let
  \(\pi(\vct{\theta})\coloneqq\psi(\vct{\theta})^2\), where
  \(\psi\in C_c^1(\Theta;\R)\) and
  \(\norm\psi_2=1\).  Then every bounded measurable estimator \(T\) taking values
  in \(\R^p\) satisfies
  \begin{equation}
    \E\norm{T-\vct{\theta}}_2^2
    \ge
    \frac{p^2}{
      \E_{\vct{\theta}}
        \Tr\bigl(\cI(\vct{\theta})\bigr)+I(\pi)}.
    \label{eq:vector-van-trees}
  \end{equation}
  The expectation first draws \(\vct{\theta}\sim\pi\) and then draws the
  observation from \(q_{\vct{\theta}}\).
\end{lemma}

The proof is given in \Cref{app:vector-van-trees}.

\subsubsection{A smooth probability density for the support parameter}
\label{sec:rectangular-local-prior}

Recall that \(m=d-r\), and put \(p\coloneqq 2mr\).  We will apply van Trees
on \(\C^{m\times r}\) with a smooth density whose
Fisher information is bounded by a universal constant times \(pd/a^2\).  We
choose the density to vanish unless \(\norm X_{\mathrm{op}}<a\); together
with the same constraint on the estimator, this will give the trace norm
conversion in the next subsection.

\begin{lemma}[Smooth probability density for rectangular parameters]
  \label{lem:rectangular-local-prior}
  There is a universal constant \(C>0\) with the following property.  For
  every \(d\ge2\), \(1\le r<d\), and \(a>0\), put
  \(m\coloneqq d-r\) and \(p\coloneqq 2mr\).  There is a nonnegative,
  infinitely differentiable
  function \(\psi\) on the real Euclidean space \(\C^{m\times r}\),
  normalized by \(\int\psi(X)^2\,\dd X=1\), that vanishes outside a compact
  subset of \(\{X:\norm X_{\mathrm{op}}<a\}\).  The probability density
  \(\pi(X)\coloneqq\psi(X)^2\) satisfies
  \begin{equation}
    I(\pi)\le C\frac{pd}{a^2}.
    \label{eq:rectangular-prior-energy}
  \end{equation}
\end{lemma}

\begin{proof}
  We begin with a Gaussian density whose variance is chosen so that most of
  its probability mass lies in the region \(\norm X_{\mathrm{op}}\le a/2\).
  We then multiply its square root by a
  smooth cutoff that equals one on this region and vanishes before
  \(\norm X_{\mathrm{op}}\) reaches \(a\).  The cutoff will have gradient
  norm \(O(1/a)\), which keeps its contribution to \(I(\pi)\) under control.

  Regard \(\C^{m\times r}\) as the real Euclidean space from
  \cref{eq:real-matrix-inner-product}.  Let \(G\in\C^{m\times r}\) have
  entries \(G_{ij}=\xi_{ij}+\ii\eta_{ij}\), for \(1\le i\le m\) and
  \(1\le j\le r\), where all the real random variables
  \(\xi_{ij}\) and \(\eta_{ij}\) are independent and distributed as
  \(\mathcal N(0,\sigma^2)\).  Thus, with respect to the \(p=2mr\) real coordinates,
  the density of \(G\) is
  \[
    g_\sigma(X)
    =\frac{1}{(2\pi\sigma^2)^{p/2}}
      \exp\left(-\frac{\norm X_\F^2}{2\sigma^2}\right).
  \]
  We first show that a constant fraction of the probability mass of $G$ lies
  where \(\norm G_{\mathrm{op}}\le a/2\).  Write \(G=A+\ii B\), where \(A\) and
  \(B\) are independent real Gaussian matrices whose entries have variance
  \(\sigma^2\).
  Applying the standard operator norm bound for real Gaussian matrices (see, e.g., Vershynin~\cite[Theorem~4.6.1]{Vershynin26-high-dim-probability}), followed by a
  union bound, gives a universal constant \(C>0\) such that
  \[
    \Prb\!\left[
      \norm G_{\mathrm{op}}>C\sigma(\sqrt m+\sqrt r)
    \right]
    \le\frac14.
  \]
  Since \(\sqrt m+\sqrt r\le\sqrt{2d}\), choosing
  \(\sigma=c a/\sqrt d\) for a sufficiently small universal constant
  \(c>0\) gives
  \[
    \Prb\!\left[\norm G_{\mathrm{op}}\le a/2\right]\ge\frac34.
  \]

  We next construct the smooth cutoff.  Let \(h:[0,\infty)\to[0,1]\) equal
  one on \([0,9/16]\), vanish on \([11/16,\infty)\), and be linear between
  these intervals.  Let \(\varphi\) be a nonnegative infinitely
  differentiable function on \(\C^{m\times r}\), supported where
  \(\norm H_\F\le a/16\), and normalized so that
  \(\int\varphi(H)\,\dd H=1\).  Average the cutoff over these small
  perturbations by setting
  \[
    \chi(X)
    \coloneqq
    \int h\left(\frac{\norm{X-H}_{\mathrm{op}}}{a}\right)
      \varphi(H)\,\dd H
    =\int h\left(\frac{\norm Y_{\mathrm{op}}}{a}\right)
      \varphi(X-Y)\,\dd Y.
  \]
  The first expression shows that \(\chi(X)\) is a weighted average of the
  cutoff at matrices close to \(X\).  The second shows that \(\chi\) is
  infinitely differentiable: differentiating with respect to \(X\)
  differentiates \(\varphi(X-Y)\), not the operator norm.  By construction,
  \(0\le\chi\le1\).

  If \(\norm X_{\mathrm{op}}\le a/2\) and \(\varphi(H)\ne0\), then
  \[
    \norm{X-H}_{\mathrm{op}}
    \le\norm X_{\mathrm{op}}+\norm H_\F
    \le\frac{9a}{16},
  \]
  so \(\chi(X)=1\).  Similarly, if
  \(\norm X_{\mathrm{op}}\ge3a/4\), then
  \(\norm{X-H}_{\mathrm{op}}\ge11a/16\) whenever \(\varphi(H)\ne0\), so
  \(\chi(X)=0\).  Finally, the operator norm is \(1\)-Lipschitz with respect
  to the Frobenius norm, and \(h\) is \(8\)-Lipschitz.  Averaging its
  translates therefore gives
  \[
    \abs{\chi(X)-\chi(Y)}
    \le\frac8a\norm{X-Y}_\F,
    \qquad
    \norm{\nabla\chi(X)}_2\le\frac8a.
  \]

  Define
  \[
    \widetilde\psi(X)\coloneqq\sqrt{g_\sigma(X)}\,\chi(X),
    \qquad
    Z\coloneqq\norm{\widetilde\psi}_2^2.
  \]
  Because \(\chi=1\) when \(\norm X_{\mathrm{op}}\le a/2\),
  \[
    Z
    =\int g_\sigma(X)\chi(X)^2\,\dd X
    \ge \Prb\!\left[\norm G_{\mathrm{op}}\le a/2\right]
    \ge\frac34.
  \]
  The Gaussian factor satisfies
  \[
    \nabla\sqrt{g_\sigma(X)}
    =-\frac{X}{2\sigma^2}\sqrt{g_\sigma(X)},
    \qquad
    \int\norm{\nabla\sqrt{g_\sigma}}_2^2
    =\frac{\E\norm G_\F^2}{4\sigma^4}
    =\frac{p}{4\sigma^2}.
  \]
  The product rule now gives
  \[
    \nabla\widetilde\psi
    =\chi\nabla\sqrt{g_\sigma}+\sqrt{g_\sigma}\,\nabla\chi,
  \]
  and the inequality
  \(\norm{\vct u+\vct v}_2^2
    \le2\norm{\vct u}_2^2+2\norm{\vct v}_2^2\) gives
  \begin{align*}
    \int\norm{\nabla\widetilde\psi}_2^2
    &\le
      2\int\chi^2\norm{\nabla\sqrt{g_\sigma}}_2^2
      +2\int g_\sigma\norm{\nabla\chi}_2^2\\
    &\le
      2\int\norm{\nabla\sqrt{g_\sigma}}_2^2
      +\frac{C}{a^2}\int g_\sigma\\
    &=\frac{p}{2\sigma^2}+\frac{C}{a^2}\\
    &=\frac{pd}{2c^2a^2}+\frac{C}{a^2}
    \le C\frac{pd}{a^2}.
  \end{align*}
  The second inequality uses \(0\le\chi\le1\) and
  \(\norm{\nabla\chi}_2\le8/a\).  The two equalities use the Gaussian
  integral above, \(\int g_\sigma=1\), and \(\sigma=ca/\sqrt d\);
  the last inequality uses \(pd\ge1\) and that \(c\) is a universal constant.
  Normalize by setting \(\psi\coloneqq\widetilde\psi/\sqrt Z\).  The function
  \(\psi\) is nonnegative and infinitely differentiable, has norm one in
  \(L_2\), and vanishes whenever \(\norm X_{\mathrm{op}}\ge3a/4\).  Its
  support is therefore a compact subset of
  \(\{X:\norm X_{\mathrm{op}}<a\}\).  Since \(Z\ge3/4\),
  \[
    4\int\norm{\nabla\psi}_2^2
    =\frac4Z\int\norm{\nabla\widetilde\psi}_2^2
    \le C\frac{pd}{a^2}.
  \]
  The density \(\pi\coloneqq\psi^2\) proves the lemma.
\end{proof}

\subsubsection{From Frobenius loss to trace norm loss}
\label{sec:localized-nuclear-risk}

The van Trees inequality gives a lower bound on squared Frobenius error, whereas the
theorem concerns trace norm loss.  If the parameter and the estimator both
have operator norm at most \(a\), then their difference has operator norm at
most \(2a\), which converts one loss into the other.

\begin{lemma}[Trace norm consequence of van Trees]
  \label{lem:localized-nuclear-risk}
  Let \(X\sim\pi\) take values in \(\C^{m\times r}\), viewed as a real
  Euclidean space of dimension \(p\coloneqq 2mr\), and suppose the density and
  likelihood satisfy the assumptions of \Cref{lem:vector-van-trees}.  Assume
  that \(\pi\) vanishes unless \(\norm X_{\mathrm{op}}<a\).  Then every
  measurable estimator \(T\in\C^{m\times r}\) satisfying
  \(\norm T_{\mathrm{op}}\le a\) almost surely obeys
  \begin{equation}
    \E\norm{T-X}_1
    \ge
    \frac{p^2}{
      2a\left(
        \E_{X\sim\pi}\Tr\bigl(\cI(X)\bigr)+I(\pi)
      \right)}.
    \label{eq:localized-nuclear-risk}
  \end{equation}
  Here the expectation is over both \(X\sim\pi\) and the corresponding
  observation.
\end{lemma}

\begin{proof}
  Under the real inner product from \cref{eq:real-matrix-inner-product}, the
  Euclidean norm in \Cref{lem:vector-van-trees} is the Frobenius norm.  For
  every realization of \(X\) and the observation, the singular values of
  \(T-X\) give
  \[
    \norm{T-X}_\F^2
    \le\norm{T-X}_{\mathrm{op}}\norm{T-X}_1
    \le2a\norm{T-X}_1.
  \]
  Taking expectations and applying \Cref{lem:vector-van-trees} proves the
  claim.
\end{proof}

\subsection{Completion of the lower bound}
\label{sec:lower-bound-completion}

It remains to connect the expected loss bound from
\Cref{lem:localized-nuclear-risk} to the high-probability guarantee in the
theorem.  We first reduce the failure probability by repeating the protocol.
We then postprocess the tomographic output into a support parameter satisfying
the operator norm constraint used to convert Frobenius loss into trace norm
loss.
Combining the resulting upper and lower bounds on its expected loss proves the
theorem when the rank is at most half the dimension.  A depolarizing channel
reduction handles the remaining ranks.

\subsubsection{From high-probability tomography to a support estimator}
\label{sec:bounded-support-estimator-postprocessing}

The tomography guarantee allows a failure event of probability \(1/3\), on
which the output may be arbitrarily inaccurate.  Before passing to expected
loss, we reduce this probability to a small constant \(\delta\).  We use
the standard confidence amplification rule that selects a
candidate with the smallest majority radius; see Hsu and Sabato
\cite[Sec.~3.2, Proposition~8 and Algorithm~2]{HsuSabato2016}.
We include the short argument for completeness.

\begin{lemma}[Confidence amplification]
  \label{lem:metric-median}
  Let \(\eta\ge0\), and suppose an estimator in a metric space is within
  \(\eta\) of its target with probability at least \(2/3\).  For every
  \(0<\delta<1\), there is an odd integer \(K=K(\delta)\) such that \(K\)
  independent repetitions can be postprocessed, without further observations,
  into an estimator within \(3\eta\) of the target with probability at least
  \(1-\delta\).
\end{lemma}

\begin{proof}
  Choose an odd \(K\) so that more than half of the repetitions are
  successful with probability at least \(1-\delta\).  Hoeffding's inequality
  shows that \(K=O(1+\log(1/\delta))\) suffices.  For outputs
  \(T_1,\ldots,T_K\), let \(R_j\) be the \((K+1)/2\)-th smallest value among
  the distances from \(T_j\) to the \(K\) outputs.  Select an output with
  minimum \(R_j\), breaking ties by index.

  Suppose that more than half of the outputs are within \(\eta\) of the
  target.  Every such output is within \(2\eta\) of every other successful
  output, and hence has \(R_j\le2\eta\).  The selected output therefore has
  more than half of the outputs within distance \(2\eta\).  This set
  intersects the successful majority, so the selected output is within
  \(3\eta\) of the target.
\end{proof}

For \(1\le r\le d/2\) and \(a>0\), write
\[
  \mathcal X_a
  \coloneqq
  \left\{X\in\C^{m\times r}:\norm X_{\mathrm{op}}\le a\right\}.
\]
Recall that \(X\mapsto\rho_X\) denotes the rank-\(r\) family from
\Cref{sec:support-rotation}.  The next lemma converts an arbitrary matrix
that accurately estimates \(\rho_X\) into an estimate of \(X\) belonging
to \(\mathcal X_a\).

\begin{lemma}[Postprocessing to a bounded support estimator]
  \label{lem:bounded-support-estimator-postprocessing}
  Let \(d\ge2\) and \(1\le r\le d/2\) be integers, let
  \(0<a\le1/4\) and \(\eta>0\), and let \(0<\delta<1\).  Suppose a measurable
  \(d\times d\) matrix-valued estimator
  \(\overline\rho\) satisfies
  \[
    \Prb_{\rho_X}\!\left[
      \norm{\overline\rho-\rho_X}_1\le3\eta
    \right]
    \ge1-\delta
  \]
  for every \(X\in\mathcal X_a\).  Then there is a measurable function of
  \(\overline\rho\), denoted \(\widehat X\), that takes values in
  \(\mathcal X_a\) and satisfies, for every \(X\in\mathcal X_a\),
  \begin{equation}
    \E_X\norm{\widehat X-X}_1
    \le14r\eta+2\delta ar.
    \label{eq:bounded-support-estimator-risk}
  \end{equation}
\end{lemma}

\begin{proof}
  The set \(\mathcal X_a\) is compact, and the map
  \(X\mapsto\rho_X\) is continuous.  For each \(Y\in\mathcal X_a\), the set
  \[
    \left\{
      X\in\mathcal X_a:
      \norm{\rho_X-\rho_Y}_1<\eta
    \right\}
  \]
  is open in \(\mathcal X_a\), and these sets cover \(\mathcal X_a\) as
  \(Y\) varies.  By compactness,
  finitely many of them, centered at parameters
  \(Y_1,\ldots,Y_M\in\mathcal X_a\), already cover \(\mathcal X_a\).  Put
  \(\mathcal C\coloneqq\{Y_1,\ldots,Y_M\}\).  Then, for every
  \(X\in\mathcal X_a\), some \(Y\in\mathcal C\) satisfies
  \begin{equation}
    \norm{\rho_Y-\rho_X}_1<\eta.
    \label{eq:finite-support-family-cover}
  \end{equation}
  Given \(\overline\rho\), choose \(\widehat X\in\mathcal C\) minimizing
  \(\norm{\overline\rho-\rho_{\widehat X}}_1\), with ties broken according
  to a fixed ordering of \(\mathcal C\).  This is a measurable
  postprocessing and guarantees \(\norm{\widehat X}_{\mathrm{op}}\le a\).

  Fix the true parameter \(X\), and choose \(Y\) as in
  \cref{eq:finite-support-family-cover}.  On the event
  \(\norm{\overline\rho-\rho_X}_1\le3\eta\), the minimizing property of
  \(\widehat X\) gives
  \[
    \norm{\overline\rho-\rho_{\widehat X}}_1
    \le\norm{\overline\rho-\rho_Y}_1
    \le4\eta.
  \]
  Hence
  \[
    \norm{\rho_{\widehat X}-\rho_X}_1
    \le
    \norm{\rho_{\widehat X}-\overline\rho}_1
    +\norm{\overline\rho-\rho_X}_1
    \le7\eta,
  \]
  and the inverse bound in \cref{eq:graph-inverse-lipschitz} gives
  \[
    \norm{\widehat X-X}_1\le14r\eta.
  \]
  On the complementary event, both parameters belong to \(\mathcal X_a\).
  Since \(\widehat X-X\) has at most \(r\) nonzero singular values,
  \[
    \norm{\widehat X-X}_1
    \le r\norm{\widehat X-X}_{\mathrm{op}}
    \le2ar.
  \]
  Averaging over the success and failure events gives
  \[
    \E_X\norm{\widehat X-X}_1
    \le14r\eta
      +2ar\Prb_{\rho_X}\!\left[
        \norm{\overline\rho-\rho_X}_1>3\eta
      \right]
    \le14r\eta+2\delta ar,
  \]
  which proves \cref{eq:bounded-support-estimator-risk}.
\end{proof}

\subsubsection{The lower bound when the rank is at most half the dimension}
\label{sec:active-family-lower-bound}

We now apply van Trees to the estimator produced above.  Successful
tomography gives it expected trace norm loss much smaller than \(ar\), once
\(a\) is chosen as a sufficiently large constant multiple of the target
accuracy and \(\delta\) is sufficiently small.  If the protocol uses too few
samples, however, the Fisher information bound forces expected loss of order
\(ar\).  Comparing these conclusions gives the lower bound on sample complexity.

\begin{proposition}[Lower bound when the rank is at most half the dimension]
  \label{prop:rectangular-regime-lower-bound}
  There are universal constants \(c_0,\eta_0>0\) with the following
  property.  Let \(d\ge2\), \(1\le r\le d/2\), and \(t\ge1\) be integers,
  and let \(0<\eta\le\eta_0\).  Consider an adaptive protocol that jointly
  measures at most \(t\) fresh samples in each round and whose output
  \(\widehat\rho\) may be an arbitrary \(d\times d\) matrix.  Suppose that,
  for every state
  \(\rho\) on \(\C^d\) of rank exactly \(r\),
  \[
    \Prb_\rho\!\left[
      \norm{\widehat\rho-\rho}_1\le\eta
    \right]
    \ge\frac23.
  \]
  If the protocol uses at most \(n\) samples on every execution, then
  \begin{equation}
    n\ge
    c_0\frac{dr}{\eta^2}
    \max\left\{1,\frac{r}{\sqrt t}\right\}.
    \label{eq:rectangular-regime-lower-bound}
  \end{equation}
\end{proposition}

\begin{proof}
  Let \(L>1\) and \(0<\delta<1\) be universal constants to be chosen below,
  and set \(a\coloneqq L\eta\).  We will choose \(\eta_0\) so that
  \(a\le1/4\).
  Repeat the protocol independently \(K=K(\delta)\) times and apply
  \Cref{lem:metric-median} in trace norm.  The amplified protocol uses at
  most \(N\coloneqq Kn\) samples on every execution, still jointly measures at
  most \(t\) samples in each round, and outputs a matrix \(\overline\rho\)
  satisfying
  \[
    \Prb_\rho\!\left[
      \norm{\overline\rho-\rho}_1\le3\eta
    \right]
    \ge1-\delta
  \]
  for every state \(\rho\) of rank exactly \(r\).  Applying
  \Cref{lem:bounded-support-estimator-postprocessing} to the family
  \(\{\rho_X:X\in\mathcal X_a\}\) produces a measurable estimator
  \(\widehat X\in\mathcal X_a\) such that, for every
  \(X\in\mathcal X_a\),
  \begin{equation}
    \E_X\norm{\widehat X-X}_1
    \le
    \left(\frac{14}{L}+2\delta\right)ar.
    \label{eq:operational-parameter-risk}
  \end{equation}

  We next derive the incompatible lower bound.  Put
  \(p\coloneqq 2mr\), the real dimension of the parameter \(X\).  Apply
  \Cref{lem:rectangular-local-prior} and draw \(X\) from the
  resulting density \(\pi\).  This density is supported where
  \(\norm X_{\mathrm{op}}<a\) and satisfies
  \begin{equation}
    I(\pi)\le C\frac{pd}{a^2}.
    \label{eq:completion-density-information}
  \end{equation}
  From now on, an unqualified expectation first draws \(X\sim\pi\) and then
  draws the complete amplified transcript.  Since
  \cref{eq:operational-parameter-risk} holds for every \(X\) in the support
  of \(\pi\), it gives the same upper bound under this joint expectation.

  We apply the van Trees inequality to \(\widehat X\), viewed as a measurable
  function of the complete transcript of the amplified protocol.  By construction,
  \(\norm{\widehat X}_{\mathrm{op}}\le a\).  By
  \Cref{cor:adaptive-information-budget},
  \begin{equation}
    \E_{X\sim\pi}
      \Tr\bigl(\cI_{\mathrm{tr}}^{(N)}(X)\bigr)
    \le
    CdN\min\left\{1,\frac{\sqrt t}{r}\right\}.
    \label{eq:completion-measurement-information}
  \end{equation}
  Therefore,
  \Cref{lem:localized-nuclear-risk} and
  \cref{eq:completion-density-information,eq:completion-measurement-information}
  give
  \begin{equation}
    \E\norm{\widehat X-X}_1
    \ge
    \frac{p^2}{
      2a\left(
        CdN\min\left\{1,\frac{\sqrt t}{r}\right\}
        +Cpd/a^2
      \right)}.
    \label{eq:completion-van-trees-risk}
  \end{equation}
  Suppose, toward a contradiction, that
  \begin{equation}
    N\le
    \frac{dr}{a^2}
    \max\left\{1,\frac{r}{\sqrt t}\right\}.
    \label{eq:insufficient-amplified-copies}
  \end{equation}
  The two factors involving \(t\) cancel:
  \[
    N\min\left\{1,\frac{\sqrt t}{r}\right\}
    \le
    \frac{dr}{a^2}
    \max\left\{1,\frac{r}{\sqrt t}\right\}
    \min\left\{1,\frac{\sqrt t}{r}\right\}
    =\frac{dr}{a^2}.
  \]
  Since \(r\le d/2\), we have \(m\ge d/2\) and hence
  \(p=2mr\ge dr\).  Both terms in the denominator of
  \cref{eq:completion-van-trees-risk} are consequently at most a universal
  constant times \(pd/a^2\).  It follows that, for a universal constant
  \(c_1>0\),
  \begin{equation}
    \E\norm{\widehat X-X}_1
    \ge c_1\frac{pa}{d}
    =2c_1ar\frac{m}{d}
    \ge c_1ar.
    \label{eq:completion-risk-lower-simplified}
  \end{equation}

  Choose \(L\) large enough and then \(\delta\) small enough that
  \(14/L+2\delta<c_1\).  The upper bound
  \cref{eq:operational-parameter-risk} then contradicts
  \cref{eq:completion-risk-lower-simplified}.  Thus
  \cref{eq:insufficient-amplified-copies} is false.  Since
  \(N=Kn\) and \(a=L\eta\),
  \[
    n>
    \frac{1}{KL^2}\frac{dr}{\eta^2}
    \max\left\{1,\frac{r}{\sqrt t}\right\}.
  \]
  Once \(L\) and \(\delta\) are fixed, so is \(K\).  Absorbing this fixed
  factor into \(c_0\), and choosing \(\eta_0\le1/(4L)\), proves the
  proposition.
\end{proof}

\subsubsection{Rank lifting by depolarization}
\label{sec:rank-lifting}

It remains to reduce ranks larger than half the ambient dimension to the
regime of the preceding proposition.  The following channel embeds the input
state and mixes it with the maximally mixed state on the same \(r\)-dimensional
subspace.  The added component makes every output have rank exactly \(r\), but
it cancels when two outputs are subtracted, so every trace norm distance is
scaled by exactly one half.  Applying the adjoint channel to each POVM then
simulates any adaptive measurement protocol for the lifted states using
samples of the original states.

\begin{lemma}[Depolarizing rank lift]
  \label{lem:depolarizing-rank-lift}
  Let \(1\le r\le d\) and \(t\ge1\) be integers, let
  \(W:\C^r\to\C^d\) be an isometry, and put \(\tau\coloneqq\Id_r/r\).
  Define the channel
  \begin{equation}
    \Phi(A)
    \coloneqq
    \frac12WAW^\dagger
    +\frac12\Tr(A)W\tau W^\dagger.
    \label{eq:depolarizing-rank-lift}
  \end{equation}
  For every state \(\sigma\) on \(\C^r\), the state \(\Phi(\sigma)\) has
  rank exactly \(r\).  For every pair of states \(\sigma,\sigma'\),
  \begin{equation}
    \norm{\Phi(\sigma)-\Phi(\sigma')}_1
    =\frac12\norm{\sigma-\sigma'}_1.
    \label{eq:rank-lift-metric-scaling}
  \end{equation}
  Moreover, an adaptive protocol for the lifted states that jointly measures
  at most \(t\) fresh samples in each round induces a protocol for the original
  states with identical transcript laws, the same round sizes, and the same
  sample count on every execution.

  If the lifted protocol outputs a \(d\times d\) matrix \(B\), the induced
  protocol may output
  \[
    \mathcal R(B)\coloneqq2W^\dagger BW-\tau.
  \]
  For every state \(\sigma\), this postprocessing satisfies
  \begin{equation}
    \norm{\mathcal R(B)-\sigma}_1
    \le2\norm{B-\Phi(\sigma)}_1.
    \label{eq:rank-lift-output-error}
  \end{equation}
\end{lemma}

\begin{proof}
  The map \(\Phi\) is completely positive and trace preserving, and hence is
  a quantum channel.  For a state \(\sigma\),
  \[
    \Phi(\sigma)
    =W\left(\frac{\sigma+\tau}{2}\right)W^\dagger.
  \]
  The operator \((\sigma+\tau)/2\) is positive definite on \(\C^r\), so
  \(W(\sigma+\tau)W^\dagger/2\) has rank \(r\).  Also,
  \[
    \Phi(\sigma)-\Phi(\sigma')
    =\frac12W(\sigma-\sigma')W^\dagger.
  \]
  Isometric embedding preserves the nonzero singular values, which proves
  \cref{eq:rank-lift-metric-scaling}.

  The output postprocessing obeys the exact identity
  \[
    \mathcal R(B)-\sigma
    =2W^\dagger\bigl(B-\Phi(\sigma)\bigr)W.
  \]
  Since \(\norm{W^\dagger A W}_1\le\norm A_1\), this proves
  \cref{eq:rank-lift-output-error}.  The map \(\mathcal R\) is continuous and
  hence is a measurable postprocessing.

  It remains to simulate the adaptive measurements.  Suppose that, after a
  history \(h\) in round \(i\), the lifted protocol chooses a POVM
  \(\mathsf M(\cdot\mid h)\) on \(t_i(h)\le t\) samples.  For every
  measurable outcome event \(E\), define
  \[
    \mathsf M'(E\mid h)
    \coloneqq
    \bigl(\Phi^{\otimes t_i(h)}\bigr)^\dagger
      \bigl(\mathsf M(E\mid h)\bigr).
  \]
  For each \(k\in\{1,\ldots,t\}\), on the measurable set of histories where
  \(t_i(h)=k\), this construction applies the fixed linear map
  \((\Phi^{\otimes k})^\dagger\) to the original measurable POVM rule.
  Hence the transformed rule is also measurable.
  The adjoint of \(\Phi^{\otimes t_i(h)}\) is completely positive and
  unital, and hence \(\mathsf M'(\cdot\mid h)\) is a POVM.  By the definition
  of the adjoint,
  \[
    \Tr\left(\mathsf M'(E\mid h)\sigma^{\otimes t_i(h)}\right)
    =
    \Tr\left(
      \mathsf M(E\mid h)\Phi(\sigma)^{\otimes t_i(h)}
    \right).
  \]
  Thus the conditional outcome law agrees at every history.  Induction over
  the rounds gives identical full transcript laws.  The simulated protocol
  uses the same \(t_i(h)\) at every history, so neither the allowed joint
  measurement size nor the total sample count on any execution changes.
\end{proof}

\begin{proof}[Proof of \Cref{thm:adaptive-block-lower-bound}]
  Let \(c_0,\eta_0\) be the constants from
  \Cref{prop:rectangular-regime-lower-bound}, and set
  \(\eps_0\coloneqq\eta_0/2\).
  First suppose \(r\le d/2\).  Applying
  \Cref{prop:rectangular-regime-lower-bound} to the states of rank \(r\),
  with accuracy \(\eta=\eps\), gives
  \[
    n\ge
    c_0\frac{dr}{\eps^2}
    \max\left\{1,\frac{r}{\sqrt t}\right\}.
  \]

  Now suppose \(r>d/2\), so \(r\ge2\), and set
  \(k\coloneqq\lfloor r/2\rfloor\).  Fix an isometry
  \(W:\C^r\to\C^d\), independently of the input state, and apply
  \Cref{lem:depolarizing-rank-lift}.  For every
  state \(\sigma\) on \(\C^r\) of rank exactly \(k\), the state
  \(W\sigma W^\dagger\) has support contained in the same fixed
  \(r\)-dimensional subspace \(\operatorname{Im}(W)\), and
  \(\Phi(\sigma)\) has support exactly \(\operatorname{Im}(W)\), hence rank
  exactly \(r\).  We may therefore simulate the assumed protocol on
  \(\Phi(\sigma)\) and postprocess its output \(\widehat\rho\) as
  \(\widehat\sigma\coloneqq\mathcal R(\widehat\rho)\).  By
  \cref{eq:rank-lift-output-error},
  \[
    \Prb_\sigma\!\left[
      \norm{\widehat\sigma-\sigma}_1\le2\eps
    \right]
    \ge
    \Prb_{\Phi(\sigma)}\!\left[
      \norm{\widehat\rho-\Phi(\sigma)}_1\le\eps
    \right]
    \ge\frac23.
  \]
  The induced protocol jointly measures at most \(t\) samples in each round and
  uses at most \(n\) samples on every execution.  Its output need not be a
  state, which is allowed in
  \Cref{prop:rectangular-regime-lower-bound}.  Applying that proposition in
  ambient dimension \(r\), at rank \(k\) and accuracy \(2\eps\), gives
  \begin{equation}
    n\ge
    \frac{c_0}{4}\frac{rk}{\eps^2}
    \max\left\{1,\frac{k}{\sqrt t}\right\}.
    \label{eq:large-rank-reduced-bound}
  \end{equation}
  For \(r\ge2\), we have \(k\ge r/3\).  Since \(r>d/2\),
  \[
    rk\ge\frac{dr}{6},
    \qquad
    \max\left\{1,\frac{k}{\sqrt t}\right\}
    \ge\frac13
      \max\left\{1,\frac{r}{\sqrt t}\right\}.
  \]
  Substituting these inequalities into
  \cref{eq:large-rank-reduced-bound} yields
  \[
    n\ge
    \frac{c_0}{72}\frac{dr}{\eps^2}
    \max\left\{1,\frac{r}{\sqrt t}\right\}.
  \]

  Taking \(c\coloneqq c_0/72\) proves
  \cref{eq:main-adaptive-block-lower-bound} in both rank regimes.
\end{proof}

\section{Rank-sensitive tomography}
\label{sec:upper-bound-proof}

This section proves \Cref{thm:bounded-block-upper-bound} by constructing and
analyzing a nonadaptive tomography protocol based on a Gaussian joint
measurement.
We begin by defining \(\mathsf M_t\), a joint measurement on \(t\) samples.
In \Cref{sec:gaussian-block-estimator-target}, we give the complete
nonadaptive protocol and state the properties of one measurement needed for
its analysis.  We prove these properties in
\Cref{sec:gaussian-joint-measurement-properties} and analyze the resulting
estimator in \Cref{sec:upper-bound-aggregation}.

For each positive integer \(t\), put
\(\ell_{\max}\coloneqq\min\{d,t\}\).  For
\(1\le\ell\le\ell_{\max}\), let \(\gamma_{d,\ell}\) be the standard complex
Gaussian law on \(\C^{d\times\ell}\), and define
\begin{equation}
  \Gamma_{\ell,t}
  \coloneqq\int_{\C^{d\times\ell}}
    (GG^\dagger)^{\otimes t}\,\dd\gamma_{d,\ell}(G).
  \label{eq:gaussian-moment-operator}
\end{equation}
Let \(\Pi_{\ell,t}\) be the orthogonal projector onto the support of
\(\Gamma_{\ell,t}\), set \(\Pi_{0,t}\coloneqq0\), and define
\begin{equation}
  \Delta_{\ell,t}\coloneqq\Pi_{\ell,t}-\Pi_{\ell-1,t},
  \qquad 1\le\ell\le\ell_{\max}.
  \label{eq:gaussian-support-differences}
\end{equation}
We prove in \Cref{lem:gaussian-moment-support} that the support projectors
\(\Pi_{\ell,t}\) are nested, so every \(\Delta_{\ell,t}\) is a projector.  The outcome space
consists of pairs \((J,G)\), where \(1\le J\le\ell_{\max}\) and
\(G\in\C^{d\times J}\).  For \(J=\ell\), define the operator-valued density
\begin{equation}
  \mathsf M_t(\ell,\dd G)
  \coloneqq
  \Delta_{\ell,t}(\Gamma_{\ell,t}^+)^{1/2}
  (GG^\dagger)^{\otimes t}
  (\Gamma_{\ell,t}^+)^{1/2}\Delta_{\ell,t}
  \,\dd\gamma_{d,\ell}(G),
  \label{eq:gaussian-block-povm}
\end{equation}
where \(\Gamma_{\ell,t}^+\) denotes the Moore--Penrose pseudoinverse of
\(\Gamma_{\ell,t}\).
The Gaussian measure in this definition is a reference measure.  The Born
rule determines the outcome density with respect to this measure.

The decomposition and the pseudoinverse factors are chosen so that
integrating \(\mathsf M_t(\ell,\dd G)\) over \(G\) gives
\(\Delta_{\ell,t}\).  These operators sum to the identity, so the densities
\(\mathsf M_t(\ell,\dd G)\), for \(\ell\in[\ell_{\max}]\) and
\(G\in\C^{d\times\ell}\), form a POVM.  We prove this and establish the
properties of the POVM in
\Cref{sec:gaussian-joint-measurement-properties}.

\subsection{The protocol}
\label{sec:gaussian-block-estimator-target}

Recall that the rank of the unknown state is at most \(r\).
For a fixed total number of samples, increasing the number \(t\) of samples
measured jointly improves the resulting asymptotic rate only until \(t\)
reaches \(r^2\).  We therefore set
\begin{equation}
  s\coloneqq\min\{t,r^2\},
  \qquad
  B\coloneqq\left\lceil
    C_0\frac{dr^2}{\eps^2 s {\sqrt s}}
  \right\rceil,
  \qquad
  N\coloneqq Bs,
  \label{eq:upper-bound-protocol-parameters}
\end{equation}
where \(C_0\) is a sufficiently large universal constant.  We apply
\(\mathsf M_s\) independently \(B\) times, using \(s\) fresh samples each
time.  Suppose the \(b\)-th outcome is the pair
\((J_b,G_b)\); here \(1\le J_b\le\min\{d,s\}\) and
\(G_b\in\C^{d\times J_b}\).  Form
\begin{equation}
  Y_b\coloneqq\frac{G_bG_b^\dagger-J_b\Id_d}{s},
  \qquad
  \overline Y\coloneqq\frac1B\sum_{b=1}^B Y_b.
  \label{eq:upper-bound-protocol-average}
\end{equation}
The matrices \(Y_b\) are Hermitian, but need not be positive semidefinite or
have trace one.
Output a density matrix nearest in Frobenius norm to \(\overline Y\),
subject to having rank at most \(r\):
\begin{equation}
  \widehat\rho\in
  \underset{\sigma\in\cD_r(\C^d)}{\arg\min}\,
    \norm{\overline Y-\sigma}_\F.
  \label{eq:upper-bound-protocol-output}
\end{equation}
Such a minimizer exists because \(\cD_r(\C^d)\) is compact, and a measurable
choice is provided by \Cref{lem:measurable-low-rank-projection}.  This is the
complete protocol.  It is nonadaptive; it uses the same joint
measurement in every repetition.

We prove the correctness of the protocol in
\Cref{sec:upper-bound-aggregation} using the following guarantee for one
application of the joint measurement.

\begin{proposition}[Properties of one joint measurement]
  \label{prop:gaussian-block-estimator}
  For every positive integer \(t\), the densities in
  \cref{eq:gaussian-block-povm} form a POVM \(\mathsf M_t\) on
  \((\C^d)^{\otimes t}\) whose outcome is a pair \((J,G)\), where
  \(1\le J\le\min\{d,t\}\) and \(G\in\C^{d\times J}\).  Associate with this
  outcome the Hermitian matrix
  \begin{equation}
    Y\coloneqq\frac{GG^\dagger-J\Id_d}{t}.
    \label{eq:gaussian-block-estimator}
  \end{equation}
  For every state \(\rho\in\cD(\C^d)\), this matrix is exactly unbiased:
  \begin{equation}
    \E_\rho[Y]=\rho.
    \label{eq:gaussian-block-unbiased}
  \end{equation}
  Its exact uncentered tensor second moment is
  \begin{equation}
    \E_\rho[Y\otimes Y]
    =\frac{t-1}{t}\rho^{\otimes2}
      +\frac1t(\rho\otimes\Id_d+\Id_d\otimes\rho)F
      +\frac{\E_\rho[J]}{t^2}F,
    \label{eq:gaussian-block-second-moment}
  \end{equation}
  where \(F\) swaps the two factors of \(\C^d\otimes\C^d\).
  If \(\rho\) has rank at most \(r\), then
  \(\E_\rho[J]=O(\min\{r,\sqrt t\})\).
  Finally, let \(\Pi\) be any orthogonal projector satisfying \(\Pi\rho=\rho\).  For
  every \(\ell\) with \(\Prb_\rho[J=\ell]>0\), conditional on \(J=\ell\),
  the columns of \(\Pi^\perp G\), viewed in \(\operatorname{Im}(\Pi^\perp)\),
  are independent standard complex Gaussian vectors and are independent of
  \(\Pi G\).
\end{proposition}

\subsection{Properties of the joint measurement}
\label{sec:gaussian-joint-measurement-properties}

We now prove \Cref{prop:gaussian-block-estimator}.  We first characterize the
support of \(\Gamma_{\ell,t}\), which is used to define \(\mathsf M_t\), then
bound the mean of \(J\), derive the Gaussian moment identities, and prove the
remaining claims.

\subsubsection{Support of the Gaussian moment operators}
\label{sec:gaussian-moment-supports}

To verify that the densities in \cref{eq:gaussian-block-povm} form a POVM,
we first show that the projectors \(\Pi_{\ell,t}\) are nested and that their
ranges span the entire tensor product space.  It then follows that the differences
\(\Delta_{\ell,t}\) are projectors that sum to the identity.
For \(1\le\ell\le d\), define
\begin{equation}
  \mathcal V_{\ell,t}
  \coloneqq
  \operatorname{span}\left(
    \bigcup_{\substack{L\subseteq\C^d\\\dim L\le\ell}}
      L^{\otimes t}
  \right),
  \qquad \mathcal V_{0,t}\coloneqq\{0\}.
  \label{eq:gaussian-support-subspace}
\end{equation}

\begin{lemma}[Support of the Gaussian moment operators]
  \label{lem:gaussian-moment-support}
  For every \(1\le\ell\le\ell_{\max}\), the support of
  \(\Gamma_{\ell,t}\) is \(\mathcal V_{\ell,t}\).  The orthogonal projector
  \(\Pi_{\ell,t}\) onto \(\mathcal V_{\ell,t}\) commutes with every tensor
  permutation \(U_\pi\), \(\pi\in\mathfrak S_t\), and with
  \(A^{\otimes t}\) for every operator \(A\) on \(\C^d\).  Moreover,
  \begin{equation}
    0=\Pi_{0,t}\preceq\Pi_{1,t}\preceq\cdots
      \preceq\Pi_{\ell_{\max},t}=\Id_d^{\otimes t}.
    \label{eq:nested-gaussian-supports}
  \end{equation}
\end{lemma}

\begin{proof}
  For every \(G\in\C^{d\times\ell}\),
  \(\dim(\operatorname{Im}(G))\le\ell\), so
  \[
    \operatorname{Im}\bigl((GG^\dagger)^{\otimes t}\bigr)
    \subseteq\operatorname{Im}(G)^{\otimes t}
    \subseteq\mathcal V_{\ell,t},
  \]
  and hence
  \begin{equation}
    \operatorname{supp}(\Gamma_{\ell,t})
    \subseteq\mathcal V_{\ell,t}.
    \label{eq:gaussian-support-first-inclusion}
  \end{equation}
  Conversely, suppose \(\vct{v}\in\ker(\Gamma_{\ell,t})\).  Then
  \[
    0
    =\ip{\vct{v}}{\Gamma_{\ell,t}\vct{v}}
    =\int
      \ip{\vct{v}}{(GG^\dagger)^{\otimes t}\vct{v}}
      \,\dd\gamma_{d,\ell}(G).
  \]
  The integrand is nonnegative and continuous.  Since the Gaussian measure
  has full support,
  \[
    \ip{\vct{v}}{(GG^\dagger)^{\otimes t}\vct{v}}=0
    \qquad\text{for every \(G\in\C^{d\times\ell}\)}.
  \]
  Given a subspace \(L\subseteq\C^d\) with \(\dim L\le\ell\), choose \(G\)
  with image \(L\) such that \(GG^\dagger\) is positive definite on \(L\).
  Then \((GG^\dagger)^{\otimes t}\) is positive definite on
  \(L^{\otimes t}\), so the preceding equality implies
  \(\vct{v}\perp L^{\otimes t}\).  This holds for every such \(L\), and
  therefore \(\vct{v}\in\mathcal V_{\ell,t}^\perp\).  Together with
  \cref{eq:gaussian-support-first-inclusion}, this proves
  \(\operatorname{supp}(\Gamma_{\ell,t})=\mathcal V_{\ell,t}\).

  Every tensor permutation \(U_\pi\), \(\pi\in\mathfrak S_t\), preserves each subspace
  \(L^{\otimes t}\), and hence \(\mathcal V_{\ell,t}\).  Also,
  \[
    A^{\otimes t}L^{\otimes t}\subseteq(AL)^{\otimes t}
      \subseteq\mathcal V_{\ell,t}.
  \]
  If \(\vct{v}\in\mathcal V_{\ell,t}\) and
  \(\vct{w}\in\mathcal V_{\ell,t}^\perp\), the same inclusion with
  \(A^\dagger\) in place of \(A\) gives
  \[
    \ip{\vct{v}}{A^{\otimes t}\vct{w}}
    =\ip{(A^\dagger)^{\otimes t}\vct{v}}{\vct{w}}=0.
  \]
  Thus both \(\mathcal V_{\ell,t}\) and its orthogonal complement are
  invariant under \(A^{\otimes t}\), so \(\Pi_{\ell,t}\) commutes with
  \(A^{\otimes t}\).  The tensor permutations \(U_\pi\) are unitary and
  preserve \(\mathcal V_{\ell,t}\), so the projector \(\Pi_{\ell,t}\)
  commutes with every \(U_\pi\) as well.

  Finally, the subspaces \(\mathcal V_{\ell,t}\) are nested and
  \(\mathcal V_{\ell_{\max},t}=(\C^d)^{\otimes t}\).  The latter is immediate
  when \(d\le t\); when \(t<d\), every tensor product basis vector lies in
  \(L^{\otimes t}\) for the span \(L\) of its \(t\) tensor factors.  This
  proves \cref{eq:nested-gaussian-supports}.
\end{proof}

By \cref{eq:nested-gaussian-supports}, the operators
\(\Delta_{\ell,t}\) in \cref{eq:gaussian-support-differences} are pairwise
orthogonal projectors and sum to \(\Id_d^{\otimes t}\).  We now use this
decomposition to establish the validity of the measurement.

\begin{lemma}[Validity of the joint measurement]
  \label{lem:gaussian-block-povm-normalization}
  The densities in \cref{eq:gaussian-block-povm} form a POVM on
  \((\C^d)^{\otimes t}\).
\end{lemma}

\begin{proof}
  The operator-valued density is positive semidefinite, and for every
  \(1\le\ell\le\ell_{\max}\),
  \begin{align*}
    \int \mathsf M_t(\ell,\dd G)
    &=
    \Delta_{\ell,t}(\Gamma_{\ell,t}^+)^{1/2}
      \Gamma_{\ell,t}
      (\Gamma_{\ell,t}^+)^{1/2}\Delta_{\ell,t}
    =\Delta_{\ell,t}\Pi_{\ell,t}\Delta_{\ell,t}
     =\Delta_{\ell,t},\\
    \sum_{\ell=1}^{\ell_{\max}}\int \mathsf M_t(\ell,\dd G)
    &=\sum_{\ell=1}^{\ell_{\max}}\Delta_{\ell,t}
     =\Id_d^{\otimes t}.
  \end{align*}
\end{proof}

Integrating over \(G\) also gives the distribution of \(J\):
\begin{equation}
  \Prb_\rho[J=\ell]
  =
  \Tr\left(
    \rho^{\otimes t}\int\mathsf M_t(\ell,\dd G)
  \right)
  =\Tr(\Delta_{\ell,t}\rho^{\otimes t}).
  \label{eq:support-index-distribution}
\end{equation}

\subsubsection{Bounding the mean of \texorpdfstring{\(J\)}{J}}
\label{sec:gaussian-support-index}

The expected value of \(J\) helps bound the second moment of the estimator
in \cref{eq:gaussian-block-second-moment}, so we bound \(\E_\rho[J]\) next.
We use antisymmetrizers to bound the tail of the distribution of
\(J\) and show that \(J\) cannot exceed the rank of the state.
For \(1\le q\le\ell_{\max}\), \cref{eq:support-index-distribution} gives
\[
  \Prb_\rho[J\ge q]
  =\Tr\bigl((\Id_d^{\otimes t}-\Pi_{q-1,t})\rho^{\otimes t}\bigr).
\]
We bound \(\Id_d^{\otimes t}-\Pi_{q-1,t}\) by a sum of antisymmetrizers.
For \(S\subseteq[t]\) with \(\abs S=q\), let \(\mathfrak S_S\) be the group
of permutations of the positions in \(S\), acting trivially outside \(S\).
Define
\[
  A_S\coloneqq\frac1{q!}
    \sum_{\pi\in\mathfrak S_S}\operatorname{sgn}(\pi)U_\pi.
\]
This is the orthogonal projector onto tensors that are antisymmetric in the positions
in \(S\).

\begin{lemma}[Characterization by antisymmetrizers]
  \label{lem:gaussian-support-antisymmetrizers}
  For \(1\le q\le\min\{d,t\}\),
  \begin{equation}
    \mathcal V_{q-1,t}
    =\bigcap_{\substack{S\subseteq[t]\\\abs S=q}}\ker(A_S).
    \label{eq:gaussian-support-antisymmetrizer-kernels}
  \end{equation}
\end{lemma}

\begin{proof}
  For every subspace \(L\subseteq\C^d\) with \(\dim L\le q-1\), any \(q\)
  vectors in \(L\) are linearly dependent, so their antisymmetrization is
  zero.  Therefore,
  \[
    A_S L^{\otimes t}=\{0\}
    \qquad\text{for every \(S\subseteq[t]\) with \(\abs S=q\)}.
  \]
  Hence
  \(\mathcal V_{q-1,t}\subseteq\bigcap_{\abs S=q}\ker(A_S)\).
  Since each \(A_S\) is a projector,
  \[
    \left(\bigcap_{\substack{S\subseteq[t]\\\abs S=q}}\ker(A_S)\right)^\perp
    =\sum_{\substack{S\subseteq[t]\\\abs S=q}}\ker(A_S)^\perp
    =\sum_{\substack{S\subseteq[t]\\\abs S=q}}\operatorname{Im}(A_S).
  \]
  Thus, taking orthogonal complements, it remains to show
  \begin{equation}
    \mathcal V_{q-1,t}^\perp
    \subseteq\sum_{\substack{S\subseteq[t]\\\abs S=q}}
      \operatorname{Im}(A_S).
    \label{eq:gaussian-support-antisymmetrizer-ranges}
  \end{equation}
  
  Let \(X \eqdef [\vct{x}_1\ \cdots\ \vct{x}_t]=(x_{ij})\) be a \(d\times t\)
  matrix of indeterminates, and consider the polynomial ring
  \(R\coloneqq\C[x_{ij}:i\in[d],\ j\in[t]]\).
  
  Associate with any vector \(\vct{w} \in (\complex^d)^{\tensor t} \) the polynomial
  \[
    F_{\vct{w}}(X)
    \coloneqq
    \ip{\vct{w}}{\vct{x}_1\otimes\cdots\otimes\vct{x}_t}.
  \]
  The polynomial $F_{\vct{w}}$ is homogeneous, multilinear, and has total degree one in the variables in any fixed column of \(X\). We refer to the latter property as having ``degree one in each column'' of~$X$.
  
  Consider \(\vct{v}\in\mathcal V_{q-1,t}^\perp\).  
  For any choice of the columns in \(\C^d\), we have
  \[
    \rank(X)\le q-1
    \quad\Longrightarrow\quad
    \vct{x}_1\otimes\cdots\otimes\vct{x}_t
      \in\mathcal V_{q-1,t}
    \quad\Longrightarrow\quad
    F_{\vct{v}}(X)=0.
  \]

  For \(T\subseteq[d]\) and \(S\subseteq[t]\) with
  \(\abs T=\abs S=q\), let \(D_{T,S}(X)\) be the determinant of the
  submatrix of~$X$ with rows \(T\) and columns \(S\), both in increasing order.
  Let \(I_{q-1}\subset R\) be the ideal generated by these determinants:
  its elements are sums of the polynomials \(D_{T,S}\) multiplied by arbitrary
  polynomials in \(R\).  Its common zero set is
  \begin{align*}
    V(I_{q-1})
    &\coloneqq\{M\in\C^{d\times t}:p(M)=0
      \text{ for every }p\in I_{q-1}\}\\
    &=\{M\in\C^{d\times t}:\rank(M)\le q-1\},
  \end{align*}
  since a matrix has rank at most \(q-1\) if and only if the determinants of all its
  \(q\times q\) minors vanish.  This set of matrices is a
  \emph{determinantal variety\/}, while \(I_{q-1}\) is the corresponding
  \emph{determinantal ideal\/} of polynomials.
  
  Since~$\vct{v} \in \mathcal V_{q-1,t}^\perp$, polynomial \(F_{\vct{v}}\) vanishes
  on \(V(I_{q-1})\).  Since \(\C\) is algebraically closed, the Hilbert
  Nullstellensatz~\cite[Ch.~VII, Theorem~14]{ZariskiSamuel1960} gives
  \[
    F_{\vct{v}}^{\,j}\in I_{q-1}
    \qquad\text{for some integer \(j\ge1\)}.
  \]
  The ideal \(I_{q-1}\) is prime
  \cite[Theorem~2.10 and Remark~2.12]{BrunsVetter1988}. I.e., it is a proper
  ideal, and whenever it contains a product of two polynomials, it also contains at least
  one of the two factors.  Applying this property repeatedly to
  \(F_{\vct{v}}^{\,j}\), we get \(F_{\vct{v}}\in I_{q-1}\).

  We can therefore write \(F_{\vct{v}}\) as a sum 
  \begin{equation*}
    F_{\vct{v}}
    =\sum_{\substack{S\subseteq[t],\ \abs S=q\\
                     T\subseteq[d],\ \abs T=q}}
      D_{T,S} \tH_{T,S},
  \end{equation*}
where~$\tH_{T,S} \in R$ is some polynomial over~$X$. Recall that~$D_{T,S}$ is homogeneous of degree~$q$, multilinear, and has degree one in the columns in~$S$. Decompose each polynomial~$\tH_{T,S}$ as a sum~$\tH_{T,S} = H_{T,S} + K_{T,S}$, where~$H_{T,S}$ has all the homogeneous monomials of~$\tH_{T,S}$ of degree~$t - q$ and with degree one in each column in~$[t] \setminus S$.
Since~$F_{\vct{v}}$ is a homogeneous multilinear polynomial with degree one in each column of~$X$, we have~$\sum_{S,T} D_{T,S} K_{T,S} = 0$ and
  \begin{equation}
    F_{\vct{v}}
    =\sum_{\substack{S\subseteq[t],\ \abs S=q\\
                     T\subseteq[d],\ \abs T=q}}
      D_{T,S} H_{T,S}.
    \label{eq:determinantal-multigraded-decomposition}
  \end{equation}
  Observe that~\(H_{T,S}\) is independent of the columns of~$X$ in \(S\).
  Permuting the columns of~$X$ in~$S$ multiplies \(D_{T,S}\) by the sign of the
  permutation and leaves \(H_{T,S}\) unchanged.  Hence every product
  \(D_{T,S} H_{T,S}\) is alternating in the columns indexed by \(S\).

  Every homogeneous polynomial in~$R$ of degree one in each column has a unique
  representation as a polynomial~\(F_{\vct{w}}\) for some~$\vct{w} \in (\complex^d)^{\tensor t}$. The coefficients of the polynomial are the complex
  conjugates of the coordinates of \(\vct{w}\) in the standard tensor
  product basis.  Let \(\vct{w}_{T,S}\) be the tensor satisfying
  \[
    \ip{\vct{w}_{T,S}}{\vct{x}_1\otimes\cdots\otimes\vct{x}_t}
    =D_{T,S}(X)H_{T,S}(X).
  \]
  For \(\pi\in\mathfrak S_S\), unitarity of \(U_\pi\) and alternation give
  \begin{align*}
    \ip{U_\pi\vct{w}_{T,S}}{\vct{x}_1\otimes\cdots\otimes\vct{x}_t}
    &=\ip{\vct{w}_{T,S}}
      {\vct{x}_{\pi(1)}\otimes\cdots\otimes\vct{x}_{\pi(t)}}\\
    &=\operatorname{sgn}(\pi)
      \ip{\vct{w}_{T,S}}{\vct{x}_1\otimes\cdots\otimes\vct{x}_t}.
  \end{align*}
  Since simple tensors span \((\C^d)^{\otimes t}\) and the sign is real,
  \(U_\pi\vct{w}_{T,S}=\operatorname{sgn}(\pi)\vct{w}_{T,S}\) for every
  \(\pi\in\mathfrak S_S\).  Averaging with the signs in \(A_S\) gives
  \(A_S\vct{w}_{T,S}=\vct{w}_{T,S}\), so
  \(\vct{w}_{T,S}\in\operatorname{Im}(A_S)\).  Thus
  \cref{eq:determinantal-multigraded-decomposition} and uniqueness of the
  associated tensor imply
  \(\vct{v}=\sum_{S,T}\vct{w}_{T,S}\in
    \sum_{\abs S=q}\operatorname{Im}(A_S)\). This proves
  \cref{eq:gaussian-support-antisymmetrizer-ranges} and the lemma.
\end{proof}

\begin{lemma}[Mean of \(J\)]
  \label{lem:gaussian-support-index-mean}
  For every state \(\rho\in\cD(\C^d)\) and every integer \(q\ge1\),
  \begin{equation}
    \Prb_\rho[J\ge q]
    \le\min\left\{1,\frac{\binom tq}{q!}\right\},
    \label{eq:gaussian-support-index-tail}
  \end{equation}
  where the probability is zero for \(q>\min\{d,t\}\).  Consequently,
  \begin{equation}
    \E_\rho[J]=O(\sqrt t).
    \label{eq:gaussian-support-index-mean}
  \end{equation}
  If \(\rank(\rho)\le r\), then every value of \(J\) with positive
  probability is at most \(r\).
\end{lemma}

\begin{proof}
  Fix \(1\le q\le\min\{d,t\}\).  By the preceding lemma, the intersection
  of the kernels of the order-\(q\) antisymmetrizers is
  \(\mathcal V_{q-1,t}\).  Define
  \[
    \mathcal A_{t,q}
    \coloneqq\sum_{\substack{S\subseteq[t]\\\abs S=q}}A_S.
  \]
  Conjugating \(\mathcal A_{t,q}\) by a tensor permutation \(U_\pi\)
  permutes the summands, so \(\mathcal A_{t,q}\) commutes with every
  \(U_\pi\), \(\pi\in\mathfrak S_t\).
  Since each \(A_S\) is a linear combination of such permutations,
  \(\mathcal A_{t,q}\) commutes with every \(A_S\).  Therefore,
  \begin{equation}
    \mathcal A_{t,q}^2
    =\sum_SA_S\mathcal A_{t,q}
    =\sum_SA_S\mathcal A_{t,q}A_S
    =\sum_{S,T}A_SA_TA_S
    \succeq\sum_SA_S
    =\mathcal A_{t,q}.
    \label{eq:antisymmetrizer-unit-gap}
  \end{equation}
  Here \(A_SA_TA_S=(A_TA_S)^\dagger(A_TA_S)\succeq0\), and the term with
  \(T=S\) equals \(A_S\).  Moreover,
  \[
    \ker(\mathcal A_{t,q})
    =\bigcap_{\abs S=q}\ker(A_S)
    =\mathcal V_{q-1,t}.
  \]
  Since \(\mathcal A_{t,q}\succeq0\),
  \cref{eq:antisymmetrizer-unit-gap} implies that each of its nonzero
  eigenvalues is at least one.  Hence
  \begin{equation}
    \mathcal A_{t,q}\succeq\Id_d^{\otimes t}-\Pi_{q-1,t}.
    \label{eq:antisymmetrizer-support-bound}
  \end{equation}

  Since
  \[
    \sum_{\ell=q}^{\ell_{\max}}\Delta_{\ell,t}
    =\Id_d^{\otimes t}-\Pi_{q-1,t},
  \]
  the tail of \(J\) is controlled by
  \cref{eq:antisymmetrizer-support-bound}.  Define the antisymmetrizer on
  \((\C^d)^{\otimes q}\) by
  \[
    A_q\coloneqq\frac1{q!}
      \sum_{\pi\in\mathfrak S_q}\operatorname{sgn}(\pi)U_\pi.
  \]
  For every \(S\subseteq[t]\) of size \(q\),
  \[
    \Tr(A_S\rho^{\otimes t})=\Tr(A_q\rho^{\otimes q}).
  \]
  Therefore, \cref{eq:support-index-distribution,eq:antisymmetrizer-support-bound}
  give
  \begin{align}
    \Prb_\rho[J\ge q]
    &=\sum_{\ell=q}^{\ell_{\max}}
      \Tr(\Delta_{\ell,t}\rho^{\otimes t}) \notag\\
    &=\Tr\bigl((\Id_d^{\otimes t}-\Pi_{q-1,t})\rho^{\otimes t}\bigr) \notag\\
    &\le\Tr(\mathcal A_{t,q}\rho^{\otimes t})
     =\binom tq\Tr(A_q\rho^{\otimes q}).
    \label{eq:support-index-tail-antisymmetrizer}
  \end{align}
  Let \(\lambda_1,\ldots,\lambda_d\) be the eigenvalues of \(\rho\).  Since
  \(A_q\) projects onto the totally antisymmetric subspace, diagonalizing
  \(\rho\) gives an eigenbasis for the restriction of \(\rho^{\otimes q}\)
  to this subspace, indexed by \(1\le i_1<\cdots<i_q\le d\), with corresponding eigenvalue
  \(\lambda_{i_1}\cdots\lambda_{i_q}\).  Hence
  \[
    \Tr(A_q\rho^{\otimes q})
    =\sum_{1\le i_1<\cdots<i_q\le d}
      \lambda_{i_1}\cdots\lambda_{i_q}
    \le\frac1{q!}\left(\sum_{i=1}^d\lambda_i\right)^q
    =\frac1{q!}.
  \]
  The inequality holds because, in the expansion of
  \((\sum_i\lambda_i)^q\), every product with \(q\) distinct indices
  appears \(q!\) times, and all remaining terms are nonnegative.
  Together with
  \cref{eq:support-index-tail-antisymmetrizer}, this proves
  \cref{eq:gaussian-support-index-tail}.

  The standard estimates
  \(\binom tq\le(\ee t/q)^q\) and \(q!\ge(q/\ee)^q\) give
  \[
    \frac{\binom tq}{q!}
    \le\left(\frac{\ee^2t}{q^2}\right)^q.
  \]
  For \(q\ge2\ee\sqrt t\), the right side is at most \(4^{-q}\).
  The tail sum identity
  \(\E_\rho[J]=\sum_{q\ge1}\Prb_\rho[J\ge q]\) therefore gives
  \[
    \E_\rho[J]
    \le \left\lceil2\ee\sqrt t\right\rceil
      +\sum_{q\ge\lceil2\ee\sqrt t\rceil}4^{-q}
    =O(\sqrt t),
  \]
  proving \cref{eq:gaussian-support-index-mean}.

  Finally, let \(L\) be the support of a state \(\rho\) of rank at most \(r\).
  If \(r<\ell_{\max}\), then
  \[
    \operatorname{supp}(\rho^{\otimes t})
      \subseteq L^{\otimes t}
      \subseteq\mathcal V_{r,t},
    \qquad\text{and}\qquad
    \Delta_{\ell,t}\Pi_{r,t}=0\quad\text{when \(\ell>r\)}.
  \]
  So \cref{eq:support-index-distribution} gives
  \(\Prb_\rho[J=\ell]=0\) for \(\ell>r\).  If \(r\ge\ell_{\max}\), this
  follows directly from \(J\le\ell_{\max}\le r\).
\end{proof}

\subsubsection{The outcome law and its moments}
\label{sec:gaussian-block-moments}

It remains to prove the claims about the moments and conditional distribution in
\Cref{prop:gaussian-block-estimator}.  We first derive the Gaussian moment
identities used to evaluate the first and second moments.  We then apply the
Born rule to obtain the contribution from each event \(J=\ell\), sum these
contributions over \(\ell\), and finally prove the claimed conditional law
of \(\Pi^\perp G\).

\paragraph{The Gaussian moment identity.}

For \(\pi\in\mathfrak S_t\), let \(c(\pi)\) be the number of cycles of
\(\pi\), including fixed points.

\begin{lemma}[Gaussian moment identity]
  \label{lem:gaussian-moment-identity}
  For every \(1\le\ell\le\ell_{\max}\),
  \begin{equation}
    \Gamma_{\ell,t}
    =\sum_{\pi\in\mathfrak S_t}
      \ell^{c(\pi)}U_\pi.
    \label{eq:gaussian-moment-identity}
  \end{equation}
  The operators \((\Gamma_{\ell,t}:1\le\ell\le\ell_{\max})\) commute pairwise.  Each
  \(\Gamma_{\ell,t}\) also commutes with every tensor permutation
  \(U_\pi\), \(\pi\in\mathfrak S_t\),
  and with \(A^{\otimes t}\) for every operator \(A\) on \(\C^d\).
\end{lemma}

\begin{proof}
  We compute an arbitrary matrix entry of \(\Gamma_{\ell,t}\).  Let
  \(\vct{e}_1,\ldots,\vct{e}_d\) be the standard basis of \(\C^d\).  For
  \(\mathbf i=(i_1,\ldots,i_t)\) and
  \(\mathbf j=(j_1,\ldots,j_t)\) in \([d]^t\), put
  \[
    \vct{e}_{\mathbf i}
    \coloneqq\vct{e}_{i_1}\otimes\cdots\otimes\vct{e}_{i_t},
    \qquad
    \vct{e}_{\mathbf j}
    \coloneqq\vct{e}_{j_1}\otimes\cdots\otimes\vct{e}_{j_t}.
  \]
  The entries of \(G\) satisfy
  \[
    \E[G_{i\alpha}G_{j\beta}]=0,
    \qquad
    \E[G_{i\alpha}\overline{G_{j\beta}}]
    =\delta_{ij}\delta_{\alpha\beta}.
  \]
  These identities hold for all \(i,j\in[d]\) and
  \(\alpha,\beta\in[\ell]\).
  Expanding the matrix entry and applying the complex form of Isserlis's
  Gaussian pairing formula \cite{Isserlis1918} gives the following
  calculation.  The formula applies to the
  jointly Gaussian entries even when some indices coincide.  Pairings
  between two unconjugated entries or two conjugated entries vanish, so
  every surviving pairing matches the \(t\) unconjugated factors with the
  \(t\) conjugated factors.  The matching is given by a permutation in
  \(\mathfrak S_t\).
  \begin{align}
    \ip{\vct{e}_{\mathbf i}}{
      \Gamma_{\ell,t}\vct{e}_{\mathbf j}}
    &=\E\left[\prod_{u=1}^t(GG^\dagger)_{i_u j_u}\right] \notag\\
    &=\sum_{\alpha_1,\ldots,\alpha_t\in[\ell]}
      \E\left[
        \prod_{u=1}^t
          G_{i_u\alpha_u}\overline{G_{j_u\alpha_u}}
      \right] \notag\\
    &=\sum_{\alpha_1,\ldots,\alpha_t\in[\ell]}
      \sum_{\pi\in\mathfrak S_t}
      \prod_{u=1}^t
        \E\left[
          G_{i_u\alpha_u}
          \overline{G_{j_{\pi(u)}\alpha_{\pi(u)}}}
        \right]
      \qquad\text{(Isserlis's theorem)} \notag\\
    &=\sum_{\alpha_1,\ldots,\alpha_t\in[\ell]}
      \sum_{\pi\in\mathfrak S_t}
      \prod_{u=1}^t
        \delta_{i_u,j_{\pi(u)}}
        \delta_{\alpha_u,\alpha_{\pi(u)}} \notag\\
    &=\sum_{\pi\in\mathfrak S_t}
      \ell^{c(\pi)}
      \prod_{u=1}^t\delta_{i_u,j_{\pi(u)}},
    \label{eq:entrywise-gaussian-moment}
  \end{align}
  since the constraints \(\alpha_u=\alpha_{\pi(u)}\) make the column indices
  constant on each cycle of \(\pi\), giving \(\ell^{c(\pi)}\) choices.
  Our convention for tensor permutations gives
  \[
    \ip{\vct{e}_{\mathbf i}}{U_\pi\vct{e}_{\mathbf j}}
    =\prod_{u=1}^t\delta_{i_u,j_{\pi^{-1}(u)}}.
  \]
  Since inversion preserves the number of cycles and permutes
  \(\mathfrak S_t\),
  \[
    \ip{\vct{e}_{\mathbf i}}{
      \left(\sum_{\pi\in\mathfrak S_t}
        \ell^{c(\pi)}U_\pi\right)\vct{e}_{\mathbf j}}
    =\sum_{\pi\in\mathfrak S_t}
      \ell^{c(\pi)}
      \prod_{u=1}^t\delta_{i_u,j_{\pi^{-1}(u)}}
    =\sum_{\pi\in\mathfrak S_t}
      \ell^{c(\pi)}
      \prod_{u=1}^t\delta_{i_u,j_{\pi(u)}}.
  \]
  Comparing this with
  \cref{eq:entrywise-gaussian-moment} for every
  \(\mathbf i,\mathbf j\in[d]^t\) proves
  \cref{eq:gaussian-moment-identity}.

  We next use this identity to prove the commutation claims.  For every
  \(\tau\in\mathfrak S_t\), conjugation preserves cycle type and gives
  \[
    U_\tau\Gamma_{\ell,t}U_\tau^\dagger
    =\sum_{\pi\in\mathfrak S_t}
      \ell^{c(\pi)}U_{\tau\pi\tau^{-1}}
    =\Gamma_{\ell,t}.
  \]
  Thus \(\Gamma_{\ell,t}\) commutes with every \(U_\tau\).  Because every
  \(\Gamma_{\ell',t}\) is a linear combination of the \(U_\tau\), the
  operators \(\Gamma_{\ell,t}\) commute pairwise for different values of
  \(\ell\).  Finally,
  \[
    U_\pi A^{\otimes t}=A^{\otimes t}U_\pi
    \quad\Longrightarrow\quad
    \Gamma_{\ell,t}A^{\otimes t}
    =A^{\otimes t}\Gamma_{\ell,t}.
  \]
\end{proof}

\paragraph{Centered Gaussian moments.}

For a fixed value \(J=\ell\), the contribution to the first moment of the
estimator involves Gaussian integrals of
\[
  \frac1t(GG^\dagger-\ell\Id_d)\otimes(GG^\dagger)^{\otimes t}.
\]
To separate the estimator matrix from the operators acting on the \(t\)
measured samples, we temporarily work on the labelled tensor product space
\[
  \C^d_a\otimes\C^d_1\otimes\cdots\otimes\C^d_t,
\]
where the subscripts label the tensor factors.  The factors
\(1,\ldots,t\) correspond to the measured samples, while \(a\) is
an auxiliary factor carrying the matrix output of the estimator.  Thus, if
\(B\) is an output matrix and \(A,C\) act on the \(t\) input factors,
then
\[
  B\Tr(AC)
  =\Tr_{[t]}\bigl((\Id_a\otimes A)(B_a\otimes C)\bigr).
\]
Here \(B_a\) means that \(B\) acts on factor \(a\), and \(\Tr_{[t]}\) traces
out factors \(1,\ldots,t\).
After factoring out \(1/t\), taking \(B=GG^\dagger-\ell\Id_d\) and
\(C=(GG^\dagger)^{\otimes t}\) gives the first Gaussian integral below.
For the second moment, we use two auxiliary ``output'' tensor factors
\(a\) and \(b\),
one for each occurrence of the estimator matrix in its tensor second moment.
Write \(F_{uv}\) for the operator that swaps tensor factors \(u\) and
\(v\) and acts as the identity on the remaining factors, and
\(\Id_{ab}\) for the identity on the two output factors.

\begin{lemma}[Centered Gaussian moments]
  \label{lem:centered-gaussian-moments}
  Fix \(1\le\ell\le\ell_{\max}\).  Then
  \begin{equation}
    \int
      (GG^\dagger-\ell\Id_d)_a\otimes(GG^\dagger)^{\otimes t}
      \,\dd\gamma_{d,\ell}(G)
    =\sum_{i=1}^t F_{ai}(\Id_a\otimes\Gamma_{\ell,t}).
    \label{eq:one-output-centered-gaussian-moment}
  \end{equation}
  With two output factors,
  \begin{align}
    &\int
      (GG^\dagger-\ell\Id_d)_a\otimes(GG^\dagger-\ell\Id_d)_b
        \otimes(GG^\dagger)^{\otimes t}
      \,\dd\gamma_{d,\ell}(G) \notag\\
    &\quad=
      \sum_{\substack{i,j\in[t]\\i\ne j}}
        F_{ai}F_{bj}(\Id_{ab}\otimes\Gamma_{\ell,t})
      +\sum_{i=1}^t
        \bigl(F_{ab}F_{bi}+F_{ab}F_{ai}\bigr)
        (\Id_{ab}\otimes\Gamma_{\ell,t})
      +\ell F_{ab}(\Id_{ab}\otimes\Gamma_{\ell,t}).
    \label{eq:two-output-centered-gaussian-moment}
  \end{align}
  When \(t=1\), the first sum in
  \cref{eq:two-output-centered-gaussian-moment} is empty and equals zero.
\end{lemma}

\begin{proof}
  For \(\pi\in\mathfrak S_t\), let \(\widetilde\pi\) be its extension that
  fixes \(a\), and write \(\mathfrak S_{\{a\}\cup[t]}\) for the permutations
  of the labels \(a,1,\ldots,t\).  Applying the Gaussian moment identity in
  \cref{eq:gaussian-moment-identity} with the output position included gives
  \begin{align*}
    &\int
      (GG^\dagger-\ell\Id_d)_a\otimes(GG^\dagger)^{\otimes t}
      \,\dd\gamma_{d,\ell}(G)\\
    &\quad=
      \int (GG^\dagger)_a\otimes(GG^\dagger)^{\otimes t}
        \,\dd\gamma_{d,\ell}(G)
      -\ell\Id_a\otimes
        \int(GG^\dagger)^{\otimes t}\,\dd\gamma_{d,\ell}(G)\\
    &\quad=
      \sum_{\sigma\in\mathfrak S_{\{a\}\cup[t]}}
        \ell^{c(\sigma)}U_\sigma
      -\sum_{\pi\in\mathfrak S_t}
        \ell^{c(\pi)+1}(\Id_a\otimes U_\pi)\\
    &\quad=
      \sum_{\substack{\sigma\in\mathfrak S_{\{a\}\cup[t]}\\
                       \sigma(a)\ne a}}
        \ell^{c(\sigma)}U_\sigma\\
    &\quad=
      \sum_{i=1}^t\sum_{\pi\in\mathfrak S_t}
        \ell^{c(\pi)}F_{ai}(\Id_a\otimes U_\pi)\\
    &\quad=
      \sum_{i=1}^tF_{ai}(\Id_a\otimes\Gamma_{\ell,t}),
  \end{align*}
  where the third equality cancels exactly the permutations that fix
  \(a\): each is the extension \(\widetilde\pi\) of a unique
  \(\pi\in\mathfrak S_t\), with
  \(c(\widetilde\pi)=c(\pi)+1\) and
  \(U_{\widetilde\pi}=\Id_a\otimes U_\pi\).
  For the fourth equality, if \(\sigma(a)\ne a\), set
  \(i=\sigma(a)\).  Then \((a\ i)\sigma\) fixes \(a\), giving the
  unique representation \(\sigma=(a\ i)\widetilde\pi\).
  This inserts \(a\) immediately before \(i\) in its cycle of \(\pi\),
  so \(c(\sigma)=c(\pi)\) and
  \(U_\sigma=F_{ai}(\Id_a\otimes U_\pi)\).
  This proves \cref{eq:one-output-centered-gaussian-moment}.

  With two outputs, expanding both centered factors and applying the same
  Gaussian moment identity gives
  \begin{align*}
    &\int
      (GG^\dagger-\ell\Id_d)_a\otimes
      (GG^\dagger-\ell\Id_d)_b\otimes(GG^\dagger)^{\otimes t}
      \,\dd\gamma_{d,\ell}(G)\\
    &\quad=
      \int (GG^\dagger)_a\otimes(GG^\dagger)_b
        \otimes(GG^\dagger)^{\otimes t}\,\dd\gamma_{d,\ell}(G)\\
    &\qquad
      -\ell\int \Id_a\otimes(GG^\dagger)_b
        \otimes(GG^\dagger)^{\otimes t}\,\dd\gamma_{d,\ell}(G)
      -\ell\int (GG^\dagger)_a\otimes\Id_b
        \otimes(GG^\dagger)^{\otimes t}\,\dd\gamma_{d,\ell}(G)\\
    &\qquad
      +\ell^2\Id_{ab}\otimes
        \int(GG^\dagger)^{\otimes t}\,\dd\gamma_{d,\ell}(G)\\
    &\quad=
      \sum_{\sigma\in\mathfrak S_{\{a,b\}\cup[t]}}
        \ell^{c(\sigma)}U_\sigma
      -\sum_{\substack{\sigma\in\mathfrak S_{\{a,b\}\cup[t]}\\
                        \sigma(a)=a}}
        \ell^{c(\sigma)}U_\sigma
      -\sum_{\substack{\sigma\in\mathfrak S_{\{a,b\}\cup[t]}\\
                        \sigma(b)=b}}
        \ell^{c(\sigma)}U_\sigma
      +\sum_{\substack{\sigma\in\mathfrak S_{\{a,b\}\cup[t]}\\
                        \sigma(a)=a,\ \sigma(b)=b}}
        \ell^{c(\sigma)}U_\sigma\\
    &\quad=
      \sum_{\substack{\sigma\in\mathfrak S_{\{a,b\}\cup[t]}\\
                       \sigma(a)\ne a,\ \sigma(b)\ne b}}
        \ell^{c(\sigma)}U_\sigma.
  \end{align*}
  Deleting \(a\) and \(b\) from the cycles of a permutation \(\sigma\)
  in the remaining sum leaves a unique permutation
  \(\pi\in\mathfrak S_t\).  To recover \(\sigma\), there are four
  possibilities: insert them before two
  distinct input labels \(i\ne j\); insert them consecutively before one input
  label \(i\), in either order; or place them together in the two-cycle
  \((a\ b)\).  Here inserting \(a\) before \(i\) replaces an arrow
  \(k\to i\) in a cycle of \(\pi\) by \(k\to a\to i\).  Writing
  \(\widetilde\pi\) for the extension of \(\pi\) that fixes \(a\) and \(b\),
  these possibilities are listed below.
  \[
    \begin{array}{c|c|c|c}
      \text{configuration}&\sigma&c(\sigma)&U_\sigma\\ \hline
      \text{distinct positions }i\ne j
        &(a\ i)(b\ j)\widetilde\pi&c(\pi)
        &F_{ai}F_{bj}(\Id_{ab}\otimes U_\pi)\\
      a\to b\to i
        &(a\ b)(b\ i)\widetilde\pi&c(\pi)
        &F_{ab}F_{bi}(\Id_{ab}\otimes U_\pi)\\
      b\to a\to i
        &(a\ b)(a\ i)\widetilde\pi&c(\pi)
        &F_{ab}F_{ai}(\Id_{ab}\otimes U_\pi)\\
      \text{isolated two-cycle }(a\ b)
        &(a\ b)\widetilde\pi&c(\pi)+1
        &F_{ab}(\Id_{ab}\otimes U_\pi)
    \end{array}
  \]
  Summing over \(\pi\in\mathfrak S_t\) in each row gives
  \begin{align*}
    &\sum_{\substack{\sigma\in\mathfrak S_{\{a,b\}\cup[t]}\\
                     \sigma(a)\ne a,\ \sigma(b)\ne b}}
      \ell^{c(\sigma)}U_\sigma\\
    &\quad=
      \sum_{\substack{i,j\in[t]\\i\ne j}}
        F_{ai}F_{bj}(\Id_{ab}\otimes\Gamma_{\ell,t})
      +\sum_{i=1}^t
        \bigl(F_{ab}F_{bi}+F_{ab}F_{ai}\bigr)
        (\Id_{ab}\otimes\Gamma_{\ell,t})
      +\ell F_{ab}(\Id_{ab}\otimes\Gamma_{\ell,t}),
  \end{align*}
  which proves \cref{eq:two-output-centered-gaussian-moment}.

  When \(t=1\), there are no distinct input labels \(i\ne j\).  The only
  permutations of \(\{a,b,1\}\) that fix neither \(a\) nor \(b\) are
  \((a\ b\ 1)\), \((a\ 1\ b)\), and \((a\ b)\).  Since
  \(\ell_{\max}=1\) and \(\Gamma_{1,1}=\Id_d\), their contributions give
  \[
    \int
      (GG^\dagger-\Id_d)_a\otimes(GG^\dagger-\Id_d)_b
        \otimes(GG^\dagger)_1\,\dd\gamma_{d,1}(G)
    =F_{ab}F_{b1}+F_{ab}F_{a1}+F_{ab}.
  \]
  Thus only the last two terms on the right side of
  \cref{eq:two-output-centered-gaussian-moment} remain.
\end{proof}

\paragraph{The contribution from a fixed value of \(J\).}

We now use the preceding Gaussian identities to compute the contribution of
the event \(J=\ell\) to the moments of the estimator.  Fix a state \(\rho\).
For \(1\le\ell\le\ell_{\max}\), define
\[
  \omega_{\ell,t}(\rho)
  \coloneqq\Delta_{\ell,t}\rho^{\otimes t},
  \qquad
  \Theta_{\ell,\rho}
  \coloneqq
  (\Gamma_{\ell,t}^+)^{1/2}\omega_{\ell,t}(\rho)
  (\Gamma_{\ell,t}^+)^{1/2}.
\]
By \Cref{lem:gaussian-moment-support}, \(\Delta_{\ell,t}\) commutes with
\(\rho^{\otimes t}\), so \(\omega_{\ell,t}(\rho)\) is the unnormalized
positive semidefinite component of \(\rho^{\otimes t}\) on
\(\operatorname{Im}(\Delta_{\ell,t})\).  Its trace is
\[
  \Tr\bigl(\omega_{\ell,t}(\rho)\bigr)
  =\Tr\bigl(\Delta_{\ell,t}\rho^{\otimes t}\bigr)
  =\Prb_\rho[J=\ell],
\]
by \cref{eq:support-index-distribution}.  The operator
\(\Theta_{\ell,\rho}\) incorporates the pseudoinverse factors from the
POVM so that the density of outcomes with \(J=\ell\), relative to
\(\gamma_{d,\ell}\), is
\(\Tr(\Theta_{\ell,\rho}(GG^\dagger)^{\otimes t})\).

\begin{lemma}[Contribution from a fixed value of \(J\)]
  \label{lem:fixed-index-moments}
  Fix a state \(\rho\) and \(1\le\ell\le\ell_{\max}\).  On the event
  \(J=\ell\), we have \(Y=(GG^\dagger-\ell\Id_d)/t\).  If
  \(\mathbf 1_{\{J=\ell\}}\) denotes the indicator of this event, then
  \begin{equation}
    \E_\rho\left[Y\mathbf 1_{\{J=\ell\}}\right]
    =\Tr_{[t]\setminus\{1\}}\bigl(\omega_{\ell,t}(\rho)\bigr)
    \label{eq:fixed-index-first-moment}
  \end{equation}
  and
  \begin{equation}
    \begin{aligned}
      \E_\rho\left[
        (Y\otimes Y)\mathbf 1_{\{J=\ell\}}
      \right]
      &=\frac{t-1}{t}
        \Tr_{[t]\setminus\{1,2\}}\bigl(\omega_{\ell,t}(\rho)\bigr)\\
      &\quad+\frac1t\left(
          \Tr_{[t]\setminus\{1\}}\bigl(\omega_{\ell,t}(\rho)\bigr)
            \otimes\Id_d
          +\Id_d\otimes
          \Tr_{[t]\setminus\{1\}}\bigl(\omega_{\ell,t}(\rho)\bigr)
        \right)F\\
      &\quad+\frac{\ell}{t^2}\Prb_\rho[J=\ell]F,
    \end{aligned}
    \label{eq:fixed-index-second-moment}
  \end{equation}
  where \(F\) is the swap operator on \(\C^d\otimes\C^d\).
\end{lemma}

\begin{proof}
  The outcome law on the event \(J=\ell\) has density
  \begin{align}
    p_{\ell,\rho}(G)
    &\coloneqq\Tr\left(
      \rho^{\otimes t}\Delta_{\ell,t}(\Gamma_{\ell,t}^+)^{1/2}
      (GG^\dagger)^{\otimes t}(\Gamma_{\ell,t}^+)^{1/2}\Delta_{\ell,t}
    \right) \notag\\
    &=\Tr\left(
      \omega_{\ell,t}(\rho)(\Gamma_{\ell,t}^+)^{1/2}
      (GG^\dagger)^{\otimes t}(\Gamma_{\ell,t}^+)^{1/2}
    \right) \notag\\
    &=\Tr\left(
      \Theta_{\ell,\rho}(GG^\dagger)^{\otimes t}
    \right)
    \label{eq:fixed-index-outcome-law}
  \end{align}
  relative to \(\gamma_{d,\ell}\).  This density has total mass
  \(\Prb_\rho[J=\ell]\); dividing by this probability, when positive, gives
  the conditional density of \(G\).

  Combining
  \cref{eq:fixed-index-outcome-law,eq:one-output-centered-gaussian-moment},
  we rewrite the contribution to the first moment as
  \[
    \E_\rho\left[Y\mathbf 1_{\{J=\ell\}}\right]
    =\frac1t\sum_{i=1}^t
      \Tr_{[t]}\left(
        (\Id_a\otimes\Theta_{\ell,\rho})F_{ai}
        (\Id_a\otimes\Gamma_{\ell,t})
      \right).
  \]
  To simplify each summand, cyclicity on the traced-out input factors gives
  \[
    \Tr_{[t]}\left(
      (\Id_a\otimes\Theta_{\ell,\rho})F_{ai}
      (\Id_a\otimes\Gamma_{\ell,t})
    \right)
    =\Tr_{[t]}\left(
      F_{ai}(\Id_a\otimes\Gamma_{\ell,t}\Theta_{\ell,\rho})
    \right).
  \]
  Thus we need to evaluate
  \(\Gamma_{\ell,t}\Theta_{\ell,\rho}\).  By
  \Cref{lem:gaussian-moment-identity}, \(\Gamma_{\ell,t}\) commutes with
  \(\rho^{\otimes t}\) and, when \(\ell>1\), with
  \(\Gamma_{\ell-1,t}\), hence with its support projector
  \(\Pi_{\ell-1,t}\).  The case \(\ell=1\) is immediate because
  \(\Pi_{0,t}=0\).  It also commutes with its own support projector
  \(\Pi_{\ell,t}\).  Thus \(\Gamma_{\ell,t}\) and
  \((\Gamma_{\ell,t}^+)^{1/2}\) commute with
  \(\Delta_{\ell,t}\) and \(\omega_{\ell,t}(\rho)\).  Moreover,
  \(\Delta_{\ell,t}\preceq\Pi_{\ell,t}\).  Hence
  \[
    \Gamma_{\ell,t}\Theta_{\ell,\rho}
    =\Gamma_{\ell,t}\Gamma_{\ell,t}^+\omega_{\ell,t}(\rho)
    =\Pi_{\ell,t}\omega_{\ell,t}(\rho)
    =\omega_{\ell,t}(\rho).
  \]
  Consequently, each summand is
  \[
    \Tr_{[t]}\left(
      F_{ai}(\Id_a\otimes\omega_{\ell,t}(\rho))
    \right).
  \]
  The swap \(F_{ai}\) moves the \(i\)-th input factor to the output factor
  \(a\) before the input factors are traced out.  Therefore,
  \[
    \Tr_{[t]}\left(
      F_{ai}(\Id_a\otimes\omega_{\ell,t}(\rho))
    \right)
    =\Tr_{[t]\setminus\{i\}}\bigl(\omega_{\ell,t}(\rho)\bigr)
    =\Tr_{[t]\setminus\{1\}}\bigl(\omega_{\ell,t}(\rho)\bigr),
  \]
  where the last equality follows from permutation invariance.  All \(t\)
  summands are therefore equal, proving
  \cref{eq:fixed-index-first-moment}.

  For the second moment, substituting
  \cref{eq:two-output-centered-gaussian-moment} and following the same
  simplification gives
  \begin{align*}
    \E_\rho\left[
      (Y\otimes Y)\mathbf 1_{\{J=\ell\}}
    \right]
    &=\frac1{t^2}
      \sum_{\substack{i,j\in[t]\\i\ne j}}
      \Tr_{[t]}\left(
        F_{ai}F_{bj}(\Id_{ab}\otimes\omega_{\ell,t}(\rho))
      \right)\\
    &\quad+\frac1{t^2}\sum_{i=1}^t
      \Tr_{[t]}\left(
        (F_{ab}F_{bi}+F_{ab}F_{ai})
        (\Id_{ab}\otimes\omega_{\ell,t}(\rho))
      \right)\\
    &\quad+\frac{\ell}{t^2}
      \Tr_{[t]}\left(
        F_{ab}(\Id_{ab}\otimes\omega_{\ell,t}(\rho))
      \right).
  \end{align*}
  For \(t\ge2\), the three kinds of terms evaluate as follows:
  \begin{align*}
    \Tr_{[t]}\left(
      F_{ai}F_{bj}(\Id_{ab}\otimes\omega_{\ell,t}(\rho))
    \right)
      &=\Tr_{[t]\setminus\{1,2\}}\bigl(\omega_{\ell,t}(\rho)\bigr)
        \quad(i\ne j),\\
    \Tr_{[t]}\left(
      (F_{ab}F_{bi}+F_{ab}F_{ai})
      (\Id_{ab}\otimes\omega_{\ell,t}(\rho))
    \right)
      &=\begin{aligned}[t]
        \Bigl(&
          \Tr_{[t]\setminus\{1\}}\bigl(\omega_{\ell,t}(\rho)\bigr)
            \otimes\Id_d\\
        &+\Id_d\otimes
          \Tr_{[t]\setminus\{1\}}\bigl(\omega_{\ell,t}(\rho)\bigr)
        \Bigr)F_{ab},
      \end{aligned}\\
    \Tr_{[t]}\left(
      F_{ab}(\Id_{ab}\otimes\omega_{\ell,t}(\rho))
    \right)
      &=\Tr\bigl(\omega_{\ell,t}(\rho)\bigr)F_{ab}
       =\Prb_\rho[J=\ell]F_{ab}.
  \end{align*}
  The first sum contains \(t(t-1)\) ordered pairs \(i\ne j\), giving the
  coefficient \(t(t-1)/t^2=(t-1)/t\).  When \(t=1\), this coefficient is
  zero and the sum is empty, so the first term is omitted.  The remaining
  terms are evaluated in the same way, so
  \cref{eq:fixed-index-second-moment} holds for all \(t\ge1\).
\end{proof}

\begin{proof}[Proof of \Cref{prop:gaussian-block-estimator}]
  By \Cref{lem:gaussian-block-povm-normalization}, the densities in
  \cref{eq:gaussian-block-povm} form a POVM.  Fix a state \(\rho\).  We first
  sum the contributions from
  \Cref{lem:fixed-index-moments} over \(\ell\).  The projectors
  \(\Delta_{\ell,t}\) sum to the identity, so
  \[
    \sum_{\ell=1}^{\ell_{\max}}\omega_{\ell,t}(\rho)
    =\rho^{\otimes t}.
  \]
  Taking one- and two-factor marginals of this identity gives
  \[
    \sum_{\ell=1}^{\ell_{\max}}
      \Tr_{[t]\setminus\{1\}}\bigl(\omega_{\ell,t}(\rho)\bigr)
      =\rho,
    \qquad
    \sum_{\ell=1}^{\ell_{\max}}
      \Tr_{[t]\setminus\{1,2\}}\bigl(\omega_{\ell,t}(\rho)\bigr)
      =\rho^{\otimes2}\quad\text{when \(t\ge2\)}.
  \]
  Summing \cref{eq:fixed-index-first-moment,eq:fixed-index-second-moment}
  over \(\ell\) gives
  \begin{align*}
    \E_\rho[Y]
    &=\sum_{\ell=1}^{\ell_{\max}}
      \Tr_{[t]\setminus\{1\}}\bigl(\omega_{\ell,t}(\rho)\bigr)
    =\rho,\\
    \E_\rho[Y\otimes Y]
    &=\frac{t-1}{t}\rho^{\otimes2}
      +\frac1t(\rho\otimes\Id_d+\Id_d\otimes\rho)F
      +\frac1{t^2}
        \left(
          \sum_{\ell=1}^{\ell_{\max}}
            \ell\Tr(\omega_{\ell,t}(\rho))
        \right)F.
  \end{align*}
  By \cref{eq:support-index-distribution},
  \[
    \sum_{\ell=1}^{\ell_{\max}}
      \ell\Tr(\omega_{\ell,t}(\rho))
    =\sum_{\ell=1}^{\ell_{\max}}\ell\Prb_\rho[J=\ell]
    =\E_\rho[J].
  \]
Combining the last two equations, we get \cref{eq:gaussian-block-second-moment}.

  It remains to prove the assertion about \(\Pi^\perp G\).  Let \(\Pi\) be
  any orthogonal projector satisfying \(\Pi\rho=\rho\).
  For a fixed \(\ell\) with \(\Prb_\rho[J=\ell]>0\), we compare the Gaussian
  reference law \(\gamma_{d,\ell}\) with the outcome law of \(G\) conditional
  on \(J=\ell\).  We show that the density of this outcome law relative to
  \(\gamma_{d,\ell}\) depends only on \(\Pi G\).

  By
  \Cref{lem:gaussian-moment-support,lem:gaussian-moment-identity}, both
  \(\Delta_{\ell,t}\) and \((\Gamma_{\ell,t}^+)^{1/2}\) commute with
  \(\Pi^{\otimes t}\).  Therefore,
  \[
    \rho^{\otimes t}
      =\Pi^{\otimes t}\rho^{\otimes t}\Pi^{\otimes t}
    \quad\Longrightarrow\quad
    \omega_{\ell,t}(\rho)
      =\Pi^{\otimes t}\omega_{\ell,t}(\rho)\Pi^{\otimes t}
    \quad\Longrightarrow\quad
    \Theta_{\ell,\rho}
    =\Pi^{\otimes t}\Theta_{\ell,\rho}\Pi^{\otimes t}.
  \]
  The density in \cref{eq:fixed-index-outcome-law} satisfies
  \[
    p_{\ell,\rho}(G)
    =\Tr\left(
      \Theta_{\ell,\rho}(\Pi GG^\dagger\Pi)^{\otimes t}
    \right)
    =p_{\ell,\rho}(\Pi G).
  \]
  Put \(G_\parallel\coloneqq\Pi G\) and
  \(G_\perp\coloneqq\Pi^\perp G\).

  Let \(\gamma_{\Pi,\ell}\) and \(\gamma_{\Pi^\perp,\ell}\) denote the
  standard complex Gaussian laws on matrices whose columns lie in
  \(\operatorname{Im}(\Pi)\) and \(\operatorname{Im}(\Pi^\perp)\),
  respectively.  Under \(\gamma_{d,\ell}\), the two components are independent,
  and their joint law factors as
  \[
    \dd\gamma_{d,\ell}(G)
    =\dd\gamma_{\Pi,\ell}(G_\parallel)
      \dd\gamma_{\Pi^\perp,\ell}(G_\perp).
  \]
  Therefore, the conditional distribution of
  \((G_\parallel,G_\perp)\) given \(J=\ell\) is
  \[
    \frac{p_{\ell,\rho}(G_\parallel)}
      {\Prb_\rho[J=\ell]}
    \,\dd\gamma_{\Pi,\ell}(G_\parallel)
    \,\dd\gamma_{\Pi^\perp,\ell}(G_\perp).
  \]
  The factor involving \(p_{\ell,\rho}\) depends only on \(G_\parallel\).
  Thus, under the outcome law conditional on \(J=\ell\),
  \(G_\perp=\Pi^\perp G\) is independent of \(G_\parallel=\Pi G\) and has
  the same law \(\gamma_{\Pi^\perp,\ell}\) as under the Gaussian reference
  measure.  Equivalently, its columns
  are independent standard complex Gaussian vectors in
  \(\operatorname{Im}(\Pi^\perp)\).  This completes the proof.
\end{proof}

\subsection{Analysis of the estimator}
\label{sec:upper-bound-aggregation}

With \Cref{prop:gaussian-block-estimator} established, we now analyze the
estimator defined in
\cref{eq:upper-bound-protocol-average,eq:upper-bound-protocol-output}.  Choose
an \(r\)-dimensional subspace containing the support of the unknown state
\(\rho\), and
let \(\Pi\) be the orthogonal projector onto this subspace.  We write the error
\(E\coloneqq\overline Y-\rho\) in the four blocks determined by the direct sum
decomposition \(\C^d=\operatorname{Im}(\Pi)\oplus\operatorname{Im}(\Pi^\perp)\).
The second moment identity in \cref{eq:gaussian-block-second-moment} helps bound
the Frobenius norm of the block supported on $\operatorname{Im}(\Pi)$ and the off-diagonal blocks.
The conditional law of
\(\Pi^\perp G\) controls the block \(\Pi^\perp E\Pi^\perp\) in operator
norm.  The projection bound in the next subsection combines these estimates.
The protocol does not need to know \(\Pi\); this decomposition is used only
in the analysis.

\subsubsection{Projection onto low-rank states}

\begin{lemma}[Rank-constrained Frobenius projection]
  \label{lem:rank-constrained-frobenius-projection}
  Let \(\rho\in\cD_r(\C^d)\), let \(M\in\C^{d\times d}\) be Hermitian, and put
  \(E\coloneqq M-\rho\).  Let
  \[
    \widehat\rho\in
    \underset{\sigma\in\cD_r(\C^d)}{\arg\min}\,
      \norm{M-\sigma}_\F.
  \]
  If \(\Pi\) is a rank-\(r\) orthogonal projector satisfying
  \(\Pi\rho=\rho\), then
  \begin{equation}
    \norm{\widehat\rho-\rho}_1
    \le 2\sqrt{2r}\left(
      \norm{\Pi E\Pi}_\F^2
      +2\norm{\Pi^\perp E\Pi}_\F^2
      +2r\norm{\Pi^\perp E\Pi^\perp}_{\mathrm{op}}^2
    \right)^{1/2}.
    \label{eq:rank-constrained-projection-bound}
  \end{equation}
\end{lemma}

\begin{proof}
  Put \(D\coloneqq\widehat\rho-\rho\).  Since \(\rho\) is feasible in
  the minimization that defines \(\widehat\rho\),
  \[
    \norm{E-D}_\F^2\le\norm E_\F^2.
  \]
  Expanding the left-hand side gives
  \begin{equation}
    \norm D_\F^2\le2\Tr(ED).
    \label{eq:rank-projection-optimality}
  \end{equation}
  Decomposing the inner product according to
  \(\Id_d=\Pi+\Pi^\perp\), and using the property that \(E\) and \(D\) are
  Hermitian, gives
  \begin{align}
    \Tr(ED)
    &=\Tr(\Pi E\Pi D\Pi)
      +2\operatorname{Re}\left(
        \Tr(\Pi E\Pi^\perp D\Pi)
      \right)
      +\Tr(\Pi^\perp E\Pi^\perp D\Pi^\perp) \notag\\
    &\le
      \norm{\Pi E\Pi}_\F\norm{\Pi D\Pi}_\F
      +2\norm{\Pi^\perp E\Pi}_\F
        \norm{\Pi^\perp D\Pi}_\F
      +\norm{\Pi^\perp E\Pi^\perp}_{\mathrm{op}}
        \norm{\Pi^\perp D\Pi^\perp}_1.
    \label{eq:rank-projection-block-trace-bound}
  \end{align}
  Both \(\rho\) and \(\widehat\rho\) have rank at most \(r\), so
  \begin{equation}
    \rank(D)\le\rank(\rho)+\rank(\widehat\rho)\le2r,
    \qquad
    \norm{\Pi^\perp D\Pi^\perp}_1
    \le\sqrt{2r}\norm{\Pi^\perp D\Pi^\perp}_\F.
    \label{eq:rank-projection-block-trace-norm}
  \end{equation}
  The four blocks are orthogonal in the Frobenius (Hilbert--Schmidt) inner
  product, so
  \[
    \norm D_\F^2
    =\norm{\Pi D\Pi}_\F^2
      +2\norm{\Pi^\perp D\Pi}_\F^2
      +\norm{\Pi^\perp D\Pi^\perp}_\F^2.
  \]
  Using \cref{eq:rank-projection-block-trace-norm} in the right-hand side of
  \cref{eq:rank-projection-block-trace-bound} and applying Cauchy--Schwarz
  to the three products gives
  \begin{align*}
    \Tr(ED)
    &\le\left(
      \norm{\Pi E\Pi}_\F^2
      +2\norm{\Pi^\perp E\Pi}_\F^2
      +2r\norm{\Pi^\perp E\Pi^\perp}_{\mathrm{op}}^2
    \right)^{1/2}\\
    &\qquad\cdot\left(
      \norm{\Pi D\Pi}_\F^2
      +2\norm{\Pi^\perp D\Pi}_\F^2
      +\norm{\Pi^\perp D\Pi^\perp}_\F^2
    \right)^{1/2}\\
    &=\left(
      \norm{\Pi E\Pi}_\F^2
      +2\norm{\Pi^\perp E\Pi}_\F^2
      +2r\norm{\Pi^\perp E\Pi^\perp}_{\mathrm{op}}^2
    \right)^{1/2}\norm D_\F.
  \end{align*}
  If \(D=0\), the result is immediate.  Otherwise, combining this
  estimate with \cref{eq:rank-projection-optimality} and dividing by
  \(\norm D_\F\) gives
  \[
    \norm D_\F
    \le2\left(
      \norm{\Pi E\Pi}_\F^2
      +2\norm{\Pi^\perp E\Pi}_\F^2
      +2r\norm{\Pi^\perp E\Pi^\perp}_{\mathrm{op}}^2
    \right)^{1/2}.
  \]
  Finally, \(\norm D_1\le\sqrt{2r}\norm D_\F\), which proves
  \cref{eq:rank-constrained-projection-bound}.
\end{proof}

\subsubsection{Bounds for the four blocks}

We next bound the norms of the blocks of~$E$ occurring in \cref{eq:rank-constrained-projection-bound}. 

Recall the conditional law of \(\Pi^\perp G\) from \Cref{prop:gaussian-block-estimator}. Namely, the columns of \(\Pi^\perp G\) remain independent standard complex Gaussian vectors in \(\operatorname{Im}(\Pi^\perp)\) after conditioning on \(J\). Building on this, we analyze the operator norm of~$\Pi^\perp E\Pi^\perp$ using the large deviation bound for Gaussian matrices stated below. 

For a standard complex Gaussian matrix
\(Z\in\C^{m\times k}\), we have \(\E[ZZ^\dagger]=k\Id_m\), so
\(ZZ^\dagger-k\Id_m\) in \cref{eq:gaussian-covariance-operator-norm} is the
deviation of \(ZZ^\dagger\) from its expectation. The bound on the expected deviation may be derived as in the proof of Theorem~4.6.1 and Remark~4.6.2 (in Sec.~4.6) of the text by Vershynin~\cite{Vershynin26-high-dim-probability}. (The results in the text are stated for real subgaussian matrices, but can be adapted to the complex case in a straightforward manner.)

\begin{lemma}
  \label{lem:gaussian-covariance-operator-norm}
  Let \(m,k\) be positive integers, and let
  \(Z\in\C^{m\times k}\) be a standard complex Gaussian matrix.  Then
  \begin{equation}
    \E\norm{ZZ^\dagger-k\Id_m}_{\mathrm{op}}^2
    =O(mk+m^2).
    \label{eq:gaussian-covariance-operator-norm}
  \end{equation}
\end{lemma}

The norms of the other blocks of~$E$ may be bounded using the second moment of the estimator in \cref{eq:gaussian-block-second-moment}. 

\begin{lemma}[Error in the four blocks]
  \label{lem:gaussian-estimator-support-blocks}
  Consider the measurement and averaging step of
  \Cref{sec:gaussian-block-estimator-target} with arbitrary positive integers
  \(B,s\), and let \(\rho\in\cD_r(\C^d)\).  Let
  \(\Pi\) be a rank-\(r\) orthogonal projector satisfying
  \(\Pi\rho=\rho\).  Put \(m\coloneqq d-r\), and apply the joint measurement
  \(\mathsf M_s\) independently \(B\) times, each time on \(s\) fresh samples
  with joint state \(\rho^{\otimes s}\).  For \(b\in[B]\), write
  \((J_b,G_b)\) for the \(b\)-th outcome and
  \(Y_b\coloneqq(G_bG_b^\dagger-J_b\Id_d)/s\).  Set
  \[
    \overline Y\coloneqq\frac1B\sum_{b=1}^B Y_b,
    \qquad E\coloneqq\overline Y-\rho,
    \qquad N\coloneqq Bs.
  \]
  We have
  \begin{align}
    \E_\rho \left[ \norm{\Pi E\Pi}_\F^2 + 2 \norm{\Pi^\perp E\Pi}_\F^2 \right]
    &\le\frac{2d}{N}+\frac{ 2dr \E_\rho[J_1]}{Ns} \enspace, \qquad \text{and}
    \label{eq:support-and-cross-block-risk}\\
    \E_\rho\norm{\Pi^\perp E\Pi^\perp}_{\mathrm{op}}^2
    &=O\left(\frac{m\E_\rho[J_1]}{Ns}+\frac{m^2}{N^2}\right).
    \label{eq:upper-complement-block-risk}
  \end{align}
\end{lemma}

\begin{proof}
For any orthogonal projection operator~$Q$ on~$\complex^d$, we have
\[
\norm{Q E Q}_\frobenius^2 = \Tr( QEQEQ) = \Tr( (QE)^2 )
   = \Tr(F (QE \tensor QE)) = \Tr( F (Q \tensor Q)(E \tensor E)),
\]
where~$F$ is the swap operator on~$\complex^d \tensor \complex^d$.
So
\begin{equation}
\label{eq-frobenius-norm-via-moment}
\E_\rho \norm{Q E Q}_\frobenius^2 = \Tr( F (Q \tensor Q) \expct_\rho (E \tensor E)).
\end{equation}

Using the independence of~$(Y_b)$, the property that~$Y_b$ is unbiased, and the second moment identity in \cref{eq:gaussian-block-second-moment}, we have
\begin{align}
\nonumber
\expct_\rho (E \tensor E) & = \expct_\rho \left[ (\overline{Y} - \rho) \tensor (\overline{Y} - \rho) \right] \\
\nonumber
  & = \frac{1}{B^2} \sum_{b_1, b_2 = 1}^B \expct_\rho (Y_{b_1} \tensor Y_{b_2}) -  \frac{1}{B} \sum_{b = 1}^B  \expct_\rho \left[ (Y_{b} \tensor \rho) + (\rho \tensor Y_{b}) \right] + \rho \tensor \rho \\
  \nonumber
  & = \frac{1}{B}  \left[ \expct_\rho (Y_1 \tensor Y_1) - \rho \tensor \rho \right] \\
\label{eq-deviation-second-moment}
  & = - \frac{1}{N} \rho \tensor \rho + \frac{1}{N} \left[ (\rho \tensor \id) + (\id \tensor \rho) \right] F + \frac{ \expct_\rho J_1 }{Ns} F.
\end{align}

Taking~$Q \eqdef \id$ in \cref{eq-frobenius-norm-via-moment}, we get
\begin{equation}
\label{eq-deviation-frobenius-norm}
\E_\rho \norm{E}_\frobenius^2 = - \frac{1}{N} \Tr(\rho^2) + \frac{2d}{N} + \frac{ d^2 \expct_\rho J_1 }{Ns}.
\end{equation}
Taking~$Q \eqdef \Pi^\perp$ instead, we get
\begin{equation}
\label{eq-deviation-complement-block-frobenius-norm}
\E_\rho \norm{\Pi^\perp E \Pi^\perp}_\frobenius^2 = \frac{\expct_\rho J_1}{Ns} \Tr\left( F ( \Pi^\perp \tensor \Pi^\perp) F \right) = \frac{m^2 \expct_\rho J_1}{Ns}.
\end{equation}
Subtracting \cref{eq-deviation-complement-block-frobenius-norm} from \cref{eq-deviation-frobenius-norm} and using \(d^2-m^2=2dr-r^2\le2dr\), we get the bound in \cref{eq:support-and-cross-block-risk}.

  It remains to control the block \(\Pi^\perp E\Pi^\perp\).  If \(m=0\),
  this block vanishes, so its bound is immediate.  Assume \(m\ge1\), and let
  \(G_{\perp,b}\in\C^{m\times J_b}\) be the coordinate matrix of
  \(\Pi^\perp G_b\) in a fixed orthonormal basis of
  \(\operatorname{Im}(\Pi^\perp)\).  By
  \Cref{prop:gaussian-block-estimator} and independence of the \(B\)
  measurements, conditional on \(J_1,\ldots,J_B\), the matrices
  \(G_{\perp,1},\ldots,G_{\perp,B}\) are independent standard complex
  Gaussian matrices.  Set
  \[
    K\coloneqq\sum_{b=1}^B J_b,
    \qquad
    Z\coloneqq
      \begin{bmatrix}
        G_{\perp,1}&\cdots&G_{\perp,B}
      \end{bmatrix}
      \in\C^{m\times K}.
  \]
  Then
  \[
    Z\mid(J_1,\ldots,J_B)
      \text{ is a standard complex Gaussian matrix},
    \qquad
    ZZ^\dagger=\sum_{b=1}^B
      G_{\perp,b}G_{\perp,b}^\dagger.
  \]
  Since \(\Pi^\perp\rho\Pi^\perp=0\) and \(N=Bs\), in the chosen basis of
  \(\operatorname{Im}(\Pi^\perp)\),
  \[
    \Pi^\perp E\Pi^\perp
    =\frac1{Bs}\sum_{b=1}^B
      \left(G_{\perp,b}G_{\perp,b}^\dagger-J_b\Id_m\right)
    =\frac1N(ZZ^\dagger-K\Id_m).
  \]
  By \Cref{lem:gaussian-covariance-operator-norm},
  \[
    \E_\rho\left[
      \norm{\Pi^\perp E\Pi^\perp}_{\mathrm{op}}^2
      \mathrel{\big|}J_1,\ldots,J_B
    \right]
    =O\left(\frac{mK+m^2}{N^2}\right).
  \]
  Taking expectation over \(J_1,\ldots,J_B\), and using
  \(\E_\rho[K]=B\E_\rho[J_1]\) and \(N=Bs\), gives
  \[
    \E_\rho\norm{\Pi^\perp E\Pi^\perp}_{\mathrm{op}}^2
    =O\left(\frac{mB\E_\rho[J_1]+m^2}{N^2}\right)
    =O\left(\frac{m\E_\rho[J_1]}{Ns}+\frac{m^2}{N^2}\right),
  \]
  proving \cref{eq:upper-complement-block-risk}.
\end{proof}

\subsubsection{Completion of the upper bound}

\begin{proof}[Proof of \Cref{thm:bounded-block-upper-bound}]
  Consider the protocol in
  \cref{eq:upper-bound-protocol-parameters,eq:upper-bound-protocol-average,eq:upper-bound-protocol-output}.
  Fix \(\rho\in\cD_r(\C^d)\), choose a rank-\(r\) orthogonal projector \(\Pi\)
  satisfying \(\Pi\rho=\rho\), and let \(E\coloneqq\overline Y-\rho\).
  By \Cref{prop:gaussian-block-estimator},
  \(\E_\rho[J_1]=O(\sqrt s)\), and the choice in
  \cref{eq:upper-bound-protocol-parameters} gives \(s\le r^2\).
  Combining the bounds in
  \Cref{lem:gaussian-estimator-support-blocks} and using
  \(d/N\le dr/(N\sqrt s)\) gives
  \[
    \E_\rho\left[
      \norm{\Pi E\Pi}_\F^2
      +2\norm{\Pi^\perp E\Pi}_\F^2
      +2r\norm{\Pi^\perp E\Pi^\perp}_{\mathrm{op}}^2
    \right]
    =O \!\left(
      \frac{dr}{N\sqrt s}+\frac{rd^2}{N^2}
    \right).
  \]
  Squaring \cref{eq:rank-constrained-projection-bound} and taking
  expectation therefore gives
  \begin{equation}
    \E_\rho\norm{\widehat\rho-\rho}_1^2
    =O \! \left(
      \frac{dr^2}{N\sqrt s} + \frac{d^2r^2}{N^2}
    \right).
    \label{eq:rank-dependent-trace-risk}
  \end{equation}
  This bound explains the choice \(s=\min\{t,r^2\}\).  For fixed \(N\),
  increasing \(s\) up to \(r^2\) decreases the first term until it reaches
  order \(dr/N\).  At that point, the bound gives sample complexity
  \(O(dr/\eps^2)\), matching the optimal rate for unrestricted joint
  measurements.  Thus using larger joint measurements cannot improve the
  asymptotic sample complexity.

  By \cref{eq:upper-bound-protocol-parameters},
  \[
    \frac{dr^2}{N\sqrt s}
    \le\frac{\eps^2}{C_0},
    \qquad
    \frac{d^2r^2}{N^2}
    \le\frac{\eps^4}{C_0^2}
    \le\frac{\eps^2}{C_0^2}.
  \]
  Thus, for sufficiently large \(C_0\),
  \cref{eq:rank-dependent-trace-risk} and Markov's inequality give
  \[
    \Prb_\rho\left[
      \norm{\widehat\rho-\rho}_1>\eps
    \right]\le\frac13.
  \]

  By the definition of \(B\) in \cref{eq:upper-bound-protocol-parameters},
  \[
    N \le C_0\frac{dr^2}{\eps^2 \sqrt s} + s.
  \]
  Since \(s\le r^2\le dr/\eps^2\) and
  \(r/\sqrt s=\max\{1,r/\sqrt t\}\), this gives
  \[
    N=O \! \left(
      \frac{dr}{\eps^2}
      \max\left\{1,\frac r{\sqrt t}\right\}
    \right).
  \]
  This proves \Cref{thm:bounded-block-upper-bound}.
\end{proof}

\section{Discussion}
\label{sec:discussion}

The matching upper and lower bounds show the role rank plays in the performance of
joint measurements in algorithms for state tomography.  Relative to single-sample tomography, jointly measuring
\(t\) samples improves the optimal sample complexity by a factor of order
\(\sqrt t\) until \(t\) reaches order \(r^2\).  At that point, the sample
complexity matches that of unrestricted collective measurements.  The upper
bound is achieved by a nonadaptive protocol, so classical adaptivity offers
no further improvement beyond constant factors.

Two ingredients of the lower bound may be useful for future work on quantum
estimation.  The first is our bound on the Fisher information trace for the
support rotation family, which holds for every joint measurement on
\(t\) samples.  The second is the use of the conditional score chain rule to
extend such bounds to adaptive protocols.  This allows the analysis of
adaptivity to build on a bound for each joint measurement.

The upper bound controls the error in the supported and off-diagonal blocks of the estimator
in Frobenius norm, and the error in the block on the orthogonal complement of
the support of the unknown state in operator norm.  The first two bounds follow from the
estimator's second moment identity; the third uses the conditional Gaussian
law outside the state's
support.  The expected number of columns of the matrix outcome \(G\) appears
in both the Frobenius and operator norm error bounds.  The rank-constrained
projection bound combines these estimates to
give the final trace norm guarantee.  This approach may be useful in other variants of tomography.

In retrospect, the Gaussian protocol appears related to several previous mixed state
tomography algorithms.  Clarifying the precise relationship
between these protocols is a direction for future work.

Our protocol uses joint measurements with continuous outcomes.  A natural
direction is to construct explicit versions with finitely many outcomes that
preserve the optimal sample complexity, and to determine how many outcomes
are needed.

Our results concern protocols that retain only classical information between
measurement rounds.  With persistent quantum memory, different rounds can
combine into a larger coherent measurement, so the largest number of fresh
samples measured in one round no longer captures the available quantum
resource.  Understanding tradeoffs that simultaneously constrain the number
of samples measured in each round, quantum memory, sample complexity, and
computational efficiency remains a natural direction.

\section{AI disclosure}
We used GPT-5.5 and GPT-5.6 Sol for mathematical assistance in this work.
Their main mathematical contributions concerned two parts of the proofs.
For the lower bound, the models provided substantial assistance with the
proofs of
\Cref{lem:raising-map-occupation-reduction,lem:occupation-number-envelope}.
These lemmas are crucial to our Fisher information bound for joint
measurements on \(t>1\) samples and to the resulting dependence of the
sample complexity lower bound on \(t\).

For the upper bound, we used the models to reformulate earlier tomography
algorithms and their analyses
without using representation theory.
Our key additional contribution is a rank-dependent error analysis of the
resulting Gaussian formulation, presented in \Cref{sec:upper-bound-proof}.

We also used these models for discussions of general mathematical questions,
literature searches, and editorial revisions.
The authors independently verified all mathematical arguments and references
and take full responsibility for the content of the paper.

\section{Acknowledgements}
The authors thank Rain Ziming Yang and Jack Spalding-Jamieson for helpful
discussions.
Ashwin Nayak's research is supported in part by NSERC grants
RGPIN-2023-03731 and ALLRP-578455-2022.
Xingyu Zhou is supported by NSERC Grant RGPIN-2024-06493.

\bibliographystyle{alpha}
\bibliography{references}

\appendix
\section{Statistical formalism}
\label{app:statistical-formalism}

This appendix collects the standard statistical formalism used in the main
proof.  We first justify the density notation for continuous-outcome POVMs
and adaptive transcripts on the general outcome spaces allowed by our model.
We then verify the likelihood regularity used by the Fisher information
argument and give the van Trees proof.  In particular, all reference measures below
are independent of the unknown state, so parameter derivatives and Fisher
information are computed within one fixed dominated model.

\subsection{Dominated models and finite-dimensional POVMs}
\label{app:single-povm-domination}

Let \(\{\Prb_{\vct{\theta}}:\vct{\theta}\in\Theta\}\) be probability
measures on \((\mathsf Z,\mathcal Z)\).  A measure \(\nu\) dominates this
family if \(\Prb_{\vct{\theta}}\ll\nu\) for every parameter, meaning that
for every \(\mathsf E\in\mathcal Z\),
\[
  \nu(\mathsf E)=0
  \quad\Longrightarrow\quad
  \Prb_{\vct{\theta}}(\mathsf E)=0.
\]
We assume \(\nu\) is \(\sigma\)-finite, meaning that \(\mathsf Z\) is a
countable union of measurable sets of finite \(\nu\)-measure.  This holds
for all the reference probability measures constructed below.
The Radon--Nikodym theorem then gives likelihood densities
\[
  q_{\vct{\theta}}(z)
  \coloneqq
  \frac{\dd{\Prb_{\vct{\theta}}}}{\dd\nu}(z),
\]
which satisfy \cref{eq:dominated-model}.  An event \(N\in\mathcal Z\) is
\(\nu\)-null if \(\nu(N)=0\), and a property holds \(\nu\)-almost
everywhere if it fails only on such an event.  Because the model is
dominated, every \(\nu\)-null event also has probability zero under every
\(\Prb_{\vct{\theta}}\).

For a POVM \(\mathsf M\) on \(\C^d\), the measure
\(\nu\) in \cref{eq:povm-dominating-measure} is a probability measure:
\[
  \nu(\mathsf Z)
  =d^{-1}\Tr\bigl(\mathsf M(\mathsf Z)\bigr)
  =d^{-1}\Tr(\Id_d)
  =1.
\]
It dominates every state-induced outcome law.  Indeed, if \(\nu(\mathsf E)=0\),
then the positive semidefinite operator \(\mathsf M(\mathsf E)\) has trace zero
and is therefore the zero operator.  Hence
\(\Tr(\rho\mathsf M(\mathsf E))=0\) for every state \(\rho\).
Each scalar measure of a matrix entry of \(\mathsf M\) is consequently absolutely
continuous with respect to \(\nu\).  Applying the scalar Radon--Nikodym
theorem entrywise gives a measurable operator-valued function
\(z\mapsto M_z\) satisfying \cref{eq:povm-density}.  Positivity holds almost
everywhere because, for a countable dense set of vectors \(\vct{u}\), the
scalar measure
\(\mathsf E\mapsto\vct{u}^\dagger\mathsf M(\mathsf E)\vct{u}\) is positive.
Taking the trace in \cref{eq:povm-density} gives
\[
  \int_{\mathsf E}\Tr(M_z)\,\dd\nu(z)
  =\Tr\bigl(\mathsf M(\mathsf E)\bigr)
  =d\nu(\mathsf E)
\]
for every measurable event \(\mathsf E\), so \(\Tr(M_z)=d\)
\(\nu\)-almost everywhere.  This is the standard trace-dominated
POVM representation used in quantum statistics
\cite{BarndorffNielsenGill2000}.

Finally, if a nonnegative likelihood is differentiable at an interior
parameter value where it vanishes, its gradient vanishes there.  Thus the
score and the Fisher integrands may be set to zero on a zero-likelihood set,
as done in \Cref{sec:classical-fisher-information}.

\subsection{Adaptive transcript domination}
\label{app:continuous-adaptive-likelihoods}

Fix a protocol with pathwise sample bound \(N\), and let
\((\mathsf U,\Sigma_{\mathsf U},\kappa)\) be the probability space of its
private random seed.  We index the protocol by \(N\) rounds, padding it with zero-sample
rounds after it halts.  Write
\[
  \mathsf H_i
  \coloneqq
  \mathsf U\times\mathsf Z_1\times\cdots\times\mathsf Z_{i-1}
\]
for the history space before round \(i\).  At \(h\in\mathsf H_i\), let
\(t_i(h)\in\{0,1,\ldots,t\}\) be the number of samples measured in round \(i\), and let
\(\mathsf M_i(\cdot\mid h)\) be the selected POVM on
\((\C^d)^{\otimes t_i(h)}\).  The protocol definition requires \(t_i\) and
the matrix entries of \(\mathsf M_i(\mathsf E\mid h)\), on each set of histories
where \(t_i(h)\) is fixed, to be measurable in \(h\) for each
\(\mathsf E\in\mathcal Z_i\).  For tuples,
write \(z_{<i}\coloneqq(z_1,\ldots,z_{i-1})\).

\begin{lemma}[Measurable domination of adaptive measurements]
  Define the probability kernel
  \begin{equation}
    \nu_i(\mathsf E\mid h)
    \coloneqq
    d^{-t_i(h)}\Tr\bigl(\mathsf M_i(\mathsf E\mid h)\bigr).
    \label{eq:block-reference-kernel}
  \end{equation}
  For each \(t'\in\{0,\ldots,t\}\), on the history slice
  \(t_i(h)=t'\), there is a jointly measurable
  positive semidefinite operator density \(M_i(z\mid h)\) on
  \((\C^d)^{\otimes t'}\) such that
  \begin{equation}
    \mathsf M_i(\mathsf E\mid h)
    =\int_{\mathsf E} M_i(z\mid h)\,\nu_i(\dd z\mid h),
    \qquad
    \Tr(M_i(z\mid h))=d^{t'},
  \end{equation}
  for every such history \(h\), every \(\mathsf E\in\mathcal Z_i\), and
  every \(z\in\mathsf Z_i\).
  Consequently, the parameter-independent measure
  \begin{equation}
    \lambda_N(\dd u\,\dd z_1\cdots\dd z_N)
    \coloneqq
    \kappa(\dd u)\prod_{i=1}^N
      \nu_i(\dd z_i\mid u,z_{<i})
  \end{equation}
  dominates the transcript law for every input state.
\end{lemma}

\begin{proof}
  The expression in \cref{eq:block-reference-kernel} has total mass one.
  If it assigns zero mass to \(\mathsf E\), then the positive semidefinite operator
  \(\mathsf M_i(\mathsf E\mid h)\) has trace zero and is itself zero.  Thus every
  scalar measure of a matrix entry of the POVM kernel is absolutely continuous
  with respect to \(\nu_i(\cdot\mid h)\).

  It remains to choose the Radon--Nikodym derivatives jointly measurably in
  the history and outcome.  Restrict to histories with \(t_i(h)=t'\), choose a
  basis of \((\C^d)^{\otimes t'}\), and let \(\mu_{ab}(\mathsf E\mid h)\) be the
  scalar measure associated with the \((a,b)\)-th entry of
  \(\mathsf M_i(\mathsf E\mid h)\).  Because \(\mathcal Z_i\) is countably
  generated, there is an increasing sequence \((\mathcal P_n)_{n\ge1}\) of
  finite measurable partitions whose cells generate \(\mathcal Z_i\).
  Write \(\boldsymbol{1}_C\) for the indicator of a cell \(C\), equal to one
  on \(C\) and zero outside \(C\), and set
  \[
    f_{ab,n}(h,z)
    \coloneqq
    \sum_{C\in\mathcal P_n}
      \boldsymbol{1}_C(z)
      \frac{\mu_{ab}(C\mid h)}{\nu_i(C\mid h)},
    \qquad 0/0\coloneqq0.
  \]
  These step functions are jointly measurable in \((h,z)\).  For each fixed
  history \(h\), the functions \(f_{ab,n}(h,\cdot)\) are conditional
  expectations of the entrywise Radon--Nikodym derivative under
  \(\nu_i(\cdot\mid h)\), with respect to the finite partitions
  \(\mathcal P_n\).  These increasing partitions generate \(\mathcal Z_i\),
  so martingale convergence gives a limit that is finite almost everywhere
  and equal to the Radon--Nikodym derivative.  On the jointly measurable
  set where every real and imaginary part has a finite limit, assemble these
  limits entrywise into \(M_i(z\mid h)\); set \(M_i(z\mid h)=\Id_{d^{t'}}\)
  elsewhere.

  For every fixed history, positivity follows by intersecting the
  full-measure sets on which
  \(\vct{u}^\dagger M_i(z\mid h)\vct{u}\ge0\) over a countable dense set of
  vectors \(\vct{u}\).  The operator density identity and
  \cref{eq:block-reference-kernel} give
  \[
    \int_{\mathsf E}\Tr(M_i(z\mid h))\,\nu_i(\dd z\mid h)
    =\Tr\bigl(\mathsf M_i(\mathsf E\mid h)\bigr)
    =d^{t'}\nu_i(\mathsf E\mid h)
  \]
  for every measurable event \(\mathsf E\), so
  \(\Tr(M_i(z\mid h))=d^{t'}\) almost everywhere.  The set on which the
  entrywise limits fail to exist finitely, positivity fails, or this trace
  identity fails is jointly measurable, and each of its sections is
  \(\nu_i(\cdot\mid h)\)-null.  Replacing the density by the identity operator
  on this exceptional set leaves every kernel integral unchanged and makes
  the density finite and positive semidefinite with the required trace everywhere.

  Perform this construction on each of the finitely many measurable slices
  \(t_i(h)=t'\), including the scalar \(t'=0\) slice.  The Ionescu--Tulcea
  construction then gives \(\lambda_N\).  Iterating the conditional Born
  rule proves the asserted domination.
\end{proof}

For the support rotation family, the conditional likelihood density at
history \(h\) is
\begin{equation}
  p_{i,X}(z\mid h)
  \coloneqq
  \Tr\left(
    M_i(z\mid h)\rho_X^{\otimes t_i(h)}
  \right).
  \label{eq:conditional-block-likelihood}
\end{equation}
Iterating these conditional densities gives the transcript likelihood in
\cref{eq:full-transcript-likelihood} with respect to \(\lambda_N\).

\subsection{Likelihood regularity}

In this subsection, \(C^1\) means continuously differentiable: the
derivative exists and varies continuously in the norm under discussion.
For a measure \(\mu\), \(L_1(\mu)\) is the space of integrable scalar
functions, with norm
\[
  \norm{f}_{L_1(\mu)}\coloneqq\int\abs{f}\,\dd\mu.
\]
Thus this function space \(L_1\) norm is distinct from the matrix trace norm,
even though both use the index \(1\).

\begin{lemma}[Likelihood regularity]
  \label{lem:block-likelihood-regularity}
  Let \(m,r\) be positive integers, let
  \(\Theta\subseteq\C^{m\times r}\) be open when the matrix space is regarded
  as real, and let \(X\mapsto\rho_X\), for \(X\in\Theta\), be a family of
  states on \(\C^d\) that is \(C^1\) in trace norm.  For each
  history \(h\), the conditional likelihood in
  \cref{eq:conditional-block-likelihood} is pointwise \(C^1\) for
  \(\nu_i(\cdot\mid h)\)-almost every outcome and is \(C^1\) as a
  map into \(L_1(\nu_i(\cdot\mid h))\).  The full likelihood in
  \cref{eq:full-transcript-likelihood} has the corresponding properties
  \(\lambda_N\)-almost everywhere and in \(L_1(\lambda_N)\).
  All derivatives are jointly measurable; conditional derivatives integrate
  to zero under \(\nu_i(\cdot\mid h)\), and full likelihood
  derivatives integrate to zero under \(\lambda_N\).  For a direction
  \(H\in\C^{m\times r}\),
  \[
    \Ddir_Hp_{i,X}(z\mid h)
    =
    \Tr\left(
      M_i(z\mid h)\Ddir_H(\rho_X^{\otimes t_i(h)})
    \right).
  \]
\end{lemma}

\begin{proof}
  The construction in \Cref{app:continuous-adaptive-likelihoods} gives
  positive semidefinite densities with
  \[
    \norm{M_i(z\mid h)}_{\mathrm{op}}
    \le\Tr(M_i(z\mid h))=d^{t_i(h)}\le d^t.
  \]
  Fix a closed ball \(K\) contained in \(\Theta\).  Since
  \(X\mapsto\rho_X^{\otimes k}\) is \(C^1\) for each
  \(k\in\{0,\ldots,t\}\), compactness of \(K\) gives a finite constant
  \(C_K\) such that
  \[
    \norm{\Ddir_H(\rho_X^{\otimes k})}_1
    \le C_K\norm H_\F
    \qquad(X\in K,\ 0\le k\le t).
  \]
  For fixed \((h,z)\), differentiating the trace pairing in
  \cref{eq:conditional-block-likelihood} gives the stated derivative formula
  and pointwise \(C^1\) regularity.  Trace/operator norm duality then gives,
  uniformly in \(h,z\) and \(X\in K\),
  \[
    0\le p_{i,X}(z\mid h)\le d^t,
    \qquad
    \abs{\Ddir_Hp_{i,X}(z\mid h)}
    \le C_Kd^t\norm H_\F.
  \]

  The transcript likelihood \(q_{N,X}\) is a product of \(N\) such
  conditional densities, so it is pointwise \(C^1\), and the product rule gives
  \[
    0\le q_{N,X}\le d^{tN},
    \qquad
    \abs{\Ddir_Hq_{N,X}}\le NC_Kd^{tN}\norm H_\F
    \quad(X\in K).
  \]
  All derivatives are jointly measurable by their trace formulas and the
  finite product rule.  Since \(\nu_i(\cdot\mid h)\) and \(\lambda_N\) are
  probability measures, these local bounds are integrable.  Dominated
  convergence applied to coordinate difference quotients and to differences
  of coordinate derivatives proves \(C^1\) regularity in the corresponding
  \(L_1\) spaces.  It also permits differentiation under the integrals.
  The conditional and full likelihood derivatives therefore integrate to
  zero, since each likelihood has integral one.
\end{proof}

\subsection{The van Trees inequality}
\label{app:vector-van-trees}

For completeness, we prove the van Trees inequality stated in
\Cref{lem:vector-van-trees}.  This is the standard argument; see also
\cite{GillLevit1995}.  Let
\(q_{\vct{\theta}}\) be likelihood densities with respect to a
parameter-independent measure \(\nu\).  We assume that, for \(\nu\)-almost
every observation, the likelihood is pointwise \(C^1\) in the parameter,
and that its coordinate derivatives are jointly measurable, agree with the
\(L_1(\nu)\) derivatives, and integrate to zero.  These are precisely the
properties established for our transcript likelihoods in
\Cref{lem:block-likelihood-regularity}.

\begin{proof}[Proof of \Cref{lem:vector-van-trees}]
  If the denominator in \cref{eq:vector-van-trees} is infinite, the claim
  is immediate.  We may therefore assume it is finite.  The expected squared
  error is finite because \(T\) is bounded and \(\pi\) has compact support.
  Write
  \[
    f(\vct{\theta},y)\coloneqq
    \pi(\vct{\theta})q_{\vct{\theta}}(y)
  \]
  for the joint density of the parameter and observation.  All derivatives
  below are with respect to \(\vct{\theta}\).  At a zero of a differentiable
  nonnegative function, its gradient vanishes.  We therefore set the
  quotients involving \(q_{\vct{\theta}}\), \(\pi\), or \(f\) to zero
  wherever their denominators vanish.

  We first bound the integral of \(\norm{\nabla f}_2^2/f\).  The product
  rule gives
  \(\nabla f=q_{\vct{\theta}}\nabla\pi
    +\pi\nabla q_{\vct{\theta}}\), and hence
  \[
    \frac{\norm{\nabla f}_2^2}{f}
    =
    q_{\vct{\theta}}\frac{\norm{\nabla\pi}_2^2}{\pi}
    +\pi\frac{\norm{\nabla q_{\vct{\theta}}}_2^2}
      {q_{\vct{\theta}}}
    +2\ip{\nabla\pi}{\nabla q_{\vct{\theta}}}.
  \]
  This identity also holds on the zero set of \(f\) under the convention
  above.  Since \(\pi=\psi^2\), we have \(\nabla\pi=2\psi\nabla\psi\), so
  \[
    \int_\Theta\frac{\norm{\nabla\pi}_2^2}{\pi}\,\dd\vct{\theta}
    =4\int_{\{\vct{\theta}\in\Theta:\,\psi(\vct{\theta})\ne0\}}
      \norm{\nabla\psi}_2^2\,\dd\vct{\theta}
    \le4\int_\Theta\norm{\nabla\psi}_2^2\,\dd\vct{\theta}
    =I(\pi)<\infty.
  \]
  For each fixed \(\vct{\theta}\), the likelihood derivatives belong to
  \(L_1(\nu)\) and integrate to zero by the stated regularity assumptions.
  Thus the cross term is integrable in \(y\), and
  \[
    \int
      \ip{\nabla\pi(\vct{\theta})}{\nabla q_{\vct{\theta}}(y)}
      \,\dd\nu(y)
    =\ip{\nabla\pi(\vct{\theta})}
      {\int\nabla q_{\vct{\theta}}(y)\,\dd\nu(y)}
    =0.
  \]
  Integrating the expansion first over \(y\) and then over
  \(\vct{\theta}\), using
  \(\int q_{\vct{\theta}}(y)\,\dd\nu(y)=1\) and the definition of the Fisher
  information matrix, therefore gives
  \begin{align*}
    J\coloneqq\int_\Theta\!\int
      \frac{\norm{\nabla f}_2^2}{f}\,\dd\nu(y)\,\dd\vct{\theta}
    &=
      \int_\Theta\frac{\norm{\nabla\pi}_2^2}{\pi}\,\dd\vct{\theta}
      +\int_\Theta\pi(\vct{\theta})
        \left(\int
          \frac{\norm{\nabla q_{\vct{\theta}}(y)}_2^2}{q_{\vct{\theta}}(y)}
          \,\dd\nu(y)\right)\dd\vct{\theta}\\
    &=
      \int_\Theta\frac{\norm{\nabla\pi}_2^2}{\pi}\,\dd\vct{\theta}
      +\E_{\vct{\theta}}\Tr\bigl(\cI(\vct{\theta})\bigr)\\
    &\le I(\pi)+\E_{\vct{\theta}}\Tr\bigl(\cI(\vct{\theta})\bigr)<\infty.
  \end{align*}

  This bound also gives the absolute integrability needed for Fubini's
  theorem.  For each \(j\in[p]\), Cauchy--Schwarz gives
  \begin{align*}
    \int_\Theta\!\int
      \abs{(T_j(y)-\theta_j)\,\partial_jf(\vct{\theta},y)}
      \,\dd\nu(y)\,\dd\vct{\theta}
    &\le
      \left(\E(T_j-\theta_j)^2\right)^{1/2}
      \left(\int_\Theta\!\int
        \frac{(\partial_jf)^2}{f}\,\dd\nu(y)\,\dd\vct{\theta}
      \right)^{1/2}\\
    &\le\left(\E(T_j-\theta_j)^2\right)^{1/2}\sqrt J<\infty.
  \end{align*}

  We now integrate by parts in \(\theta_j\).  The boundary term vanishes
  because \(\pi\), and hence \(f\), is compactly supported inside
  \(\Theta\).  The estimator \(T_j(y)\) does not depend on
  \(\vct{\theta}\), so \(\partial_j(T_j(y)-\theta_j)=-1\).
  Fubini's theorem lets us put the parameter integral first, and
  integration by parts then gives
  \begin{align*}
    \int_\Theta\!\int
      (T_j(y)-\theta_j)\,\partial_jf(\vct{\theta},y)
      \,\dd\nu(y)\,\dd\vct{\theta}
    &=\int\!\int_\Theta
      (T_j(y)-\theta_j)\,\partial_jf(\vct{\theta},y)
      \,\dd\vct{\theta}\,\dd\nu(y)\\
    &=-\int\!\int_\Theta
      \partial_j(T_j(y)-\theta_j)f(\vct{\theta},y)
      \,\dd\vct{\theta}\,\dd\nu(y)\\
    &=\int\!\int_\Theta
      f(\vct{\theta},y)\,\dd\vct{\theta}\,\dd\nu(y)\\
    &=\int_\Theta\pi(\vct{\theta})
      \left(\int q_{\vct{\theta}}(y)\,\dd\nu(y)\right)
      \dd\vct{\theta}=1.
  \end{align*}
  Summing over \(j\in[p]\) yields
  \[
    p=\int_\Theta\!\int
      \ip{T(y)-\vct{\theta}}{\nabla f(\vct{\theta},y)}
      \,\dd\nu(y)\,\dd\vct{\theta}.
  \]

  Finally, insert \(\sqrt f\) into the identity for \(p\):
  \[
    p
    =\int_\Theta\!\int
      \ip{T(y)-\vct{\theta}}{\nabla f}
      \,\dd\nu(y)\,\dd\vct{\theta}
    =\int_\Theta\!\int
      \ip{(T(y)-\vct{\theta})\sqrt f}{\nabla f/\sqrt f}
      \,\dd\nu(y)\,\dd\vct{\theta}.
  \]
  Cauchy--Schwarz for these two vector-valued functions gives
  \begin{equation}
    \begin{aligned}
      p^2
      &\le
        \left(\int_\Theta\!\int
          \norm{T(y)-\vct{\theta}}_2^2 f
          \,\dd\nu(y)\,\dd\vct{\theta}\right)
        \left(\int_\Theta\!\int
          \frac{\norm{\nabla f}_2^2}{f}
          \,\dd\nu(y)\,\dd\vct{\theta}\right)\\
      &=\E\norm{T-\vct{\theta}}_2^2\,J.
    \end{aligned}
    \label{eq:van-trees-cauchy}
  \end{equation}
  Substituting the bound for \(J\) and
  rearranging proves
  \[
    \E\norm{T-\vct{\theta}}_2^2
    \ge
    \frac{p^2}{
      \E_{\vct{\theta}}\Tr\bigl(\cI(\vct{\theta})\bigr)+I(\pi)},
  \]
  as claimed.
\end{proof}

\section{Measurable projection onto low-rank states}
\label{app:measurable-projection}

The upper bound estimator chooses a state nearest in Frobenius norm to a
Hermitian matrix.  The following standard consequence of the measurable
maximum theorem \cite[Theorem~18.19]{AliprantisBorder2006} ensures that
ties can be resolved measurably.

\begin{lemma}[Measurable nearest state selection]
  \label{lem:measurable-low-rank-projection}
  For every \(1\le r\le d\), there is a Borel map from the Hermitian matrices
  on \(\C^d\) to \(\cD_r(\C^d)\) that assigns to each Hermitian matrix \(M\) a
  minimizer of \(\norm{M-\sigma}_\F\) over
  \(\sigma\in\cD_r(\C^d)\).
\end{lemma}

\begin{proof}
  Apply the measurable maximum theorem to the constant compact
  correspondence \(M\mapsto\cD_r(\C^d)\) and the continuous objective
  \((M,\sigma)\mapsto-\norm{M-\sigma}_\F\).  The theorem gives a Borel measurable
  selection of minimizers.
\end{proof}

\end{document}